\documentclass[11pt,a4paper,thmsb,ukenglish,fleqn]{article}

\usepackage[utf8]{inputenc}

\usepackage{setspace}
\usepackage[margin=1.25in]{geometry}
\usepackage{graphicx}
\graphicspath{ {./figures/} }
\usepackage{subcaption}
\usepackage{color}
\usepackage{comment}
\usepackage{lineno}
\usepackage[backref=page]{hyperref}
\usepackage{amsmath, amsfonts, amssymb,  amsthm}
\usepackage{amssymb}
\usepackage{bbm}
\usepackage{natbib}
\usepackage{csquotes}

\allowdisplaybreaks

\title{Statistical inference for L\'{e}vy-driven graph supOU processes: From  
short- to long-memory in high-dimensional time series}

\author{Shreya Mehta \\
Department of Mathematics, Imperial College London,\\ 180 Queen's Gate, London, SW7 2AZ, UK\\
 s.mehta20@imperial.ac.uk
\and Almut E.~D.~Veraart\\
Department of Mathematics, Imperial College London,\\ 180 Queen's Gate, London, SW7 2AZ, UK\\
a.veraart@imperial.ac.uk
}
\date{\today}

\begin{document}

\newtheorem{theorem}{Theorem}[section]
\newtheorem{definition}[theorem]{Definition}
\newtheorem{corollary}[theorem]{Corollary}
\newtheorem{lemma}[theorem]{Lemma}
\newtheorem{proposition}[theorem]{Proposition}
\newtheorem{example}[theorem]{Example}
\newtheorem{remark}[theorem]{Remark}
\newtheorem{assumption}{Assumption}
\newtheorem*{theorem*}{Theorem}

\maketitle
\begin{abstract}
This article introduces L\'{e}vy-driven graph supOU processes,  a parsimonious parametrisation for high-dimensional time series in which dependence between components is governed by a graph structure.  Specifically, the model bridges short- and long-range dependence within a single parametric family while accommodating a wide range of marginal distributions. 
We further develop a generalised method of moments estimator, establish its consistency and asymptotic normality, and assess its finite-sample performance through a simulation study. 
   Finally, we illustrate the practical relevance of our model and estimation method in an empirical study of wind capacity factors in a European electricity network context.
\end{abstract}
\emph{MSC classes:}
62M10, %Time series, auto-correlation, regression, etc.
62F99, %Parametric inference
60G10, %Stationary processes
62P99, %Applications
60F05, %Central limit and other weak theorems
60G57,  %Random measures
60E07. %Infinitely divisible distributions; stable distributions

%\tableofcontents
%\
\section{Introduction}
There has been growing interest in modelling high-dimensional time series via network time series models or, more generally, network stochastic processes. Such models describe (part of) the dependence structure of the multivariate time series via a graph or network, often imposing sparsity conditions that yield more parsimonious and tractable models than traditional multivariate approaches. Within this framework, \cite{Zhu2017} introduced the network vector autoregression (NAR), while \cite{KLNN2016} and \cite{KLNN2020} proposed the generalised network autoregressive (GNAR) process. Both model the dependence of each node's time series on its own past values and those of its neighbours (and possibly higher-order neighbours), and are shown to outperform vector autoregressive (VAR) models in applications, where the number of parameters grows rapidly with dimension. Complementary to these node-based approaches, edge-based models, including the GNAR-edge model, have been developed by \cite{JJQ} and \cite{MCMR}, while \cite{ZHU2019} proposed a network quantile autoregression framework. Together, these works reflect the breadth of modelling strategies that exploit network structure to capture complex dependence in high-dimensional time series.

Network stochastic processes were recently proposed in a continuous-time setting to handle possibly irregularly spaced and high-frequency data. An important class of stochastic processes in this context are the 
L\'{e}vy-driven Ornstein-Uhlenbeck (OU) processes, see \cite{BNS2001,Bar01}, which have numerous applications in economics, finance, volatility modelling, neuroscience, biology, ecology, epidemiology, physics, turbulence, meteorology, and social science. 
\cite{CV22, CV22b} proposed the L\'{e}vy-driven graph Ornstein-Uhlenbeck (graph OU) process, combining the flexibility of continuous-time models with the sparsity of graphical models. These processes exhibit an exponentially decaying autocorrelation structure.
 Higher-order graph continuous-time autoregressive (graph CAR(p)) models  were proposed in \cite{LPV2024} together with suitable statistical inference tools.

We note that the question of parameter estimation for  network processes is also broadly  related to recent work on drift estimation for high-dimensional OU processes and general diffusions, see 
\cite{CMP2020, CMP2024}.

A powerful framework for more flexible serial dependence structures, including the potential for long-range dependence, was first introduced by  \cite{BNS2001, Bar01}, in the form of superpositions of OU (supOU) processes. Intuitively, supOU processes can be understood as (possibly infinite) sums of independent OU processes, or equivalently as processes obtained by randomising the memory parameter of the OU process. Multivariate extensions of supOU processes were later studied by \cite{Bar11}, whose work forms the foundation for the class of models introduced in this article.

The main contributions of this article are as follows. We extend the graph OU model of \cite{CV22} to propose graph supOU processes, thereby accommodating short- and long-memory behaviour as well as more flexible serial dependence structures.
We then provide concrete parametric specifications for the class of graph supOU models and discuss their probabilistic properties. 
 We develop a GMM estimation methodology for graph supOU processes, establishing consistency in full generality and asymptotic normality in the short-memory setting, building on the weak-dependence theory of \cite{CS}.
The finite-sample performance of the estimator is confirmed through a simulation study.  Finally, we apply the model to wind capacity factors in a European electricity network,  where the data exhibit memory behaviour beyond the reach of classical OU processes.

To support reproducibility, all \texttt{R} code for the simulation study and empirical application is publicly available on GitHub\footnote{\url{https://github.com/almutveraart/grapsupOU-simulation-estimation-application}}  and archived on Zenodo\footnote{\url{https://doi.org/10.5281/zenodo.14857667}}.

This article is organised as follows. 
Section \ref{sec:GrsupOU} reviews multivariate supOU processes and defines graph supOU processes as a special case, establishing conditions for their existence and proposing parametric specifications suited to applications. 
Section \ref{sec:GMMestimator}   presents the estimation procedure for graph supOU processes. For the concrete parametric examples introduced earlier, we develop a two-step estimation procedure based on suitable eigenvalues of the scaled empirical autocovariance matrix. We note that this estimation procedure does not involve any optimisation in high-dimensional settings and is, hence, fast to perform even if the network  is large. For the general graph supOU case, we also describe a traditional GMM approach. 
We assess the finite sample performance of the new methodology in 
Section \ref{sec:sim}, and apply it to modelling wind capacity factors in an electricity network in Portugal in 
Section \ref{sec:emp}. 
Section \ref{sec:conclusion} concludes the article and outlines some potential future work.  Appendix \ref{app:GMM} presents the technical details of the GMM framework.

%%%%%%%%%%%%%%%%%%%%%%%%%%%%%%%%%%%%%
%%%%%%%%%%%%%%%%%%%%%%%%%%%%%%%%%%%%%%%%%%%%%%%%%%
%%%%%%%%%%%%%%%%%%%%%%%%%%%%%%%%%%%%%%%%%%%%%%%%%
\section{Graph supOU processes}\label{sec:GrsupOU}
\subsection{Notation} Throughout the article, we consider a filtered probability space $(\Omega, \mathcal{F}, \mathbb{P})$ endowed with a filtration $(\mathcal{F}_t, t \in \mathbb{R})$. 
We define various sets and operations related to matrices as follows. Let $M_{d,k}(\mathbb{R})$ represent the set of real $d\times k$ matrices. When $k=d$, we denote this set as $M_{d}(\mathbb{R})$. The  group of invertible $d\times d$-matrices is denoted by $GL_d$.  The linear subspace of $d\times d$ symmetric matrices is denoted by $\mathbb{S}_d$, the closed positive cone of symmetric matrices with non-negative real parts of their eigenvalues is denoted as $\mathbb{S}_d^+$, and the open positive definite cone of symmetric matrices with strictly positive real parts of eigenvalues is denoted as $\mathbb{S}_d^{++}$. The identity matrix of size $d\times d$ is represented by $\textbf{I}_{d\times d}$. The elements of a matrix $\mathbf{A}\in M_{d,k}(\mathbb{R})$ are denoted by $A_{ij}$ and the transposed matrix is given by $\mathbf{A}^\top$. Moreover, the spectrum of a matrix consisting of all eigenvalues is denoted by $\sigma(\cdot)$.  Also, we define
$$M_d^-:=\{\mathbf{X}\in M_d(\mathbb{R}):\sigma(\mathbf{X})\subset (-\infty,0)+i\mathbb{R}\},$$
as the set of square $d\times d$ matrices with negative real parts of their eigenvalues and $\mathcal{B}_b(M_d^-\times\mathbb{R})$ to be the collection of bounded Borel sets of $M_d^-\times\mathbb{R}$. The Kronecker (tensor) product of two matrices $\mathbf{A}\in M_{d,n}(\mathbb{R})$ and $\mathbf{B}$ is denoted as $\mathbf{A}\otimes \mathbf{B}$. The vectorisation transformation, which stacks the columns of a $d\times d$ matrix into a vector in $\mathbb{R}^{d^2}$, is represented as $\text{vec}$. Additionally, the half vectorisation, which transforms the upper or lower triangular elements of a matrix into a vector, is denoted as %$\text{vec}_h$
$\operatorname{vech}$.
The norm of vectors or matrices is denoted by $\|\cdot\|$. The chosen norm does not affect the results, as all norms are equivalent. However, it can be interpreted as either the Euclidean norm or the induced operator norm. The operator $\mathbf{1}_T$ stands for the indicator function of a set $T$. 
The Borel $\sigma$-algebra is denoted as $\mathcal{B}$ and $\lambda$ denotes the Lebesgue measure.

\subsection{Lévy bases and multivariate supOU processes}
In the following, we will construct graph supOU processes, where the driving noise is given by a Lévy basis, i.e.~an infinitely divisible independently scattered random variable defined as follows.
\begin{definition} An $\mathbb{R}^d$-valued Lévy basis on $M_d^-\times\mathbb{R}$ is defined as a collection $\Lambda=\{\Lambda(B):B\in\mathcal{B}_b(M_d^-\times\mathbb{R})\}$ of random variables taking values in $\mathbb{R}^d$ satisfying:
1)  the probability distribution of $\Lambda(B)$ is infinitely divisible for every $B\in\mathcal{B}_b(M_d^-\times\mathbb{R})$,
    2) for any natural number $l$ and pairwise disjoint sets $B_1,\dots,B_l\in\mathcal{B}_b(M_d^-\times\mathbb{R})$, the random variables $\Lambda(B_1),\dots,\Lambda(B_l)$ are independent, and
    3) for any pairwise disjoint sets $B_i\in\mathcal{B}_b(M_d^-\times\mathbb{R})$ for $i\in\mathbb{N}$, where $\cup_{l\in\mathbb{N}} B_l\in\mathcal{B}_b(M_d^-\times\mathbb{R})$, the series $\sum_{l=1}^\infty \Lambda(B_l)$ converges almost surely, and $\Lambda(\cup_{l\in\mathbb{N}} B_l)=\sum_{l=1}^\infty \Lambda(B_l)$ almost surely.
\end{definition}
In the realm of supOU processes, we focus on Lévy bases satisfying 
the Lévy Khintchine representation:
$
\text{E}(\exp(iu^\top\Lambda(B)))=\exp(\phi(u)\Pi(B))$,
for all $u\in\mathbb{R}^d$ and $B\in\mathcal{B}_b(M_d^-(\mathbb{R})\times\mathbb{R})$, where $\Pi=\pi\times\lambda$ is the product measure of a (probability) measure $\pi$ on $M_d^-(\mathbb{R})$ and the Lebesgue measure $\lambda$ on $\mathbb{R}$.
The function 
$\phi(u)=iu^\top\gamma-\frac{1}{2}u^\top\mathbf{\Sigma} u+\int_{\mathbb{R}^d}\left(e^{iu^\top x}-1-iu^\top x \mathbf{1}_{[0,1]}(\|x\|)\right)\nu(dx)$
represents the cumulant transform of an infinitely divisible distribution on $\mathbb{R}^d$ with Lévy Khintchine triplet $(\gamma,\mathbf{\Sigma},\nu)$, where $\gamma\in\mathbb{R}^d$, $\mathbf{\Sigma}\in\mathbb{S}_d^+$, and $\nu$ is a Lévy measure. As described in \cite{Bar11}, the distribution of the Lévy bases is fully determined by the generating quadruple $(\gamma,\mathbf{\Sigma},\nu,\pi)$.
We recall that a $d$-dimensional supOU
 is given by $X=(X_t)_{t\in\mathbb{R}}$, where
\begin{equation}\label{eq:supOU}
X_t=\int_{M_d^-}\int_{-\infty}^t e^{\mathbf{Q}(t-s)}\Lambda(d\mathbf{Q},ds), \quad t \in \mathbb{R}.
\end{equation}
We review the conditions that ensure that such a process is well-defined next.
\begin{theorem}{(Theorem 3.1, \cite{Bar11})}\label{th:supOU}
Let $\Lambda$ be an $\mathbb{R}^d$-valued Lévy basis on $M_d^-\times\mathbb{R}$ with generating quadruple $(\gamma,\Sigma,\nu,\pi)$ satisfying 
$\int_{\|x\|>1} \ln(\|x\|)\nu(dx)<\infty$
and assume there exist measurable functions $\rho:M_d^-\rightarrow \mathbb{R}^+\setminus\{0\}$ and $\kappa:M_d^-\rightarrow [1,\infty)$ such that
$\|e^{\mathbf{Q}s}\|\leq \kappa(\mathbf{Q})e^{-\rho(\mathbf{Q})s}$ for all $s\in\mathbb{R}^+$, $\pi$-almost surely,
and
    $\int_{M_d^-}\frac{\kappa(\mathbf{Q})^2}{\rho(\mathbf{Q})}\pi(d\mathbf{Q})<\infty$.
Then the supOU process $(X_t)_{t\in\mathbb{R}}$ defined in \eqref{eq:supOU} 
and is well-defined for all $t\in\mathbb{R}$.
\end{theorem}
As a special case of a mixed moving average process, the  supOU process inherits the  strict stationarity and infinite divisibility, cf.~\cite[Definition 2.2]{CS}.

\subsection{Definition and properties of graph supOU processes}
The graph supOU process is defined as a special case of the multivariate supOU process and extends the ideas of \cite{CV22, CV22b,LPV2024} developed for graph OU  and graph CAR(p) processes, respectively.
As in these earlier articles, we interpret the components of the process $X$  as representing nodes/vertices of a graph, connected to each other by a set of edges. These edges are represented by the adjacency matrix $\textbf{A}=(a_{ij})$: We set $a_{ij}=1$ if there is an (undirected) edge from $i$ and $j$ and $0$ otherwise. Also, we set $a_{ii}=0$, for all $i\in\{1, \ldots,d\}$. 
For all $i\in\{1, \ldots,d\}$, we denote by $n_i:=\max\{1,\sum_{j=1}^da_{ji}\}$ the in-degree of node $i$, i.e.~the number of edges leading to node $i$.
Let $\textbf{D}:=\mathrm{diag}(n_1^{-1},\ldots, n_d^{-1})$.
In the following, we will work with the 
column-normalised adjacency matrix defined as  $\bar{\textbf{A}}:=\textbf{A}\textbf{D}=(\frac{a_{ij}}{n_j}) \in M_d(\mathbb{R})$. 
 Other normalisations, including symmetric ones,  are possible. However,  column normalisation is the natural choice when the focus is on the incoming relationships from neighbouring nodes, as it scales the adjacency matrix according to the number of incoming edges. By contrast, a row normalisation would scale according to the number of outgoing edges.
In the graph supOU process, we assume that  the
  drift matrix $\mathbf{Q}$ is parametrised  as 
\begin{equation}\label{eq:Q}
   \mathbf{Q}=\mathbf{Q}( \theta)=-\left(\theta_2\mathbf{I}_{d\times d}+\theta_1\bar{\textbf{A}}^{\top}\right),
\end{equation}
where $\theta
=(\theta_1,\theta_2)^T\in\mathbb{R}^2$ is the parameter vector of interest and $\bar{\textbf{A}}$ represents the %column-
normalised adjacency matrix with 
\begin{align*}
\bar{\textbf{A}}^{\top}= 
\begin{pmatrix}
 \frac{a_{11}}{n_1} &   \frac{a_{21}}{n_1} & \cdots &   \frac{a_{d1}}{n_1} \\
  \frac{a_{12}}{n_2} &  \frac{a_{22} }{n_2} & \cdots &  \frac{a_{d2} }{n_2} \\
\vdots & \vdots & \ddots & \vdots \\
 \frac{a_{1d}}{n_d} &  \frac{a_{2d} }{n_d} & \cdots &  \frac{a_{dd} }{n_d} \\
\end{pmatrix}, 
\quad \mathrm{i.e.}\, (\bar{\textbf{A}}^{\top})_{ij}=\frac{a_{ji}}{n_i}, \; i, j \in \{1, \ldots, d\}.
\end{align*}
Throughout the article, we assume that the following assumption holds. 
\begin{assumption}\label{as:ev}
Assume that 
$\theta_2>|\theta_1|$. 
\end{assumption}
\begin{proposition}\label{prop:eigenv}
    If Assumption \ref{as:ev} holds,  then $\mathbf{Q}(\mathbf{\theta})\in M_d^-$.
\end{proposition} 
%%%%%%%%%%%%%%%Proof Gresshgorin circle
\begin{proof}[Proof of Proposition \ref{prop:eigenv}]  We show that all eigenvalues of \( \mathbf{Q}(\theta) \) are strictly negative. By Gershgorin's circle theorem \cite{GSC}, any eigenvalue of \( \mathbf{Q}(\theta) \) lies within at least one Gershgorin disc. For the \(i\)-th row of \( \mathbf{Q}(\theta) \), the centre of the Gershgorin disc is the diagonal entry \( Q_{ii} \), which is \(- \theta_2 \), and the radius, $R_i$, is the sum of the absolute values of the off-diagonal entries in the \(i\)-th row. I.e.
%\begin{align*}
$R_i=\sum_{j \neq i} |Q_{ij}(\theta)| = |\theta_1| \sum_{j \neq i} |(\bar{A}^{\top})_{ij}| 
= |\theta_1| \sum_{j \neq i} \frac{a_{ji}}{n_i}
= |\theta_1|$.
%\end{align*}
Therefore, each Gershgorin disc for \( Q(\theta) \) has a center at \( -\theta_2 \) and a radius smaller or equal to \( |\theta_1| \). Since \( \theta_2 > |\theta_1| \), the disc centered at \( -\theta_2 \) with radius \( |\theta_1| \) lies entirely in the left half-plane of the complex plane and does not intersect the imaginary axis. 
Consequently, all eigenvalues of \( \mathbf{Q}(\theta) \)  must have  strictly negative real parts, i.e.~$\mathbf{Q}(\mathbf{\theta})\in M_d^-$. 
\end{proof}
Let us denote by
\begin{align*}
    \Theta:=\{\theta=(\theta_1,\theta_2)^{\top}\in \mathbb{R}^2: \mathbf{Q}(\mathbf{\theta})\in M_d^-\},  
\end{align*}
the parameter space for which the drift matrix has eigenvalues with strictly negative real parts.
We have now introduced all prerequisites for defining graph supOU processes. 
\begin{definition}\label{def:grsupOU}
A graph supOU process $X=(X_t)_{t\in \mathbb{R}}$ is defined as a multivariate supOU process satisfying the conditions given in Theorem \ref{th:supOU} and drift matrix defined as in \eqref{eq:Q} satisfying $\mathbf{Q}(\mathbf{\theta})\in M_d^-$.
\end{definition}
 We note that, as a special 
case of a mixed moving average process, the graph supOU process inherits strict 
stationarity, infinite divisibility, ergodicity and the mixing property, see 
\cite{CS}.

We are interested in randomising the parameter vector $\theta$ to obtain more flexible models. 

\begin{proposition}\label{prop:mom}
Let $X=(X_t)_{t\in \mathbb{R}}$ denote a graph supOU process as defined in Definition \ref{def:grsupOU}. 
Suppose  that $\int_{\mathbb{R}^d}\|x\|^2\nu(dx)<\infty$, and  set 
$$\mu_L:=\gamma+\int_{\|x\|>1}x\nu(dx), \quad
\boldsymbol{\sigma}^2_{L}:=\Sigma+\int_{\mathbb{R}^d}xx^\top\nu(dx).$$ The first and second moments of the graph supOU process are given by  

\begin{multline}\label{eq:2ndmom}
\mathrm{E}(X_0)
=-\int_{\Theta}\mathbf{Q(\theta)}^{-1} \pi(d\theta) \mu_L, \quad %
\mathrm{var}(X_0) 
=-\int_{\Theta}(\mathcal{A}(\mathbf{Q}(\theta)))^{-1} (\boldsymbol{\sigma}^2_{L})\pi(d\theta) ,%\\
\\
  \mathrm{cov}(X_h,X_0)
  =-\int_{\Theta}e^{\mathbf{Q}(\theta)h}(\mathcal{A}(\mathbf{Q}(\theta)))^{-1}(\boldsymbol{\sigma}^2_{L})\pi(d\theta),
  \end{multline}
 where $\mathcal{A}(\mathbf{Q}(\theta)):M_d(\mathbb{R})\rightarrow M_d(\mathbb{R}), \mathbf{X} \rightarrow \mathbf{Q}(\theta)\mathbf{X}+\mathbf{X}\mathbf{Q}(\theta)^\top$,  and $h\geq 0$.
\end{proposition}
%%%%%%%%%%%%%%%%%%%%%%%%%%%%%%%
%Proofs moments
\begin{proof}[Proof of Proposition \ref{prop:mom}]
The results are an immediate consequence of  
\cite[Theorem 3.11]{Bar11}. 
To aid the interpretation of the notation used in the proposition, we note the following. 
When computing the variance (and similarly for the covariance), we get the expression
%\begin{align*}
$\mathrm{var}(X_0)=\int_{M_d^-}\int_0^{\infty}e^{\mathbf{Q}(\theta)s}\boldsymbol{\sigma}^2_{L}e^{\mathbf{Q}(\theta)^{\top}s}ds\pi(d\mathbf{Q}(\theta))$.
%\end{align*}
Let us focus on the inner integral:
\begin{align}\label{eq:I}
\mathbf{P}:=\int_0^{\infty}e^{\mathbf{Q}(\theta)s}\boldsymbol{\sigma}^2_{L}e^{\mathbf{Q}(\theta)^{\top} s}ds.
\end{align}
It is well-known that $\mathbf{P}$ is a solution to the Lyapunov equation
%\begin{align*}
$\mathbf{Q}(\theta)\mathbf{P}+\mathbf{P}\mathbf{Q}(\theta)^{\top}=-\boldsymbol{\sigma}^2_{L}$.
%\end{align*}
We define the map $\mathcal{A}(\mathbf{Q}(\theta)):M_d(\mathbb{R})\rightarrow M_d(\mathbb{R}), \mathbf{P} \rightarrow \mathbf{Q}(\theta)\mathbf{P}+\mathbf{P}\mathbf{Q}(\theta)^\top$, i.e.
$
\mathcal{A}(\mathbf{Q}(\theta))(\mathbf{P})=\mathbf{Q}(\theta)\mathbf{P}+\mathbf{P}\mathbf{Q}(\theta)^\top$.
The integral $\mathbf{P}$ defined in \eqref{eq:I} can be obtained using the inverse map $\mathcal{A}(\mathbf{Q}(\theta))^{-1}:M_d(\mathbb{R})\rightarrow M_d(\mathbb{R})$. More precisely, 
we solve
$
\mathcal{A}(\mathbf{Q}(\theta))(\mathbf{P})=\mathbf{Q}(\theta)\mathbf{P}+\mathbf{P}\mathbf{Q}(\theta)^\top=\mathbf{Y},$
for $\mathbf{P}$ and then later set $\mathbf{Y}=-\boldsymbol{\sigma}^2_{L}$.
Using vectorisation, we get 
%\begin{align*}
$\mathrm{vec}(\mathbf{Y})=\mathrm{vec}(\mathbf{Q}(\theta)\mathbf{P}+\mathbf{P}\mathbf{Q}(\theta)^\top)
=(\mathbf{I}_{d\times d}\otimes \mathbf{Q}(\theta)+\mathbf{Q}(\theta)^\top\otimes \mathbf{I}_{d\times d})\mathrm{vec}(\mathbf{P})$. 
%\end{align*}
We then get
%\begin{align*}
$\mathrm{vec}(\mathbf{P})
=(\mathbf{I}_{d\times d}\otimes \mathbf{Q}(\theta)+\mathbf{Q}(\theta)^\top\otimes \mathbf{I}_{d\times d})^{-1}\mathrm{vec}(\mathbf{Y})$. 
%\end{align*}
Hence,
%\begin{align*}
$\mathcal{A}(\mathbf{Q}(\theta))^{-1}(\mathbf{Y})=\mathrm{vec}^{-1}((\mathbf{I}_{d\times d}\otimes \mathbf{Q}(\theta)+\mathbf{Q}(\theta)^\top\otimes \mathbf{I}_{d\times d})^{-1}\mathrm{vec}(\mathbf{Y}))$.
%\end{align*}
We note that 
%\begin{align*}
$\mathbf{P}=\mathcal{A}(\mathbf{Q}(\theta))^{-1}(-\boldsymbol{\sigma}^2_{L})=\mathrm{vec}^{-1}((\mathbf{I}_{d\times d}\otimes \mathbf{Q}(\theta)+\mathbf{Q}(\theta)^\top\otimes \mathbf{I}_{d\times d})^{-1}\mathrm{vec}(\mathbf{-\boldsymbol{\sigma}^2_{L}}))
%\\
%&
=-\mathcal{A}(\mathbf{Q}(\theta))^{-1}(\boldsymbol{\sigma}^2_{L})$.
%\end{align*}
Hence, altogether we have 
%\begin{align*}
$\mathrm{var}(X_0)=-\int_{M_d^-}\mathcal{A}(\mathbf{Q}(\theta))^{-1}(\boldsymbol{\sigma}^2_{L})\pi(d\mathbf{Q}(\theta))$.
%\end{align*}
\end{proof}
%%%%%%%%%%%%%%%%%%%%%%%%%

\subsection{Parametric models  of graph supOU type -- Bridging between long and short memory}\label{sec:bridge}
In this section, we introduce new parametric models within the class of graph supOU processes. Specifically, we focus on model specifications that are parsimonious, analytically tractable, and capable of capturing both short- and long memory.
In the case of graph supOU processes, there are two parameters of interest: $\theta_1$ representing the network effect, i.e.~the influence of the neighbouring nodes, and $\theta_2$ representing the momentum effect, i.e.~the autoregressive effect for a particular node. 
We will now work with a reparametrisation of the drift matrix $\mathbf{Q( \theta)}$ defined in \eqref{eq:Q}, by setting $c:=\frac{\theta_1}{\theta_2}$ and write
\begin{align*}%\label{eq:Qc}
    \mathbf{Q( \theta)}&=-\left(\theta_2\mathbf{I}_{d\times d}+\theta_1\bar{\textbf{A}}^{\top}\right)
    =\theta_2\left(-\mathbf{I}_{d\times d}-\frac{\theta_1}{\theta_2}\bar{\textbf{A}}^{\top}\right)
    =\theta_2\left(-\mathbf{I}_{d\times d}-c\bar{\textbf{A}}^{\top}\right)\\
    &=:\mathbf{Q}(c,\theta_2).
\end{align*}
As long as Assumption \ref{as:ev} holds, we have
$\left(-\mathbf{I}_{d\times d}-c\bar{\textbf{A}}^{\top}\right)\in M_d^-$.
Equivalently, assuming that $\theta_2>0$ and $|c|<1$, we have $\mathbf{Q}(c,\theta_2)\in M_d^-$ since $\rho(c\bar{\textbf{A}}^{\top})=|c|\rho(\bar{\textbf{A}})\leq|c|< 1$. We can then write
\begin{equation}\label{eq:Qtheta2}
    \mathbf{Q}(c, \theta_2)
    =\theta_2\mathbf{K}(c), \quad \mathrm{where}\quad \mathbf{K}(c):=\left(-\mathbf{I}_{d\times d}-c\bar{\textbf{A}}^{\top}\right).
\end{equation}
In the following, we will assume that $|c|<1$ is fixed and we will randomise the parameter $\theta_2$, ensuring that we use only probability distributions supported on the positive half line, consistent with $\theta_2>0$. 
In this case, the second-order properties presented in \eqref{eq:2ndmom} simplify to 
\begin{align}\begin{split}\label{eq:2ndmomb}
\mathrm{E}(X_0)&=-\left(\int_{0}^{\infty}\frac{1}{\theta_2}\pi(d\theta_2) \right)\, \mathbf{K}(c)^{-1}  \mu_L,
\\
\mathrm{var}(X_0)&
=-\left(\int_{0}^{\infty}\frac{1}{\theta_2}\pi(d\theta_2)\right) (\mathcal{A}(\mathbf{K}(c)))^{-1} (\boldsymbol{\sigma}^2_{L}),\\
  \mathrm{cov}(X_h,X_0)&
  =-\left(\int_{0}^{\infty}\frac{1}{\theta_2}e^{\theta_2 \mathbf{K}(c) h}\pi(d\theta_2)\right) (\mathcal{A}(\mathbf{K}(c)))^{-1}(\boldsymbol{\sigma}^2_{L}),
  \end{split}
  \end{align}
where $\mathcal{A}(\mathbf{K}(c)):M_d(\mathbb{R})\rightarrow M_d(\mathbb{R}), \mathbf{X} \rightarrow \mathbf{K}(c)\mathbf{X}+\mathbf{X}\mathbf{K}(c)^\top$,  and $h\geq 0$.
\begin{remark}
The formula for the mean contains the matrix
$(-\mathbf{K}(c))^{-1}
    = \left(\mathbf{I}_{d\times d}-(-c)\bar{\textbf{A}}^{\top}\right)^{-1}
    =\sum_{k=0}^{\infty}(-c)^k \left(\bar{\textbf{A}}^{\top}\right)^k
    =\mathbf{I}_{d\times d}-c \bar{\textbf{A}}^{\top}+c^2(\bar{\textbf{A}}^{\top})^2+\cdots$,
where the Neumann series converges since $\rho(c\bar{\mathbf{A}}^\top)<1$. This expansion admits a natural graph-theoretic 
interpretation: the $(i,j)$-th entry of $(\bar{\mathbf{A}}^\top)^k$ aggregates 
the weights of all paths of length $k$ from node $j$ to node $i$, so the mean 
$\mathrm{E}(X_{0,i})$ receives contributions from the L\'{e}vy basis means 
$\mu_{L,j}$ at all nodes $j$, with paths of length $k$ contributing a factor 
of $(-c)^k$. Since $|c|<1$, distant nodes contribute negligibly, and $c$ 
controls the spatial reach of dependence across the network.
\end{remark}

The second order properties further reveal that, for $h\geq 0$, the scaled autocovariance matrix  
\begin{align*}
 \mathrm{cov}(X_h,X_0)   (\mathrm{var}(X_0))^{-1}=
 \left(\int_{0}^{\infty}\frac{1}{\theta_2}e^{\theta_2 \mathbf{K}(c) h}\pi(d\theta_2)\right)\left(\int_{0}^{\infty}\frac{1}{\theta_2}\pi(d\theta_2)\right)^{-1},
\end{align*}
does neither depend on $\mu_L$  nor $\boldsymbol{\sigma}^2_{L}$. This suggests that we can consider a two-step procedure for inference, where the parameters of $\pi$ and the parameter $c$ can be inferred from the empirical counterpart of the scaled autocovariance matrix, and the  remaining parameters can then be estimated from the sample mean and variance. 
This decoupling is intrinsic to the second-order structure of the model 
and forms the theoretical basis for the two-step procedure developed 
in Section~\ref{sec:GMMestimator}.

\subsubsection{Choosing a sum of exponentials for $\pi(\theta_2)$ }

As mentioned earlier, the idea of supOU processes is connected to a (countable) sum of independent OU processes. Hence, we can consider choosing
\begin{align}\label{eq:Exp}
    \pi(d\theta_2)=\sum_{i=1}^{\infty}w_i\delta_{\lambda_i}(d\theta_2),
    \, \mathrm{for}\, \lambda_i>0, w_i\geq0, \sum_{i=1}^{\infty}w_i=1,
\end{align}
where $\delta_x(\cdot)$ denotes the Dirac delta function at $x>0$.
In that case, we have
\begin{align*}
\int_{0}^{\infty}\frac{1}{\theta_2}\pi(d\theta_2)
&= \sum_{i=1}^{\infty}\frac{w_i}{\lambda_i},\\
\int_{0}^{\infty}\frac{1}{\theta_2}e^{\theta_2 \mathbf{K}(c) h}\pi(d\theta_2)
&= \sum_{i=1}^{\infty}\frac{w_i}{\lambda_i}e^{ \lambda_i \mathbf{K}(c) h}, \; h \geq 0.
\end{align*}
We note that the expression above implies that the autocovariance function behaves like a sum of weighted exponentials. 
\begin{remark}
Let us consider a special case: 
    Suppose that $\mathbf{K}(c)$ is diagonalisable and let $k_1, \ldots, k_d$ denote the eigenvalues of $\mathbf{K}(c)$ and suppose that $\mathbf{O}\mathbf{K}(c)\mathbf{O}^{-1}=\mathrm{diag}(k_1, \ldots, k_d)$. Then, 
\begin{align*}
    \sum_{i=1}^{\infty}\frac{w_i}{\lambda_i}e^{ \lambda_i \mathbf{K}(c) h}
     =\mathbf{O}
    \sum_{i=1}^{\infty}\frac{w_i}{\lambda_i}
        \mathrm{diag}(e^{ \lambda_ik_1 h}, \ldots, e^{ \lambda_ik_d h})
    \mathbf{O}^{-1}, \; h \geq 0.
\end{align*}
This result shows the structure of the autocovariance function even more clearly than the general expression above. 
\end{remark}

\begin{example}
We note that the graph OU case can be recovered from the above example by setting $w_1=1$, $w_i=0$ for all $i\geq 2$. To simplify the exposition, we also write $\lambda_1=\lambda$. Then 
\begin{align*}
\int_{0}^{\infty}\frac{1}{\theta_2}\pi(d\theta_2)
= \frac{1}{ \lambda},&&
\int_{0}^{\infty}\frac{1}{\theta_2}e^{\theta_2 \mathbf{K}(c) h}\pi(d\theta_2)
= \frac{1}{\lambda}e^{ \lambda \mathbf{K}(c) h}, \; h \geq 0.
\end{align*}
\end{example}

\subsubsection{Choosing a Gamma law for $\pi(\theta_2)$ }
Next, we propose a specification that can capture both short- and long-range dependence. 
Adapting  \cite[Example 3.1]{Bar11}, we assume that 
$\theta_2 \sim \Gamma(\alpha, 1)$, with 
 \begin{equation}{\label{eq:PD}}\pi(d\theta_2)=\frac{1}{\Gamma(\alpha)}\theta_2^{\alpha-1}e^{- \theta_2}\mathbf{1}_{(0,\infty)}(\theta_2)d\theta_2,\quad\; \mathrm{where}\; \alpha>1.\end{equation}
  We note that setting the second parameter in the Gamma distribution to a constant (here chosen as 1) is needed to ensure the model identifiability. 
The moments of the graph supOU process can then be obtained from \eqref{eq:2ndmomb} by noting that, for $h\geq 0$, 
\begin{align*}
\int_{0}^{\infty}\frac{1}{\theta_2}\pi(d\theta_2)
= \frac{1}{\alpha-1},&&
\int_{0}^{\infty}\frac{1}{\theta_2}e^{\theta_2 \mathbf{K}(c) h}\pi(d\theta_2)
= \frac{1}{\alpha-1}(\mathbf{I}_{d\times d}-\mathbf{K}(c)h)^{1-\alpha}. %, \; h \geq 0.
\end{align*}
Hence, 
\begin{align}\begin{split}\label{eq:2ndmombgamma}
\mathrm{E}(X_0)&=-\frac{1}{\alpha-1} \mathbf{K}(c)^{-1}  \mu_L,
\\
\mathrm{var}(X_0)&
=-\frac{1}{\alpha-1}(\mathcal{A}(\mathbf{K}(c)))^{-1} (\boldsymbol{\sigma}^2_{L}),\\
  \mathrm{cov}(X_h,X_0)&
  =-\frac{1}{\alpha-1}(\mathbf{I}_{d\times d}-\mathbf{K}(c)h)^{1-\alpha} (\mathcal{A}(\mathbf{K}(c)))^{-1}(\boldsymbol{\sigma}^2_{L}),
  \end{split}
  \end{align}
where $\mathcal{A}(\mathbf{K}(c)):M_d(\mathbb{R})\rightarrow M_d(\mathbb{R}), \mathbf{X} \rightarrow \mathbf{K}(c)\mathbf{X}+\mathbf{X}\mathbf{K}(c)^\top$,  and $h\geq 0$.

We observe that the autocovariance function exhibits polynomial decay since $(\mathbf{I}_{d\times d}-\mathbf{K}(c)h)^{1-\alpha} \sim h^{1-\alpha}$, as $h\to \infty$. Specifically, for values of $\alpha\in(1,2)$, we have long memory and for $\alpha\geq 2$ short memory. 

\begin{remark}
We note that the case of the Gamma law can easily be generalised to settings where $\pi(d \theta_2)=\theta_2^{\alpha-1}l(\theta_2)d\theta_2$, where $\alpha >1$ and $l$ is a function that slowly varies at 0.
\end{remark}

\section{Parameter estimation for the graph supOU process}\label{sec:GMMestimator}
Since supOU processes are generally not Markovian and not of semimartingale type, maximum likelihood estimation is not straightforward. Given their analytically tractable moment structure, a (generalised) method of moments approach seems promising for parameter estimation.

Let us  denote by $\vartheta$ the parameter (vector) associated with the (probability) measure $\pi$ with appropriate parameter space $\Theta$.
The parameters of interest are $\xi=(\vartheta, c, \mathrm{vec}(\mu_L), \operatorname{vech}(\boldsymbol{\sigma}_L^2))$, see \eqref{eq:2ndmombgamma}, assuming the finiteness of the mean and variance throughout. 
Suppose we observe a sample of the  graph supOU process $\{X_t, t=\Delta, 2\Delta, \ldots, N\Delta\}$, for $\Delta>0, N\in\mathbb{N}$. 

\subsection{A two-step estimation procedure}\label{sec:ident}
Let us first
present a method for stepwise parameter estimation that simultaneously tackles the question of the parameter identifiability in the context of the parameterisation proposed in Section \ref{sec:bridge}.  Define the scaled autocovariance matrix by 
\begin{align*}
\mathbf{R}(h\Delta; \pi, c)&:= \mathrm{cov}(X_{h\Delta},X_0) (\mathrm{var}(X_0))^{-1}
\\
&=
 \left(\int_{0}^{\infty}\frac{1}{\theta_2}e^{\theta_2 \mathbf{K}(c) h}\pi(d\theta_2)\right)\left(\int_{0}^{\infty}\frac{1}{\theta_2}\pi(d\theta_2)\right)^{-1}.
 \end{align*}
 Adapting the  idea suggested in  \cite[Section 5]{Bar11} in the context of general supOU processes, we consider estimating the parameters of $\pi$ and the parameter $c$ by considering suitable empirical counterparts of the eigenvalues of  $\mathbf{R}(h\Delta; \pi, c)$. To this end, we first note that the eigenvalues of $\mathbf{K}(c)$ can be expressed in terms of the eigenvalues of the normalised adjacency matrix $\overline{\mathbf{A}}$. Let $a_1, \ldots, a_d\in \sigma(\overline{\mathbf{A}})$ denote the (not necessarily distinct) eigenvalues of $\overline{\mathbf{A}}$ (which are the same as for its transpose). 
By construction, for all $i=1, \ldots, d$, we have $|a_i|\leq 1$. 
Let $a^*=\rho(\overline A)$ denote the spectral radius of $\overline A$. Then $|a_i|\leq a^*$, for all $i=1,\ldots,d$. 
By the Perron-Frobenius theorem for nonnegative matrices, $a^*$ is a real eigenvalue of $\overline A$. 
If the graph associated with $A$ is strongly connected in the sense that every node is reachable from every other node,  then $a^*=1$. We note that, since we assume that the network structure is known, $a^*$ is known, too, and does not need to be inferred.
We also note that the eigenvalues of $\mathbf{K}(c)$ take the form $-1-ca_i$, for $i=1, \ldots, d$.
We observe that the maximum eigenvalue of the scaled autocovariance function can be written as a function of $\pi$ and $c$ as follows:
\begin{align}\begin{split}\label{eq:evrelation}
 \rho(h\Delta; \pi, c)&:=\rho(\mathbf{R}(h\Delta; \pi, c)) \\
 &= \left(\int_{0}^{\infty}\frac{1}{\theta_2}e^{-\theta_2 (1+ca^*) h}\pi(d\theta_2)\right)\left(\int_{0}^{\infty}\frac{1}{\theta_2}\pi(d\theta_2)\right)^{-1}.
 \end{split}
\end{align}
Let us consider the two examples introduced earlier.
\begin{example} In the case when $\pi$ is as in \eqref{eq:Exp}, then set
$$
\phi_i:= \left( \sum_{j=1}^{\infty}\frac{w_j}{\lambda_j}\right)^{-1}\frac{w_i}{\lambda_i},
$$ and note that
     \begin{align*}
     \mathbf{R}(h\Delta; \pi,  c)=
\mathbf{R}(h\Delta; (\lambda_i)_i, (w_i)_i, c)&
=\sum_{i=1}^{\infty}\phi_ie^{ \lambda_i \mathbf{K}(c) h}, \; h \geq 0,
\end{align*}
which has eigenvalues of the form
$\sum_{i=1}^{\infty}\phi_ie^{ \lambda_i(-1-c a_j) h}$, for $j=1, \ldots, d$. We can hence set
\begin{align}\label{eq:estfct-exp}
   \rho(h\Delta; \pi, c)= \rho(h\Delta; (\lambda_i)_i, (w_i)_i, c):=\sum_{i=1}^{\infty}\phi_ie^{- \lambda_i (1+c a^*) h}, \quad h \in \{1, \ldots, N-1\}.
\end{align}
\end{example}
\begin{example} In the case when $\pi$ is as in \eqref{eq:PD}, then 
     \begin{align*}
     \mathbf{R}(h\Delta; \pi,  c)=
\mathbf{R}(h\Delta; (\alpha, c)^{\top})&
=(\mathbf{I}_{d\times d}-\mathbf{K}(c)h\Delta)^{1-\alpha}, %\\
\quad h \geq 0,
\end{align*}
which has eigenvalues of the form
$(1+(1+c a_i) h\Delta)^{1-\alpha}$, for $i=1, \ldots, d$. We can hence set
\begin{align}\label{eq:estfct-LM}
   \rho(h\Delta; \pi, c)= \rho(h\Delta; (\alpha, c)^{\top})=(1+(1+c a^*) h\Delta)^{1-\alpha}, \quad h \in \{1, \ldots, N-1\}.
\end{align}
\end{example}

The relation \eqref{eq:evrelation} can be directly used for inference. Specifically, it enables us to identify the parameters associated with $\pi$ and the parameter $c$
 independently of $\mu_L$ and $\boldsymbol{\sigma}_L^2$.

 Let us introduce the empirical quantities needed for inference. 
We define
\begin{align*}
\widehat{\mathrm{E}(X_0)}&=\overline X = \frac{1}{N}
\sum_{t=1}^{N} X_{t\Delta},\\
\widehat{\mathrm{var}(X_0)}&=\frac{1}{N-1}
\sum_{t=1}^{N} (X_{t\Delta}-\overline X) (X_{t\Delta}-\overline X)^{\top}, 
\\
\widehat{\mathrm{cov}(X_{h\Delta},X_0)} &=\frac{1}{N-h}\sum_{t=1}^{N-h}(X_{t\Delta}-\overline X) (X_{(t+h)\Delta}-\overline X)^{\top}, \quad h \in \{1, \ldots, N-1\}.
\end{align*}
We can then define the empirical counterpart of the scaled autocovariance matrix as
\begin{align}\label{eq:EmpR} \widehat{\mathbf{R}}(h\Delta):= \widehat{\mathrm{cov}(X_{h\Delta},X_0)} ( \widehat{\mathrm{var}(X_0)} )^{-1}, \quad h \in \{1, \ldots, N-1\}. \end{align}
\begin{remark}
    We  assume that  the empirical variance matrix is invertible. Otherwise, for instance,  when $N<<d$, the matrix is near-singular and suitable regularisation techniques would need to be applied in this step.
\end{remark}
Next, we find the corresponding empirical spectral radii: For $h \in \{1, \ldots, N-1\}$, we set 
\begin{align*} \hat l(h\Delta):=
 \rho(\widehat{\mathbf{R}}(h\Delta)).\end{align*}

We choose a constant $N^*\in \mathbb{N}$ with $N^*\leq N-1$ and define the quadratic loss function
\begin{align}\label{eq:loss}
L(\vartheta, c)&:=\frac{1}{N^*}\sum_{h=1}^{N^*}(\hat l(h\Delta) 
- \rho(h\Delta; \pi, c))^2.
\end{align}
Recall that we denote by $\vartheta$ the parameter (vector) associated with the (probability) measure $\pi$ with the appropriate parameter space $\Theta$. 
Minimising the loss function \eqref{eq:loss} leads to the estimate
\begin{align*}
(\hat \vartheta, \hat c):=\mathrm{arg} \min_{\vartheta \in \Theta, |c|<1} L(\vartheta,c).
\end{align*}
Provided the parameters $\vartheta$ and $c$ are identifiable via  $\rho$, as discussed in Section \ref{sec:identifiability}, we obtain the estimates $(\hat \vartheta, \hat c)$.
Finally, we seek to identify $\mu_L$ and $\boldsymbol{\sigma}^2_{L}$.
Note that if the parameters of the L\'{e}vy basis $\Lambda$ can be identified via the mean and variance, then that concludes the estimation procedure. Otherwise, additional  higher moment conditions will need to be considered.
We propose the  following estimators:
\begin{align*}
    \widehat{\mu_L}&=
    -\int_{0}^{\infty}\frac{1}{\theta_2}\widehat{\pi}(d\theta_2) 
    \mathbf{K}(\hat c)\overline{X},\\
 \widehat{\boldsymbol{\sigma}^2_{L}}
 &=
  -\int_{0}^{\infty}\frac{1}{\theta_2}\widehat{\pi}(d\theta_2) 
 [\mathbf{K}(\hat c)\widehat{\mathrm{var}(X_0)}+\widehat{\mathrm{var}(X_0)}\mathbf{K}(\hat c)^\top],
 \end{align*}
where $\widehat{\pi}$ denotes the (probability) measure $\pi$, when the parameter $\vartheta$ is replaced by its estimate $\hat \vartheta$.
We conclude 
this section with the following consistency result.

\begin{proposition}\label{prop:con} Suppose that the parameters $\vartheta$ of $\pi$ and $c$ are identifiable via the loss function \eqref{eq:loss} and that $\boldsymbol{\sigma}^2_{L}$ is invertible.
 Then, the 
parameters $\mu_L$ and $\boldsymbol{\sigma}^2_{L}$ are identifiable and the 
estimators  $\hat \vartheta, \hat c, \widehat{\mu_L}, \widehat{\boldsymbol{\sigma}^2_{L}}$ are consistent.
\end{proposition}
%%%%%%%%%%%%%%%%%%%%%%%%%%%%%%%%%%%
\begin{proof}[Proof of Proposition \ref{prop:con}]
This result follows directly from 
the fact that supOU processes and hence graph supOU processes are stationary and ergodic and hence their empirical moments convergence in probability to their population counterparts. Also, the eigenvalue considered in the estimation is a continuous function of the empirical scaled autocovariance matrix and hence the continuous mapping theorem applies.
Now assume that the matrix $(\mathcal{A}(\mathbf{K}(c))^{-1}(\boldsymbol{\sigma}^2_{L})$ is invertible, which is the case if $\boldsymbol{\sigma}^2_{L}$ is invertible.
Then the mean and variance can be identified as follows. We have
%\begin{align*}
    $\mathrm{E}(X_0)%&
    =-\left(\int_{0}^{\infty}\frac{1}{\theta_2}\pi(d\theta_2) \right)\,
 \mathbf{K}(c)^{-1}  \mu_L$, which is equivalent to $\mu_L=-\left(\int_{0}^{\infty}\frac{1}{\theta_2}\pi(d\theta_2) \right)^{-1}\,
\mathbf{K}(c)\mathrm{E}(X_0)$,
%\end{align*}
and for the variance we note that 
$\mathrm{var}(X_0)%&
=-\left(\int_{0}^{\infty}\frac{1}{\theta_2}\pi(d\theta_2) \right)\,
(\mathcal{A}(\mathbf{K}(c))^{-1} (\boldsymbol{\sigma}^2_{L})$, which is equivalent to $-\left(\int_{0}^{\infty}\frac{1}{\theta_2}\pi(d\theta_2) \right)^{-1}\,
\mathrm{var}(X_0) =(\mathcal{A}(\mathbf{K}(c))^{-1} (\boldsymbol{\sigma}^2_{L})$.
Hence,
\begin{align*}
\boldsymbol{\sigma}^2_{L}&=\mathcal{A}(\mathbf{K}(c))\left(-\left(\int_{0}^{\infty}\frac{1}{\theta_2}\pi(d\theta_2) \right)^{-1}\mathrm{var}(X_0)\right)\\
&
=
-\left(\int_{0}^{\infty}\frac{1}{\theta_2}\pi(d\theta_2) \right)^{-1}[\mathbf{K}(c)\mathrm{var}(X_0)+\mathrm{var}(X_0)\mathbf{K}(c)^\top].
\end{align*}
The consistency of the estimates then follows directly from the stationarity and ergodicity.

\end{proof}
%%%%%%%%%%%%%%%%%%%%%%%%%%%%%%%%%%
\subsubsection{Parameter identifiability}\label{sec:identifiability}

Identifiability needs to be checked for each concrete specification of
$\pi$. We discuss the $\Gamma(\alpha,1)$ case in detail, as it is the
primary model of interest. The mixture-of-exponentials case follows
analogously. 
In the $\Gamma(\alpha,1)$ case, the spectral radius function takes the
explicit form
\begin{align}\label{eq:rho_gamma}
    \rho(h\Delta;\alpha,c) = (1+(1+ca^*)h\Delta)^{1-\alpha},
    \quad h \in \{1,\ldots,N-1\},
\end{align}
and $(\alpha, c)$ are identifiable from \eqref{eq:rho_gamma} provided
$1+ca^*\neq 0$. Indeed, for any two lags $h_1\neq h_2$, the ratio
\begin{align*}
    \frac{\rho(h_1\Delta;\alpha,c)}{\rho(h_2\Delta;\alpha,c)}
    = \left(\frac{1+(1+ca^*)h_1\Delta}{1+(1+ca^*)h_2\Delta}\right)^{1-\alpha}
\end{align*}
is strictly monotone in $\alpha$ for fixed $c$ and vice versa, so
$(\alpha,c)$ is uniquely determined by any two distinct lag ratios.
The degenerate case $c=-1/a^*$ is excluded since $|c|<1$ and $a^*\leq 1$.

\begin{figure}[htbp]
  \centering
  \begin{subfigure}[b]{0.23\textwidth}
    \centering
    \includegraphics[width=\textwidth]{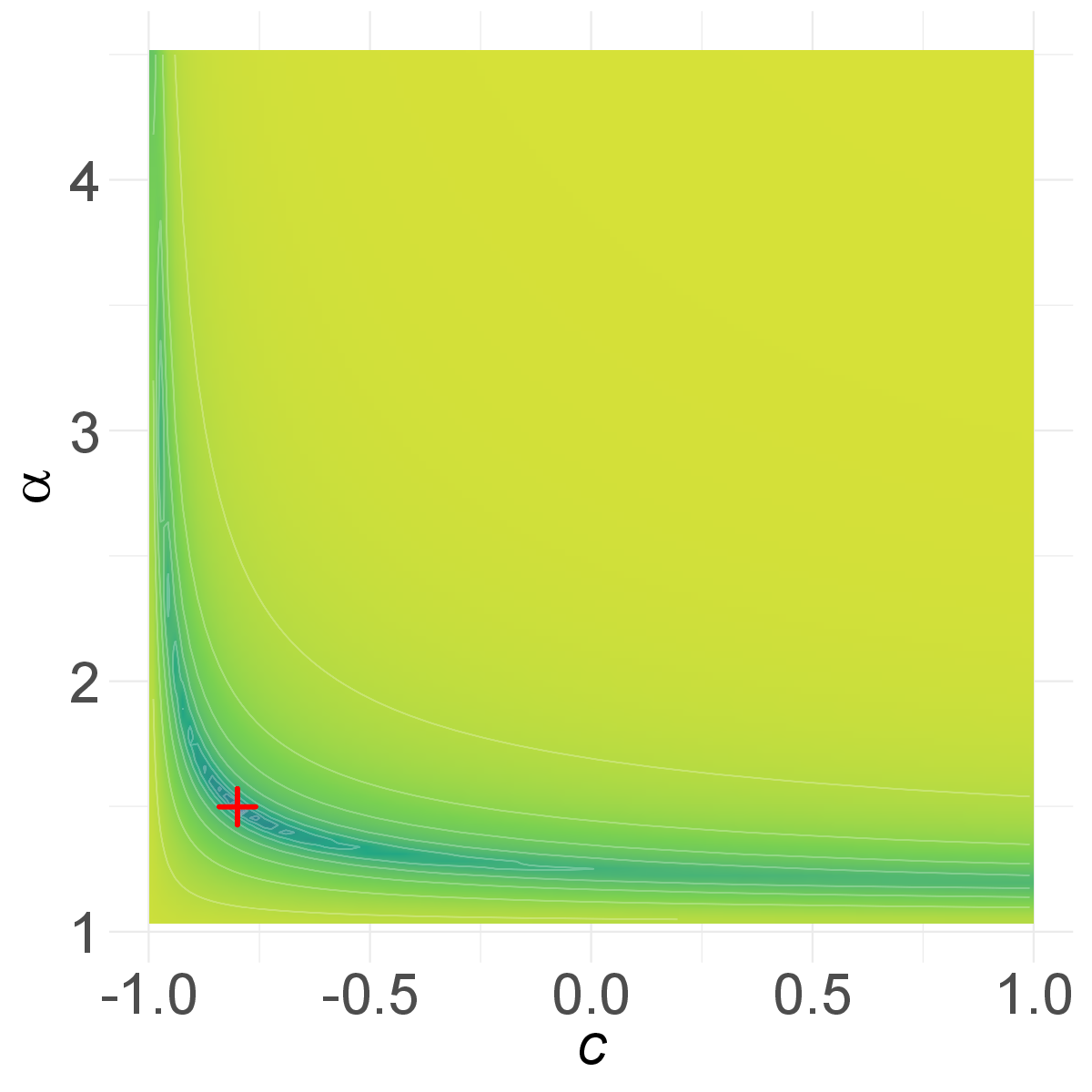}
    \caption{$\alpha_{0} = 1.5, \; c_0=-0.8$\\(long memory)}
  \end{subfigure}%
  \hfill
  \begin{subfigure}[b]{0.23\textwidth}
    \centering
    \includegraphics[width=\textwidth]{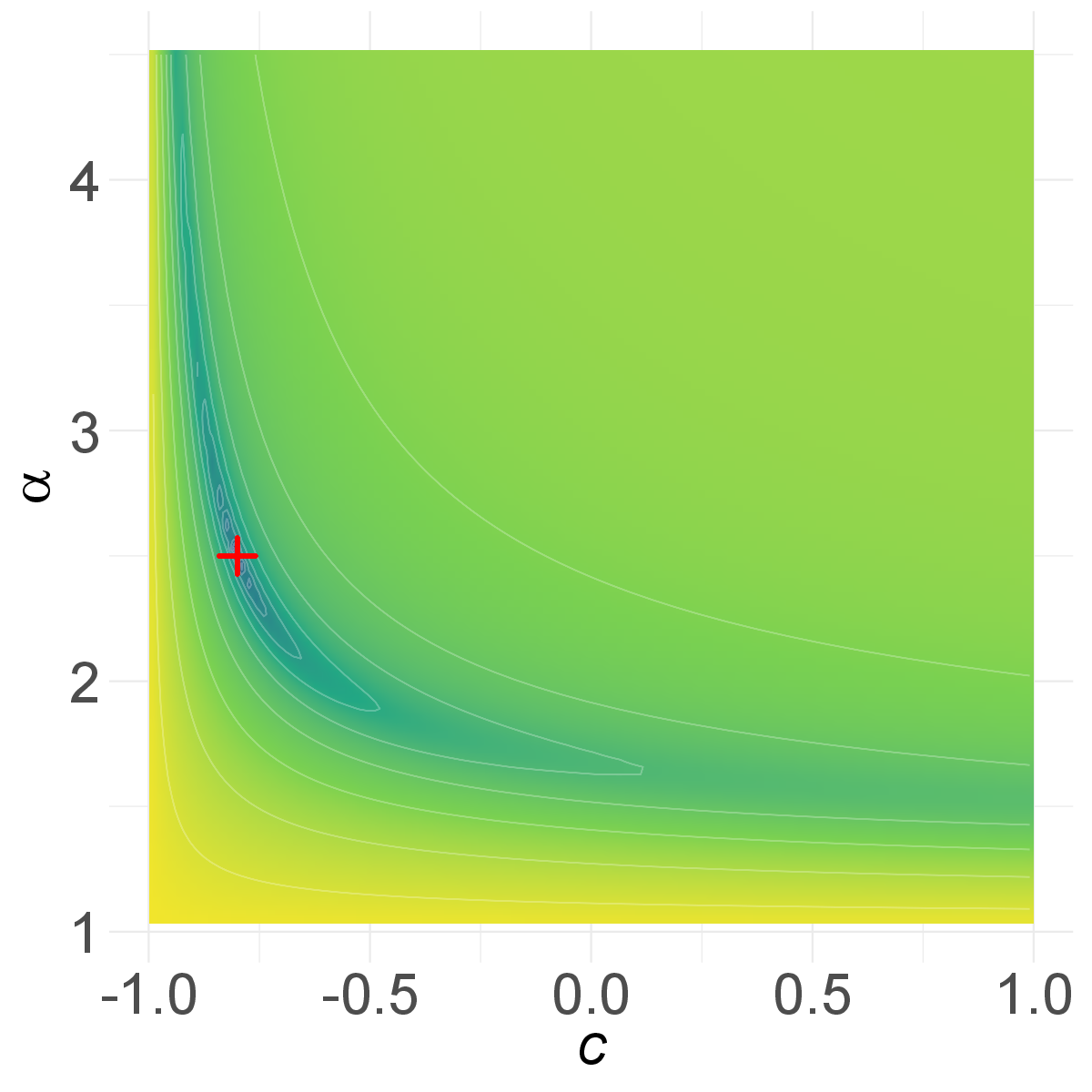}
    \caption{$\alpha_{0}=2.5,\; c_{0}=-0.8$\\(short memory)}
  \end{subfigure}%
  \hfill
  \begin{subfigure}[b]{0.23\textwidth}
    \centering
    \includegraphics[width=\textwidth]{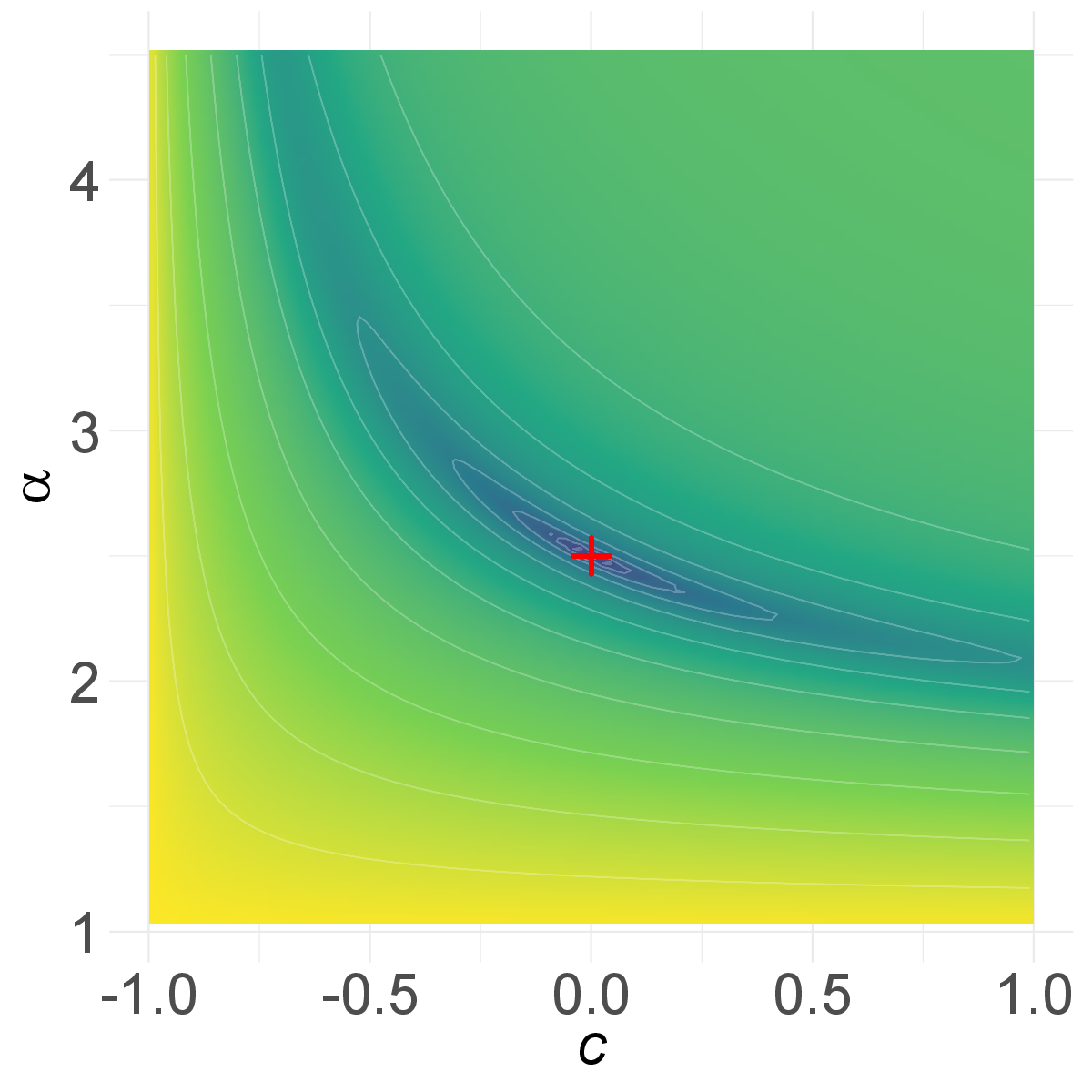}
    \caption{$\alpha_{0}=2.5,\; c_{0}=0$ \\ (no network effect)}
  \end{subfigure}%
  \hfill
  \begin{subfigure}[b]{0.23\textwidth}
    \centering
    \includegraphics[width=\textwidth]{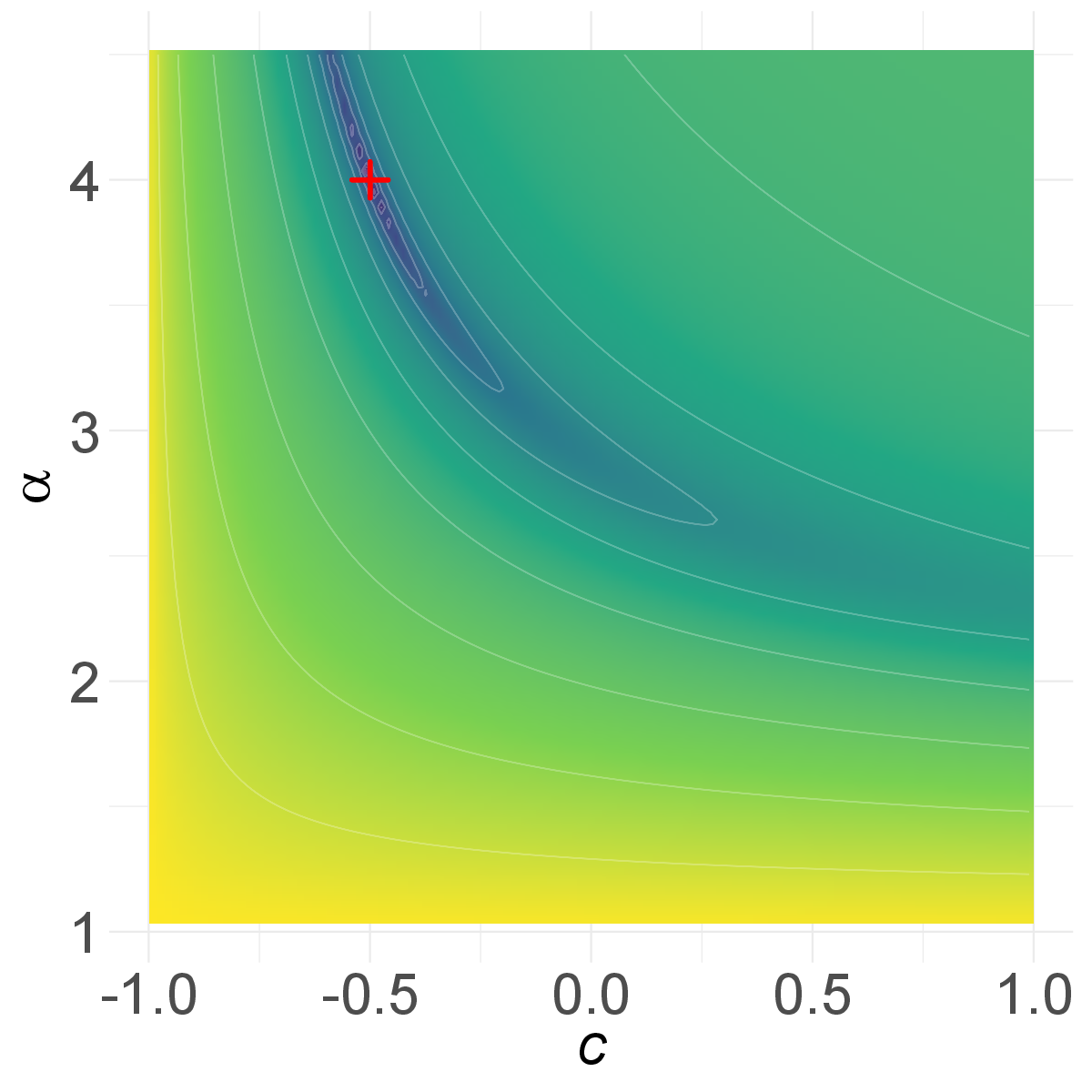}
    \caption{$\alpha_{0}=4.0,\; c_{0}=-0.5$\\(well-identified)}
  \end{subfigure}

  \vspace{0.3em}
  \includegraphics[scale=0.23]%width=0.85\textwidth]
  {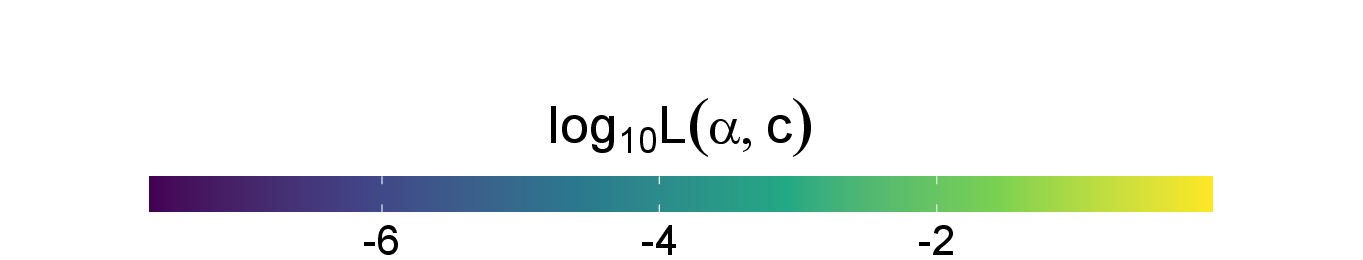}

  \caption{$\log_{10}L(\alpha,c)$ as defined in \eqref{eq:pop_loss}
    for four configurations of $(\alpha_{0},c_{0})$; red cross = true
    value. The flat regions of the loss surface indicate parameter
    combinations that are difficult to distinguish from the true
    parameter values. Here $a^{*}=1$, $\Delta=1$, and $N^{*}=30$.}
  \label{fig:loss_surface}
\end{figure}
While identifiability holds theoretically, in practice, the loss surface 
$L(\alpha,c)$ can be very flat in certain directions. This  means that 
multiple parameter pairs $(\alpha, c)$ produce nearly identical fitted 
functions $\rho(h\Delta;\cdot,\cdot)$ and are therefore difficult to 
distinguish from one another numerically.
This is most pronounced when $\alpha$ is close to 1 or $|1+ca^*|$ is
small. Figure~\ref{fig:loss_surface} illustrates this for four
representative configurations, using the population-level loss
\begin{equation}\label{eq:pop_loss}
\log_{10} L(\alpha, c) = \log_{10} \!\left( \frac{1}{N^*} \sum_{h=1}^{N^*}
\left( \rho(h\Delta;\alpha_0, c_0) - (1+(1+ca^*)h\Delta)^{1-\alpha}
\right)^{\!2} \right),
\end{equation}
with $a^{*}=1$, $\Delta=1$, and $N^{*}=30$.

The identifiability structure simplifies considerably in the 
mixture-of-exponentials case~\eqref{eq:Exp}, where the autocovariance 
function decays exponentially rather than polynomially. Specifically, a 
mixture of $n$ components has $2n-1$ free parameters $(\lambda_i, 
w_i)_{i=1}^n$ and  the parameter $c$. They are identifiable from $\rho(h\Delta;\cdot)$ at 
$2n$ distinct lags, provided the $\lambda_i$ are distinct and 
$1+ca^*\neq 0$.

In view of Proposition~\ref{prop:con}, consistency of $\widehat{\mu_L}$
and $\widehat{\boldsymbol{\sigma}^2_L}$ follows once $(\hat\theta,\hat c)$
are consistent. Estimation is most reliable in the short-memory regime
($\alpha\geq 2$) with a pronounced network effect ($|1+ca^*|$ bounded
away from zero), and we recommend choosing $N^{*}$ large relative to
the effective memory of the process
in the minimisation of~\eqref{eq:loss}.

\subsection{Generalised Method of Moments (GMM): Main results} 
We now present the key results from the generalised method of moments framework for the parameter estimation of graph supOU processes, not necessarily restricted to the setting considered in Section \ref{sec:bridge}. Full technical details and proofs are given in Appendix \ref{app:GMM}. 

Given a discrete sample of the  graph supOU process $\{X_t, t=\Delta, 2\Delta, \ldots, N\Delta\}$, for $\Delta>0, N\in\mathbb{N}$, we
set  
$X_{t:t+m}:=\left(X_{t\Delta}^{\top},\dots,X_{(t+m)\Delta}^{\top}\right)^{\top}\in \mathbb{R}^{d (m+1)}$, for $t\in\{1,\dots,N-m\}$.
We want to estimate the parameter vector that contains (potentially a subset of interest of) the
parameters of the underlying Lévy basis $\Lambda$ 
and the parameters associated with the measure $\pi$. 
We denote by $\xi$ parameter vector of interest and by $\Xi$ the corresponding parameter space. 
Let $q:=d+\frac{(m+1)d(d+1)}{2}$ denote the number of moment conditions (assumed to be bigger than the dimension of $\xi$).
We define the mean $\mu(\xi):=\mathrm{E}(X_0)$ and second moment  $\mathbf{D}_i(\xi):=\mathrm{E}(X_0X_{i\Delta}^{\top})=\mathrm{cov}(X_0,X_{i\Delta})+\mathrm{E}(X_0)(\mathrm{E}(X_0))^{\top}$, for $i=0, \ldots, m$, expressing the dependence on (a subset) of the parameter vector $\xi\in \Xi$.
Next, we  define a measurable function $f:\mathbb{R}^{d(m+1)}\times\Xi\rightarrow \mathbb{R}^{q}$ by
 \begin{align}\label{eq:f}
 f(X_{t:t+m},\xi)=\left(\begin{matrix}
            f_E(X_{t:t+m},\xi)\\
             f_0(X_{t:t+m},\xi)\\
             f_1(X_{t:t+m},\xi)\\
             \vdots\\
             f_m(X_{t:t+m},\xi)\\
         \end{matrix}\right)=\left(\begin{matrix}
      \text{vec}     (  X_{t\Delta}-\mu(\xi))\\
      \text{vec}_h      ( X_{t\Delta} X_{t\Delta}^{\top}-\mathbf{D}_0(\xi))\\
    \text{vec}_h         ( X_{t\Delta} X_{(t+1)\Delta}^{\top}-\mathbf{D}_1(\xi))\\
             \vdots\\
    \text{vec}_h          (X_{t\Delta} X_{(t+m)\Delta}^{\top}-\mathbf{D}_{m}(\xi))\\
         \end{matrix}\right),
         \end{align}
and its empirical counterpart, the  sample moment function, by 
\begin{align}\label{eq:samplem}
f_{N}(\xi):=\frac{1}{N-m}\sum_{t=1}^{N-m}f(X_{t:t+m},\xi)=\left(\begin{matrix}
           \frac{1}{N-m}\sum_{t=1}^{N-m}  f_E(X_{t:t+m},\xi)\\
           \frac{1}{N-m}\sum_{t=1}^{N-m} f_0(X_{t:t+m},\xi)\\
          \frac{1}{N-m}\sum_{t=1}^{N-m} f_1(X_{t:t+m},\xi)\\
             \vdots\\
         \frac{1}{N-m}\sum_{t=1}^{N-m}    f_m(X_{t:t+m},\xi)\\
         \end{matrix}\right).
         \end{align}
One can  estimate the true parameter $\xi_0$ by minimising  the GMM loss function, and we set
\begin{equation}{\label{eq:Estimator}}
\hat{\xi}_0^{N}:=\text{argmin}_{\xi \in \Xi}\:f_{N}(\xi)^T\mathbf{W}_Nf_{N}(\xi),\end{equation}
where $\mathbf{W}_N$ is the $q\times q$--dimensional stochastic positive definite weight matrix associated with the moment function.
 A natural choice for $\mathbf{W}_N$ is  the  inverse of a
consistent estimator of the long-run variance of $f_N(\xi_0)$. Alternatively, the identity weight matrix
$\mathbf{W}_N=\mathbf{I}$ is also a valid choice and computationally simpler, but typically less efficient.

The following two theorems establish the consistency and asymptotic normality of 
$\hat{\xi}_0^{N}$ for the concrete models of Section \ref{sec:bridge}.

\begin{theorem}[Consistency]\label{th:consistency-gamma} Suppose that the graph supOU process has finite second moments, the drift is given by \eqref{eq:Qtheta2} , for $|c|<1$. Moreover, for $\alpha>1$, we have $\theta_2\sim \Gamma(\alpha,1)$ as in \eqref{eq:PD} or $\theta_2$ is chosen as in \eqref{eq:Exp}.
Suppose that Assumptions \ref{A3} and \ref{A4} (see Appendix \ref{app:GMM}) hold. 
Then the GMM estimator $\hat{\xi}_0^{N}$,  defined in equation \eqref{eq:Estimator}, is weakly consistent, i.e.
$$
\hat{\xi}_0^{N} \stackrel{\mathrm{P}}{\longrightarrow} \xi_0, \quad \mathrm{as}\, 
N\to \infty.$$
 \end{theorem}

\begin{theorem}[Asymptotic normality]\label{th:AN-main}
Under the conditions of Theorem~\ref{th:consistency-gamma}, assume
additionally that $\int_{\|x\|>1}\|x\|^{4+\delta}\nu(dx)<\infty$ for
some $\delta>0$, and that the $\zeta$-weak dependence bound (see Appendix~\ref{app:GMM}) satisfies
$\widehat{\zeta}(r)=O(r^{-\epsilon})$ for
$\epsilon>(1+\frac{1}{\delta})\frac{3+\delta}{2+\delta}$.
Suppose that Assumptions~\ref{as:A_weightmatrix} and~\ref{as:asymcov}
hold (see Appendix~\ref{app:GMM}).
Then, as $N\to\infty$,
\begin{align*}
\sqrt{N}\bigl(\hat{\xi}_0^{N}-\xi_0\bigr)
\xrightarrow{d}
N\!\left(0,\,\mathbf{M}\mathbf{F}_{\Sigma}\mathbf{M}^{\top}\right),
\end{align*}
where
\begin{multline*}
\mathbf{F}_{\Sigma}=\sum_{l\in\mathbb{Z}}
\mathrm{cov}\!\bigl(f(X_{0:m},\xi_0),f(X_{l:l+m},\xi_0)\bigr),
\quad
\overline{\mathbf{F}}=\frac{\partial f(X_{t:t+m},\xi)}{\partial\xi^{\top}}
\bigg|_{\xi=\xi_0},
\quad\\
\mathbf{M}=\bigl(\overline{\mathbf{F}}^{\top}\mathbf{W}
\overline{\mathbf{F}}\bigr)^{-1}\overline{\mathbf{F}}^{\top}\mathbf{W}.
\end{multline*}
\end{theorem}

\begin{remark}[Scope of asymptotic normality]\label{rem:longmem-excl}
In the Gamma-law case, $\widehat{\zeta}(r)=O(r^{-\alpha/2})$, so the
condition on $\epsilon$ requires
$\alpha>2\!\left(1+\frac{1}{\delta}\right)\!\frac{3+\delta}{2+\delta}$,
which implies $\alpha>2$.
Consequently, \emph{asymptotic normality currently covers only the
short-memory regime}; the long-memory case
($\alpha\in(1,2)$) remains an open problem and is deferred to future
work.  The mixture-of-exponentials model~\eqref{eq:Exp}, which has
exponentially decaying dependence, is fully covered by
Theorem~\ref{th:AN-main}.
\end{remark}

\section{Simulation study}\label{sec:sim}
We assess the finite-sample performance of our proposed estimation methodology in a Monte Carlo study.
Adapting the ideas presented in 
\cite{FK2007, RSM}, we will first present a simulation algorithm for graph supOU processes. 
 We note that the 
 L\'{e}vy-It\^{o} decomposition implies that the graph supOU process has a modification of the form 
\begin{align*}
Y_t&= \gamma 
\int_{M_d^-}\int_{-\infty}^t e^{\mathbf{Q}(\theta)(t-s)} ds\pi(d\mathbf{Q}(\theta))
+ \int_{M_d^-}\int_{-\infty}^t e^{\mathbf{Q}(\theta)(t-s)}W(d\mathbf{Q}(\theta), ds)
\\
&
+ \int_{x\in \mathbb{R}^d: \|x\|\leq 1}\int_{M_d^-}\int_{-\infty}^t e^{\mathbf{Q}(\theta)(t-s)}x[\mu(dx, d\mathbf{Q}(\theta), ds)-ds\pi(d\mathbf{Q}(\theta))\nu(dx)]\\
&
+ \int_{x\in \mathbb{R}^d: \|x\|> 1}\int_{M_d^-}\int_{-\infty}^t e^{\mathbf{Q}(\theta)(t-s)}x\mu(dx, d\mathbf{Q}(\theta), ds)\\
&=:D_t+G_t+S_t+J_t,
\end{align*}
corresponding to the drift term $D_t$, the Gaussian term $G_t$, the term involving the  compensated small jumps $S_t$ and the big jumps $J_t$.
Since there is extensive literature on the simulation of Gaussian processes, we will focus on the jump case from now on. More precisely, we consider the scenario where the L\'{e}vy basis $\Lambda$ is given by a compound Poisson random measure. In that case, we can write, for $t\geq 0$,
\begin{align*}
    X_t &=
    \int_{x\in \mathbb{R}^d}\int_{M_d^-}\int_{-\infty}^t e^{\mathbf{Q}(\theta)(t-s)}x \mu(dx, d\mathbf{Q}(\theta), ds)  \\
    &
    =    
    %X_t^-
    \sum_{i\in \mathbb{N}}e^{\mathbf{Q}(\theta)_{-i}(t+\tau_{-i})}U_{-i}
    +\sum_{i\in \mathbb{N}: \tau_i\leq t}e^{\mathbf{Q}(\theta)_i(t-\tau_i)}U_i, 
\end{align*}
where $\tau_i$ are the arrival times of a homogeneous Poisson process of rate
$\varrho =\nu(\mathbb{R})$, i.e.
%\begin{align*}
    $\tau_i:=\sum_{j=1}^i T_j$, and
    %\quad     \mathrm{and}\quad 
    $\tau_{-i}:=\sum_{j=1}^i T_{-j}, \; \forall i \in \mathbb{N}$,
%\end{align*}
where $(T_i)_{i\in \mathbb{Z}\setminus\{0\}}$, $(U_i)_{i\in \mathbb{Z}\setminus\{0\}}$ and $(\mathbf{Q}(\theta))_{i\in \mathbb{Z}\setminus\{0\}}$ are mutually independent and i.i.d.~sequences with   $T_i\sim \mathrm{Exp}(\varrho)$, 
$\mathbf{Q}(\theta)_i\sim \pi$ and  $U_i\sim \nu/\varrho$.

In the following, we consider the concrete parametric model considered in Section \ref{sec:bridge} and Equation \eqref{eq:PD}, where  
\begin{align*}
    X_t &=
    \sum_{i\in \mathbb{N}}e^{\theta_{2,-i}\mathbf{K}(c)(t+\tau_{-i})}U_{-i} +\sum_{i\in \mathbb{N}: \tau_i\leq t}e^{\theta_{2,i}\mathbf{K}(c)(t-\tau_i)}U_i, 
\end{align*}
and we choose $\varrho =5$, $c=-0.8$, 
$\theta_{2,i}\sim \Gamma(\alpha,1)$ with  $\alpha =1.5$ (which is the long memory setting), and $U_i\sim \mathrm{N}(0, \mathbf{I}_{d\times d})$. 
I.e.~$\mu_L=0, {\boldsymbol \sigma}^2_L=5\mathbf{I}_{d\times d}$.
We consider the  network used in our subsequent empirical study, see Figure \ref{fig:MapData}, which has $d=24$ nodes. Note that the values chosen for $c, \alpha$ are also motivated by our subsequent empirical study.
Further, we set $\Delta=1$,
$N=1000$ observations per time series and we choose 
1000 Monte Carlo runs.

We present two approaches to simulating such a graph supOU process at times
$t_j = j\Delta$, for $j = 0, \ldots, N-1$ and $\Delta > 0$.

\paragraph{Exact simulation.}
We can simulate directly from the representation
\begin{align*}
    X_t = \sum_{i\in \mathbb{N}} X_t^{(-i)}
        + \sum_{i\in \mathbb{N}:\, \tau_i\leq t} X_t^{(i)},
\end{align*}
where $X_t^{(i)} = e^{\mathbf{Q}(\theta)_i(t-\tau_i)}U_i$ and
$X_t^{(-i)} = e^{\mathbf{Q}(\theta)_{-i}(t+\tau_{-i})}U_{-i}$.
Each contribution requires evaluating the matrix exponential
$e^{\mathbf{Q}(\theta)_i s}$ for a $d\times d$ matrix, which costs $O(d^3)$ per
call. With an expected $\varrho N$ jumps, each contributing to all subsequent
time steps, the total number of matrix-exponential evaluations is
$O(\varrho N^2)$, giving an overall cost of $O(\varrho N^2 d^3)$, which becomes
prohibitive for large $d$ or $N$.

\paragraph{Accelerated exact simulation.}
Since $\mathbf{Q}(\theta)_i = \theta_{2,i}\mathbf{K}(c)$ and $\mathbf{K}(c)$ is
fixed across all jumps, we compute its eigendecomposition
\begin{equation*}
    \mathbf{K}(c) = \mathbf{V}\,\mathrm{diag}(k_1,\ldots,k_d)\,\mathbf{V}^{-1}
\end{equation*}
once at the start of the simulation. Using the standard identity
$e^{\mathbf{V}\mathbf{D}\mathbf{V}^{-1}}=\mathbf{V}e^{\mathbf{D}}\mathbf{V}^{-1}$,
the exact contribution of jump $i$ to time point $t_j$ becomes
\begin{align*}
e^{\theta_{2,i}\mathbf{K}(c)\,s_{ij}}\,U_i
&= \mathbf{V}\,\mathrm{diag}\!\left(e^{\theta_{2,i}k_1 s_{ij}},\ldots,
   e^{\theta_{2,i}k_d s_{ij}}\right)\underbrace{\mathbf{V}^{-1}U_i}_{=:\,\widetilde{U}_i}
 = \mathbf{V}\left(\mathbf{e}_{ij}\odot\widetilde{U}_i\right)
 =: \mathbf{V}\,\mathbf{c}_{ij},
\end{align*}
where $s_{ij}:=t_j-\tau_i$, $\mathbf{e}_{ij}:=\left(e^{\theta_{2,i}k_1
s_{ij}},\ldots,e^{\theta_{2,i}k_d s_{ij}}\right)^\top$, and $\odot$ denotes
elementwise multiplication. The vector $\widetilde{U}_i:=\mathbf{V}^{-1}U_i$
is the jump size $U_i$ expressed in the eigenbasis of $\mathbf{K}(c)$. Since
$\mathbf{V}$ does not depend on $i$, it is computed once per jump and reused
across all future time points. 
The process is then recovered as
\begin{equation}
    X_{t_j} = \sum_{i:\, \tau_i \leq t_j}
        \mathrm{Re}\!\left(\mathbf{V}\,\mathbf{c}_{ij}\right), \qquad
        j=1,\ldots,N,
        \label{eq:fastsim}
\end{equation}
where $\mathrm{Re}(\cdot)$ denotes the real part of a complex number. Applying the real part above discards residual imaginary components from
 rounding. 
Each term costs $O(d)$ to form $\mathbf{e}_{ij}$,
$O(d)$ for the elementwise product, and $O(d^2)$ for the matrix-vector product
$\mathbf{V}\mathbf{c}_{ij}$, replacing the original $O(d^3)$ matrix-exponential
call. The total cost reduces from $O(\varrho N^2 d^3)$ to
$O(d^3 + \varrho N^2 d^2)$, which is an $O(d)$ speedup. For our parameter choices
($d=24$, $\varrho=5$, $N=1000$), this amounts to a factor of approximately $24$.
This is the scheme used in our simulation study.

In practice, the stationary initialisation corresponding to the $X_t^-$ term
is approximated by adding a burn-in period and discarding
those observations.

We summarise the findings from our simulation experiment in Figure \ref{fig:simresults}.
\begin{figure}[htb]
    \centering
    % First row
        \begin{subfigure}[b]{0.3\textwidth}
        \centering
        \includegraphics[scale=0.2]{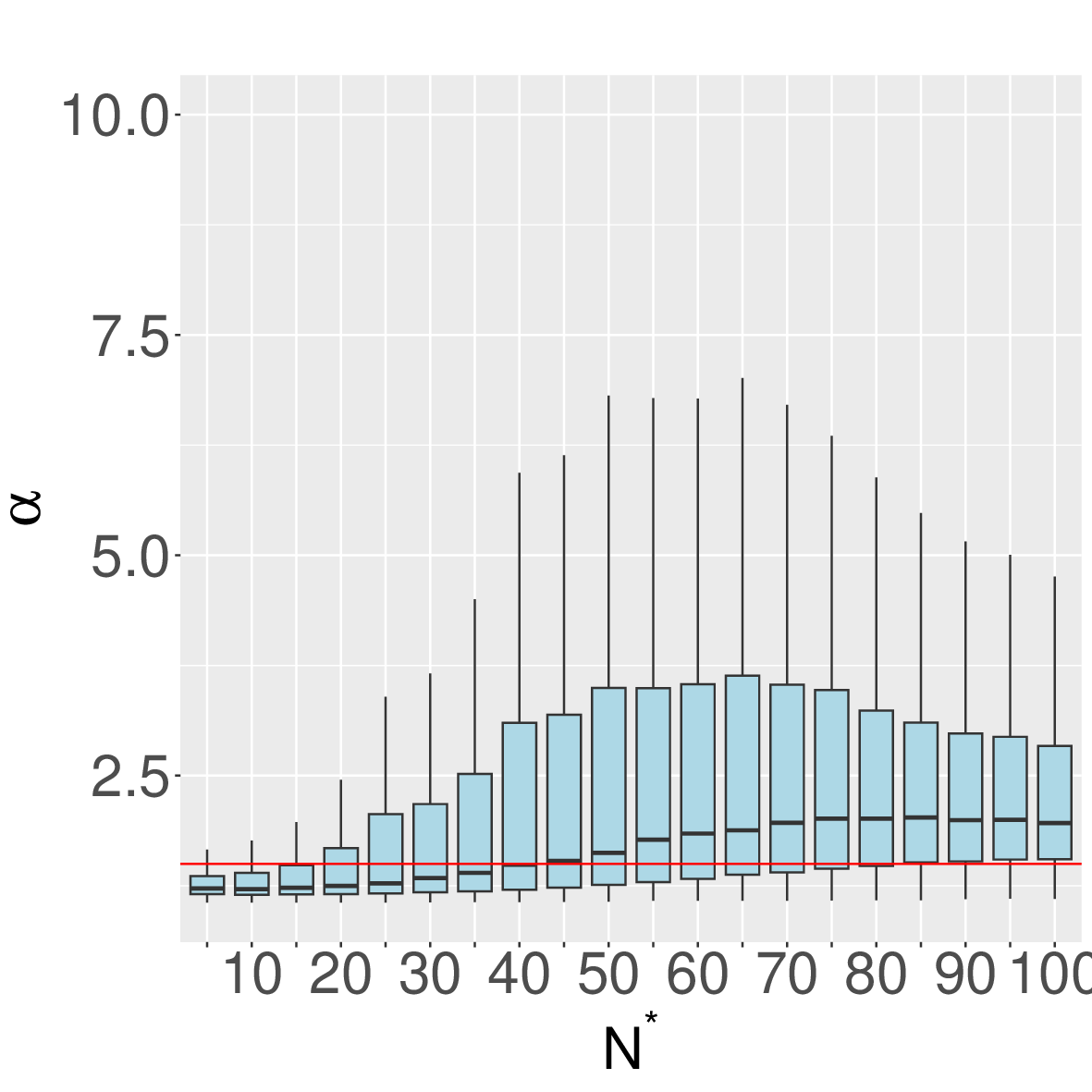}
        \caption{}
        \label{fig:sim_alpha}
    \end{subfigure}    
    %\hfill
    \begin{subfigure}[b]{0.3\textwidth}
        \centering
        \includegraphics[scale=0.2]{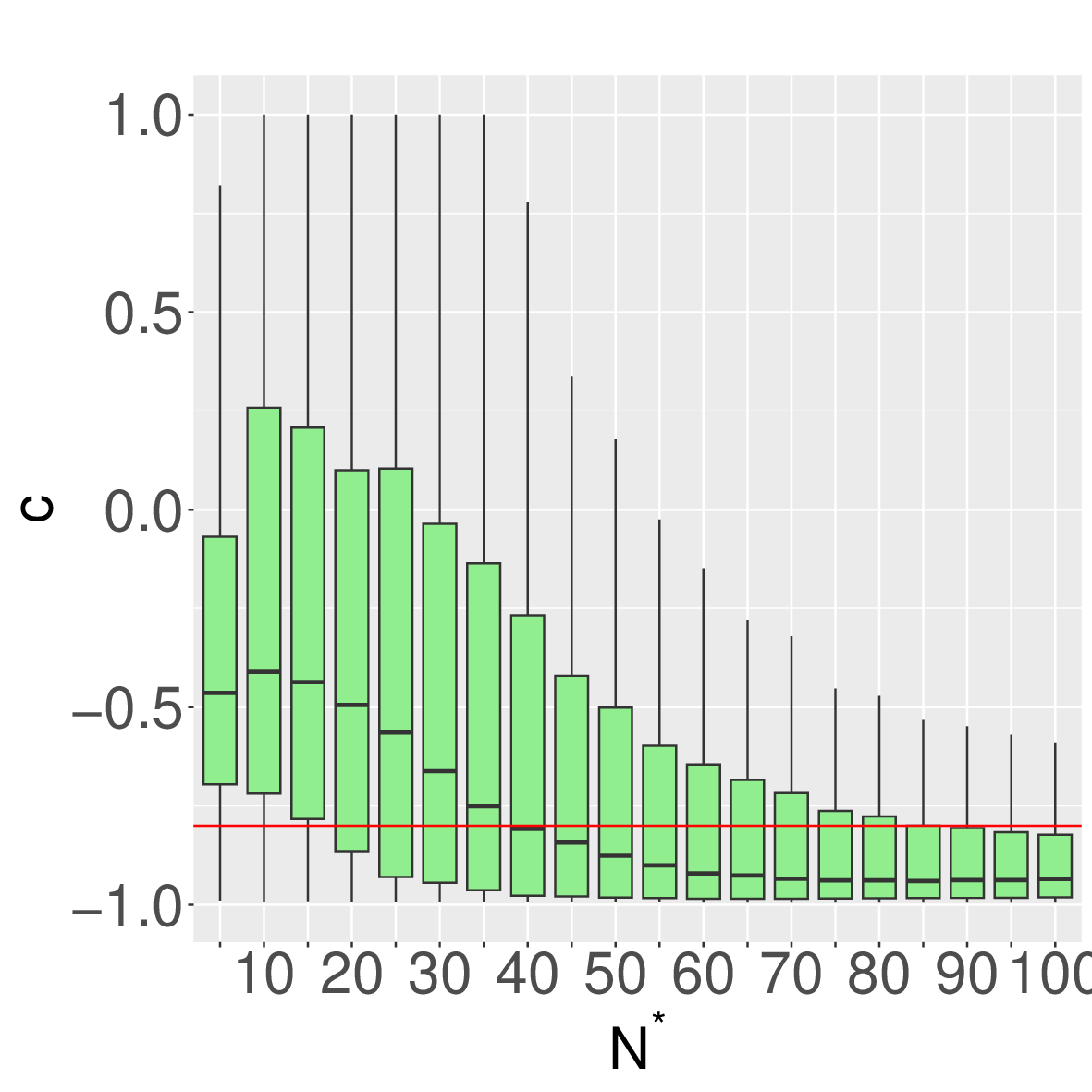}
        \caption{}
        \label{fig:sim_c}
    \end{subfigure}
    %\hfill
    \begin{subfigure}[b]{0.3\textwidth}
        \centering
        \includegraphics[scale=0.2]{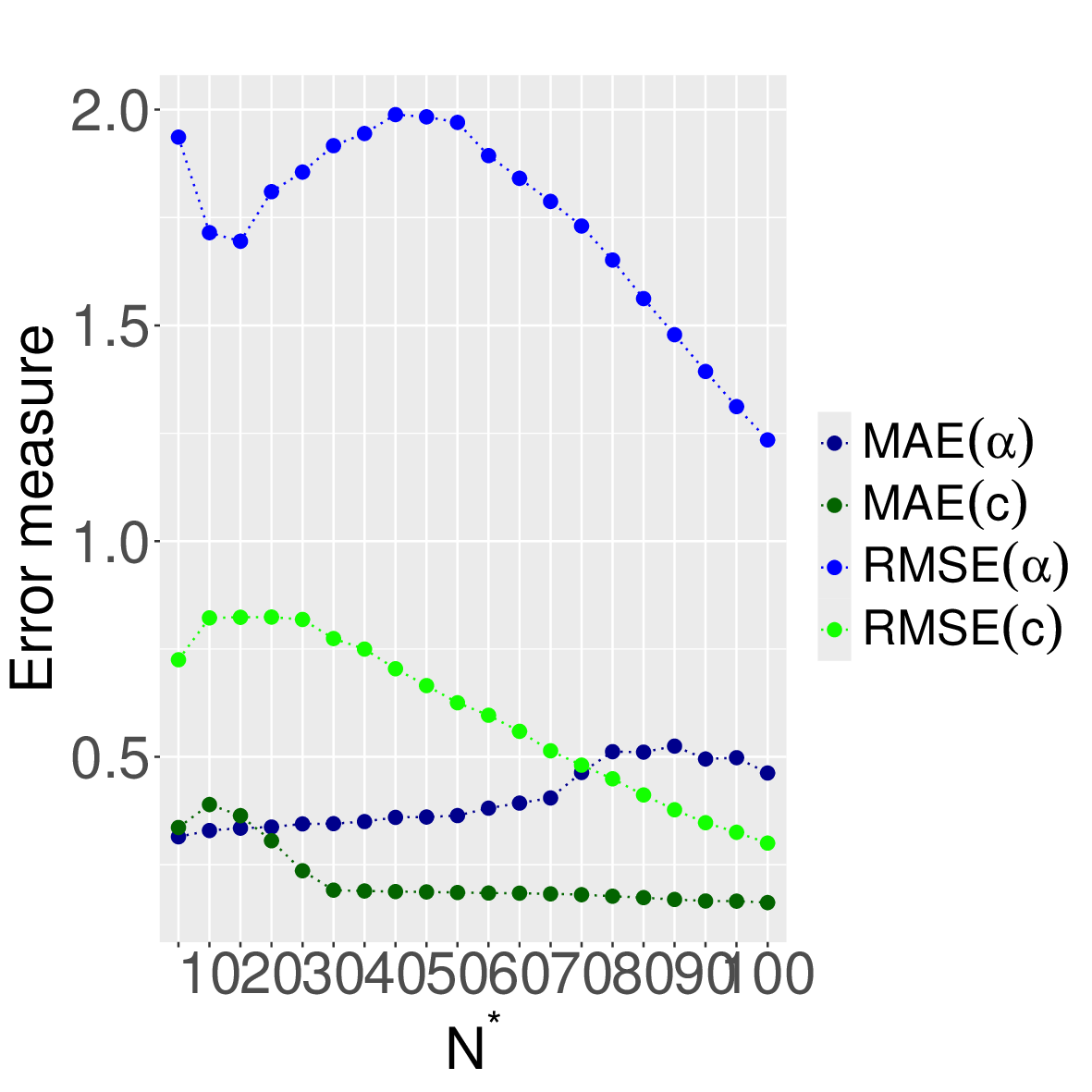}
        \caption{}
        \label{fig:sim-error}
    \end{subfigure}
    %second row
    \begin{subfigure}[b]{0.3\textwidth}
        \centering
        \includegraphics[scale=0.2]{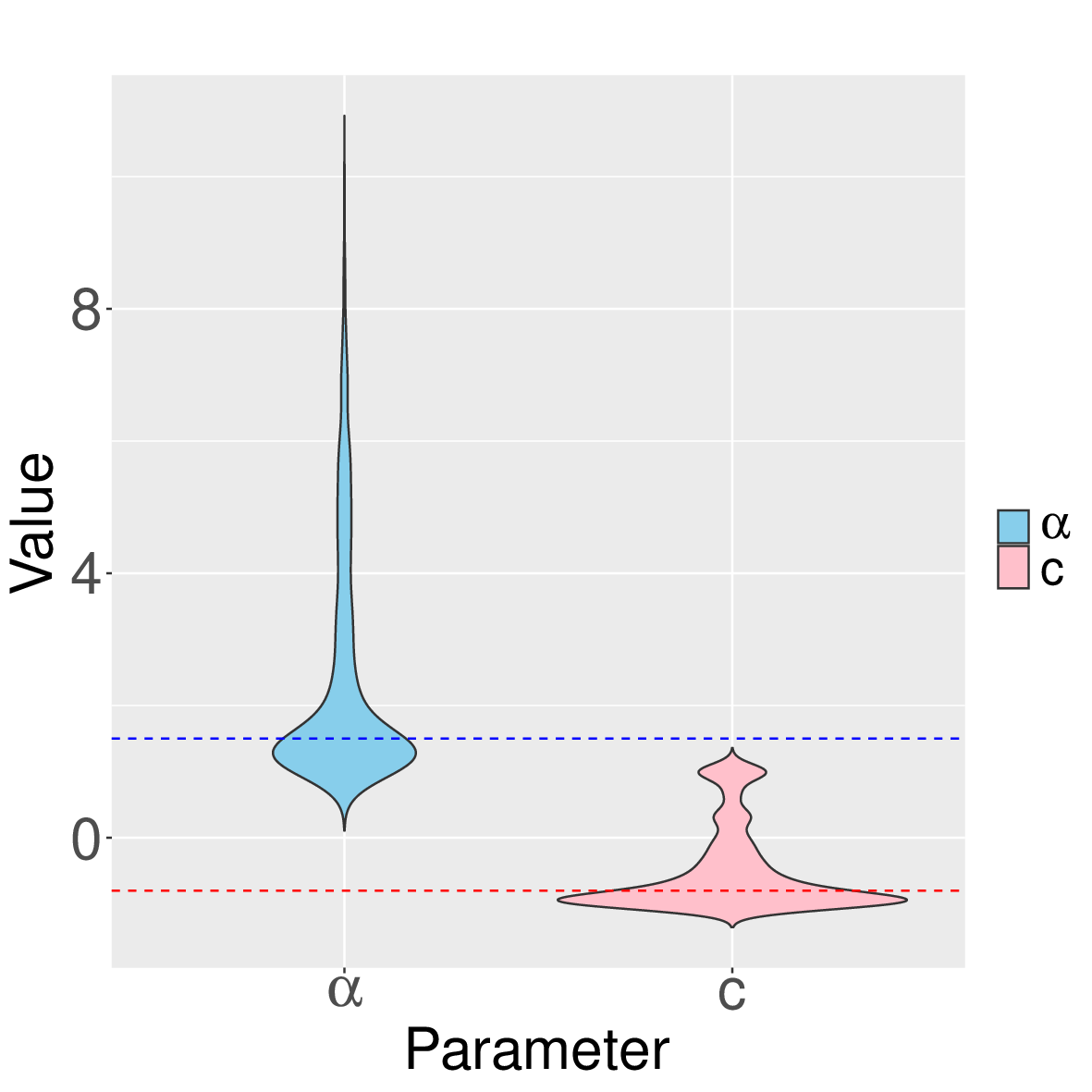}
        \caption{}
        \label{fig:p_vio}
    \end{subfigure}    
   % \hfill
    \begin{subfigure}[b]{0.3\textwidth}
        \centering
        \includegraphics[scale=0.2]{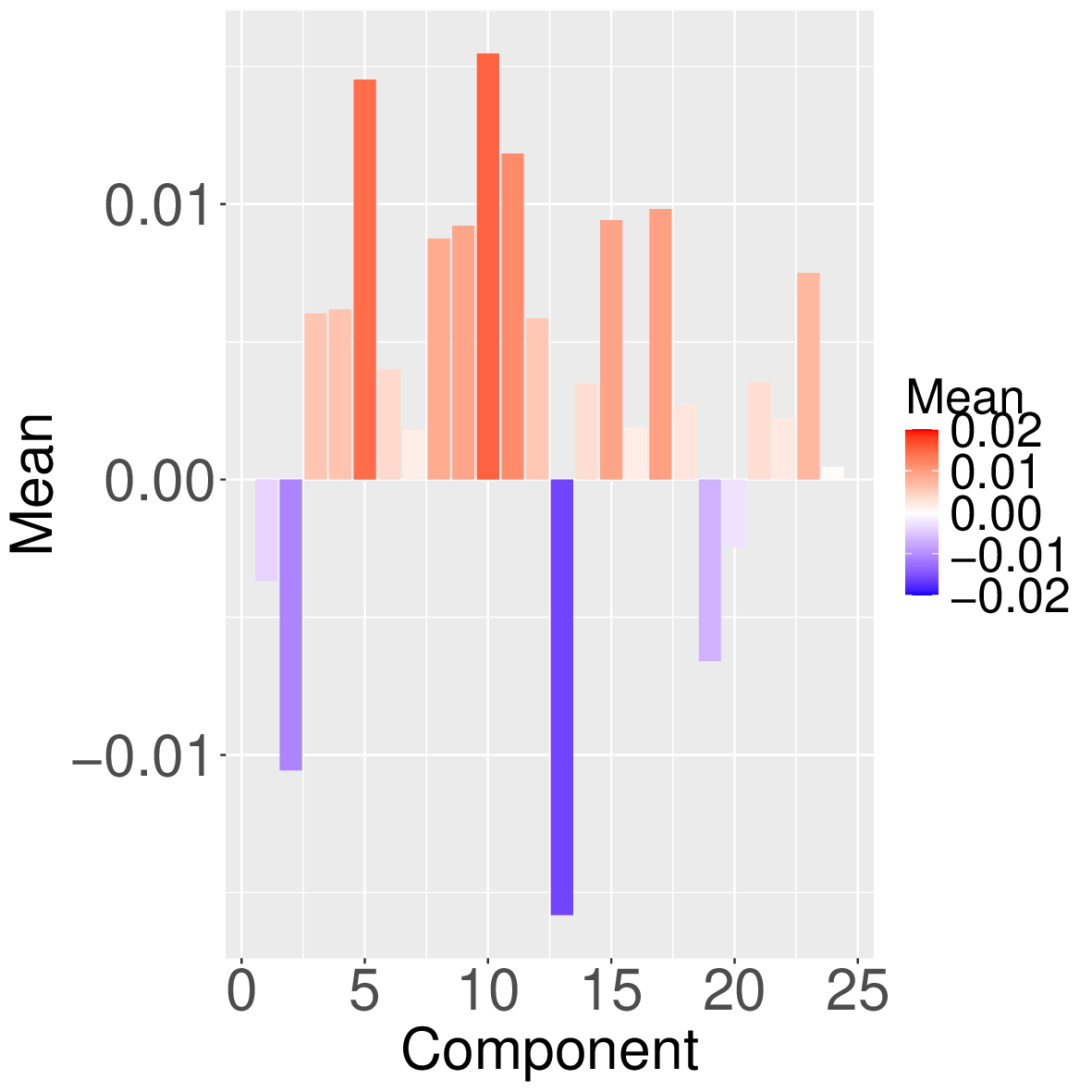}
        \caption{}
        \label{fig:estmean_median_sim}
    \end{subfigure}
    %\hfill
    \begin{subfigure}[b]{0.3\textwidth}
        \centering
        \includegraphics[scale=0.2]{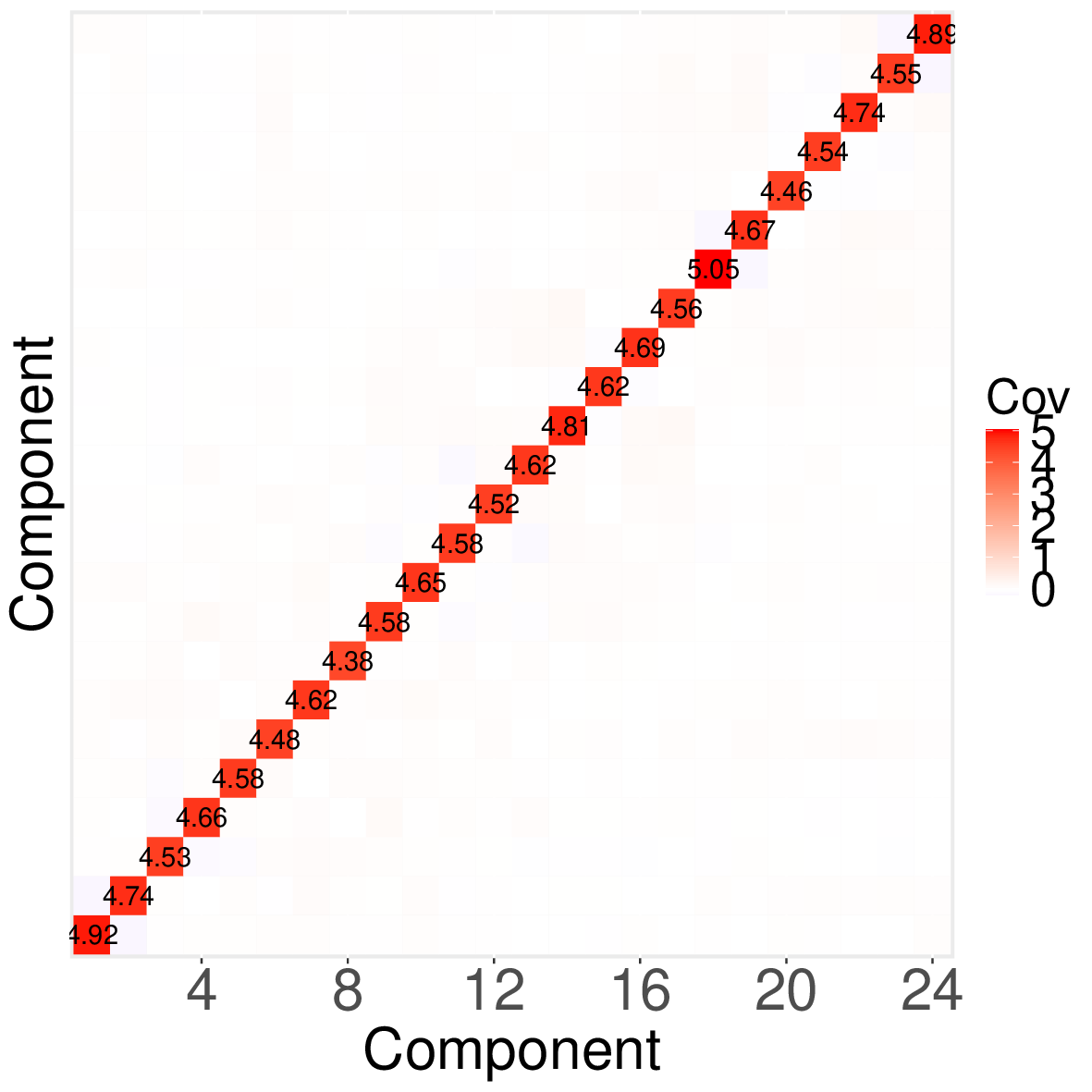}
        \caption{}
        \label{fig:estvar_median_sim}
    \end{subfigure}
             \caption{Results from the simulation study: Figures
             \ref{fig:sim_alpha} and \ref{fig:sim_c} depict the boxplots for the estimates of $\alpha$ and $c$, respectively, for different choices of $N^*$  in the loss function. The outliers are not depicted to improve readability. The solid red line presents the true parameter value in both cases. Figure
              \ref{fig:sim-error} displays the corresponding error measures. Figure             
             \ref{fig:p_vio} displays the violin plots for the estimates of $\alpha$ and $c$ for fixed $N^*=40$. Finally, Figures          \ref{fig:estmean_median_sim}
             and \ref{fig:estvar_median_sim} display the heatmap for the medians of the estimated means and variances over the 1000 Monte Carlo runs.
              }
    \label{fig:simresults}
\end{figure}
Recall that we  need to select the tuning parameter $N^*$ in the quadratic loss function, \eqref{eq:loss}, which determines how many lags of the  scaled autocovariance \eqref{eq:EmpR} enter our estimation. In our simulation study, we tested various choices of $N^*$, where we set $\Delta=1$ and $N^* \in \{5, 10,  \ldots 100\}$. For each of the $N^*$, we estimate $\alpha$ and $c$, and we provide the boxplots of the distribution of these estimates over the 1000 Monte Carlo runs in 
Figures  \ref{fig:sim_alpha} and \ref{fig:sim_c}, respectively. 
Figure  \ref{fig:sim-error} complements this part of the study by providing the root mean squared errors and the median absolute errors for the estimated parameters. All three plots together suggest that $N^*=40$ provides a reliable choice, balancing estimation accuracy against the increased variance introduced by longer lags. This is further confirmed by the scatter plots of the estimated parameters provided $(\hat{\alpha}, \hat{c})$ in Figure \ref{fig:simresults-contours1}, where the estimates concentrate around the true value once $N^*$ reaches  40.
 We then zoom into this scenario in Figure   \ref{fig:p_vio}, which contains the violin plots for the estimates of $\alpha$ and $c$ for $N^*=40$. The dotted lines show the true value. Overall, we see that the parameters can be well estimated by our method. 
As a final step, we estimate the mean and the variance, i.e.~$\mu_L=0, {\boldsymbol \sigma}^2_L=5\mathbf{I}_{d\times d}$, of the underlying L\'{e}vy basis. We carry out the estimation for each of the 1000 Monte Carlo runs and depict the median of  the estimated  mean value $\mu_L$, which is close to $0$ as it should be, and of the variance matrix, which is also close to its true counterpart given by ${\boldsymbol \sigma}^2_L=5\mathbf{I}_{d\times d}$.

\begin{figure}[htb]
    \centering
    % First row
        \begin{subfigure}[b]{0.3\textwidth}
        \centering
        \includegraphics[scale=0.2]{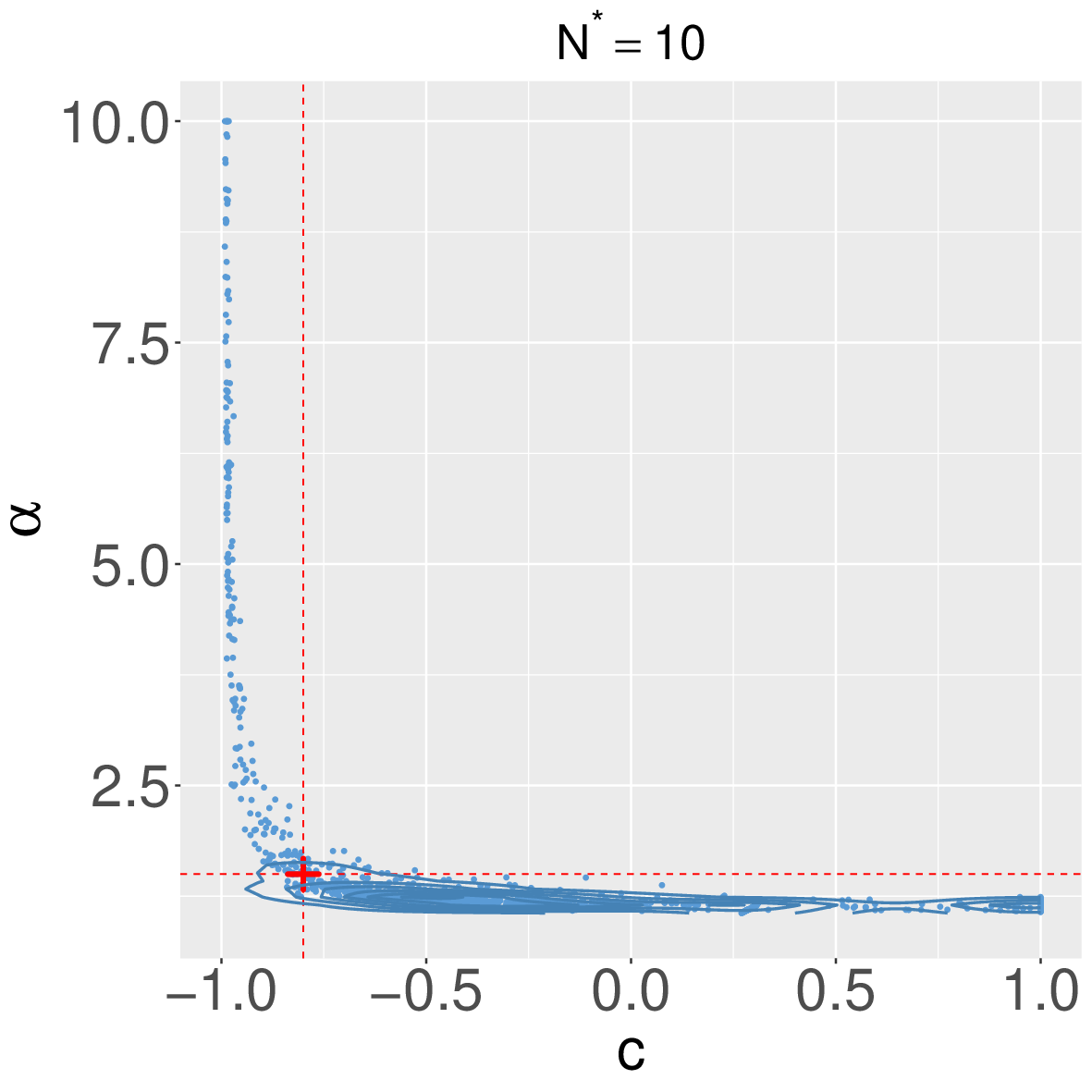}
        \caption{}
        \label{fig:sim_pairs_lag10}
    \end{subfigure}    
    %\hfill
    \begin{subfigure}[b]{0.3\textwidth}
        \centering
        \includegraphics[scale=0.2]{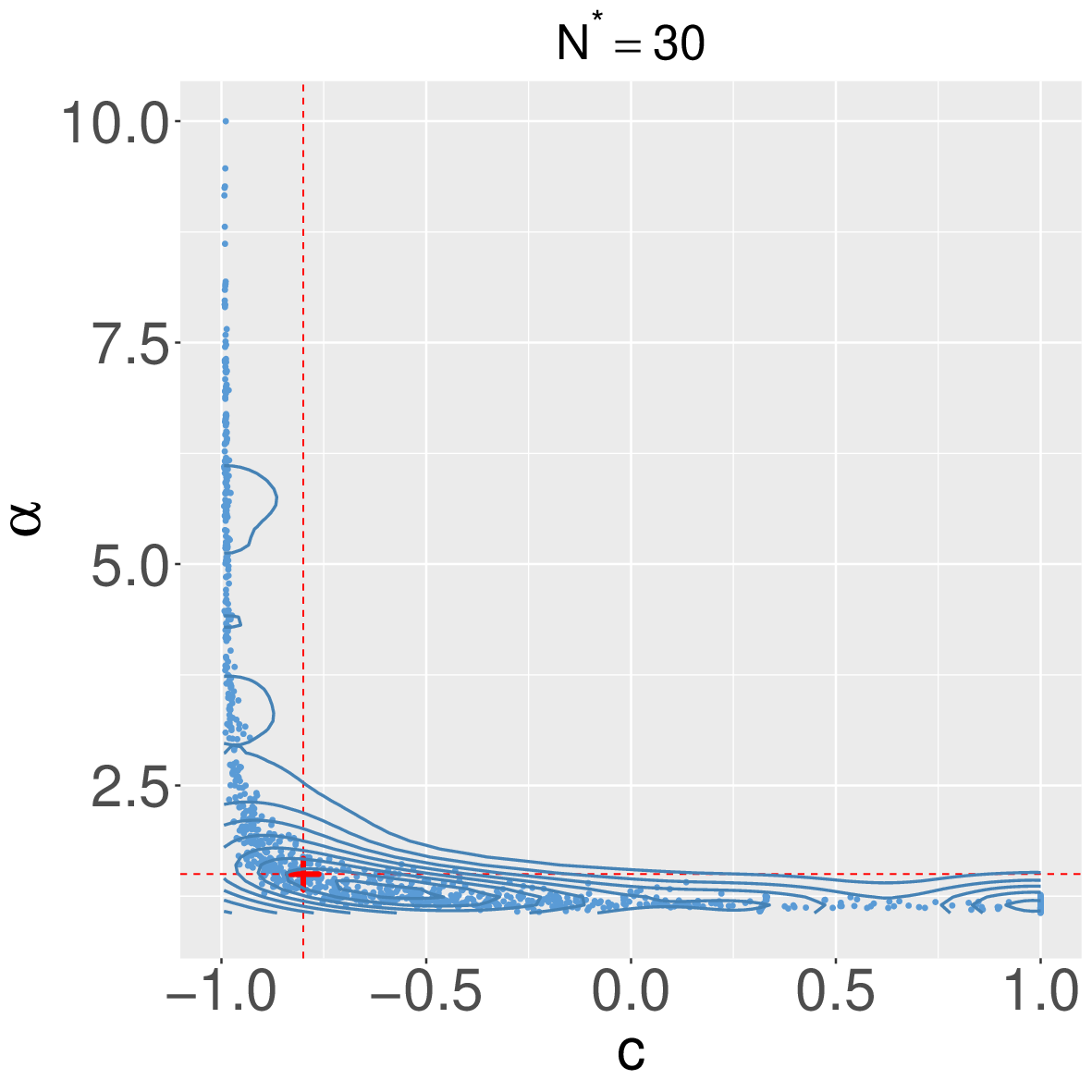}
        \caption{}
        \label{fig:sim_pairs_lag30}
    \end{subfigure}
    %\hfill
    \begin{subfigure}[b]{0.3\textwidth}
        \centering
        \includegraphics[scale=0.2]{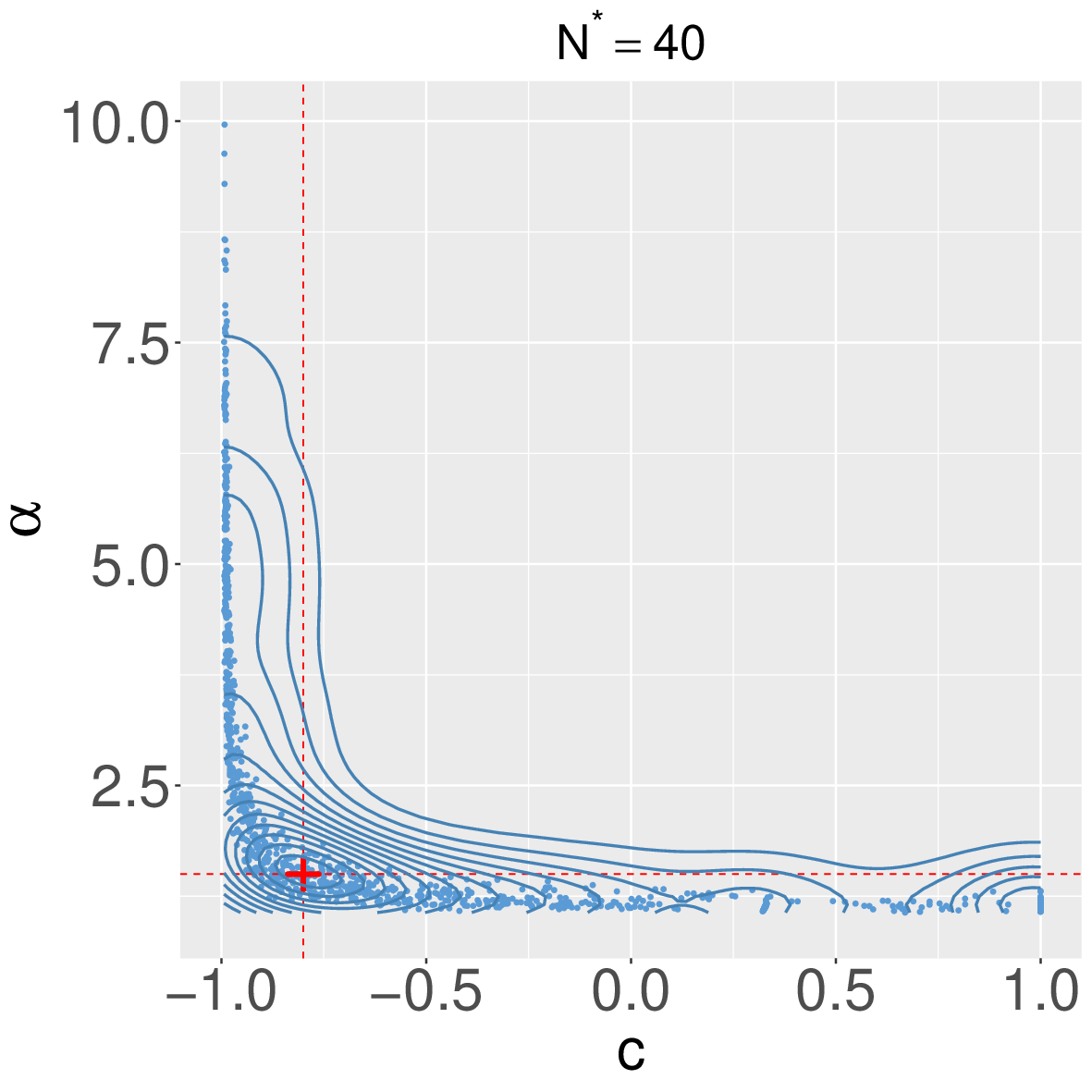}
        \caption{}
        \label{fig:sim_pairs_lag40}
    \end{subfigure}
    %second row
        \begin{subfigure}[b]{0.3\textwidth}
        \centering
        \includegraphics[scale=0.2]{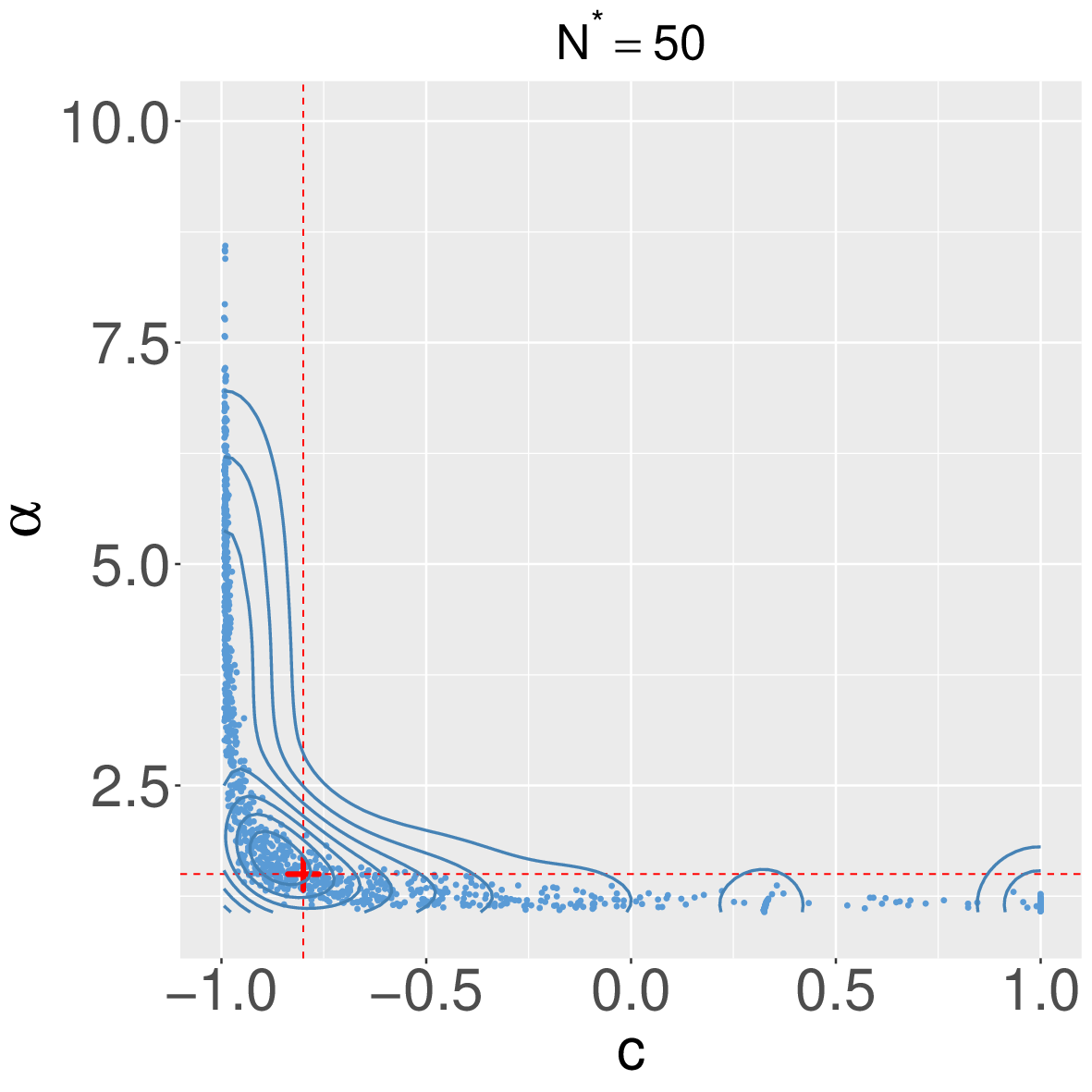}
        \caption{}
        \label{fig:sim_pairs_lag50}
    \end{subfigure}
    %\hfill
    \begin{subfigure}[b]{0.3\textwidth}
        \centering
        \includegraphics[scale=0.2]{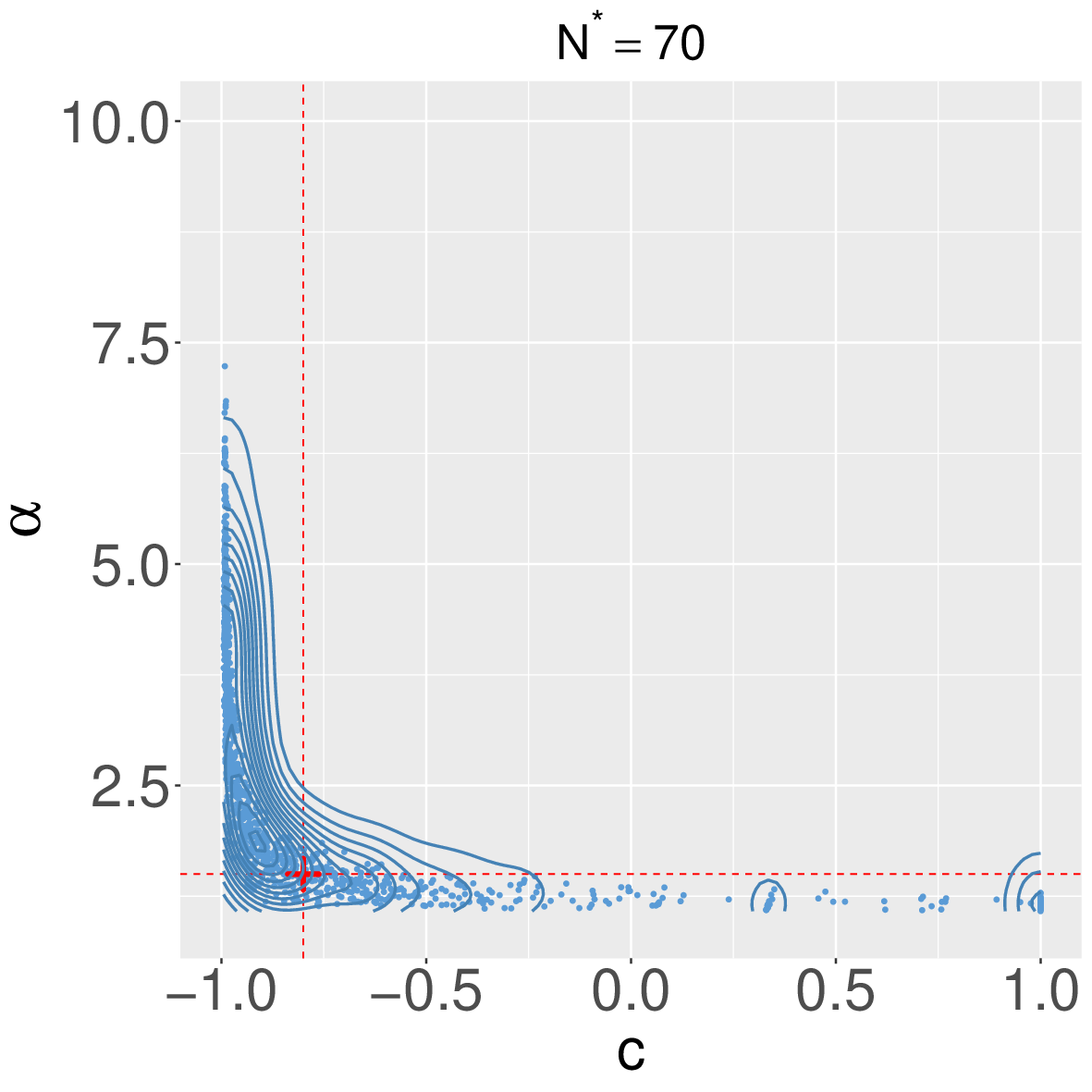}
        \caption{}
        \label{fig:sim_pairs_lag70}
    \end{subfigure}
    \begin{subfigure}[b]{0.3\textwidth}
        \centering
        \includegraphics[scale=0.2]{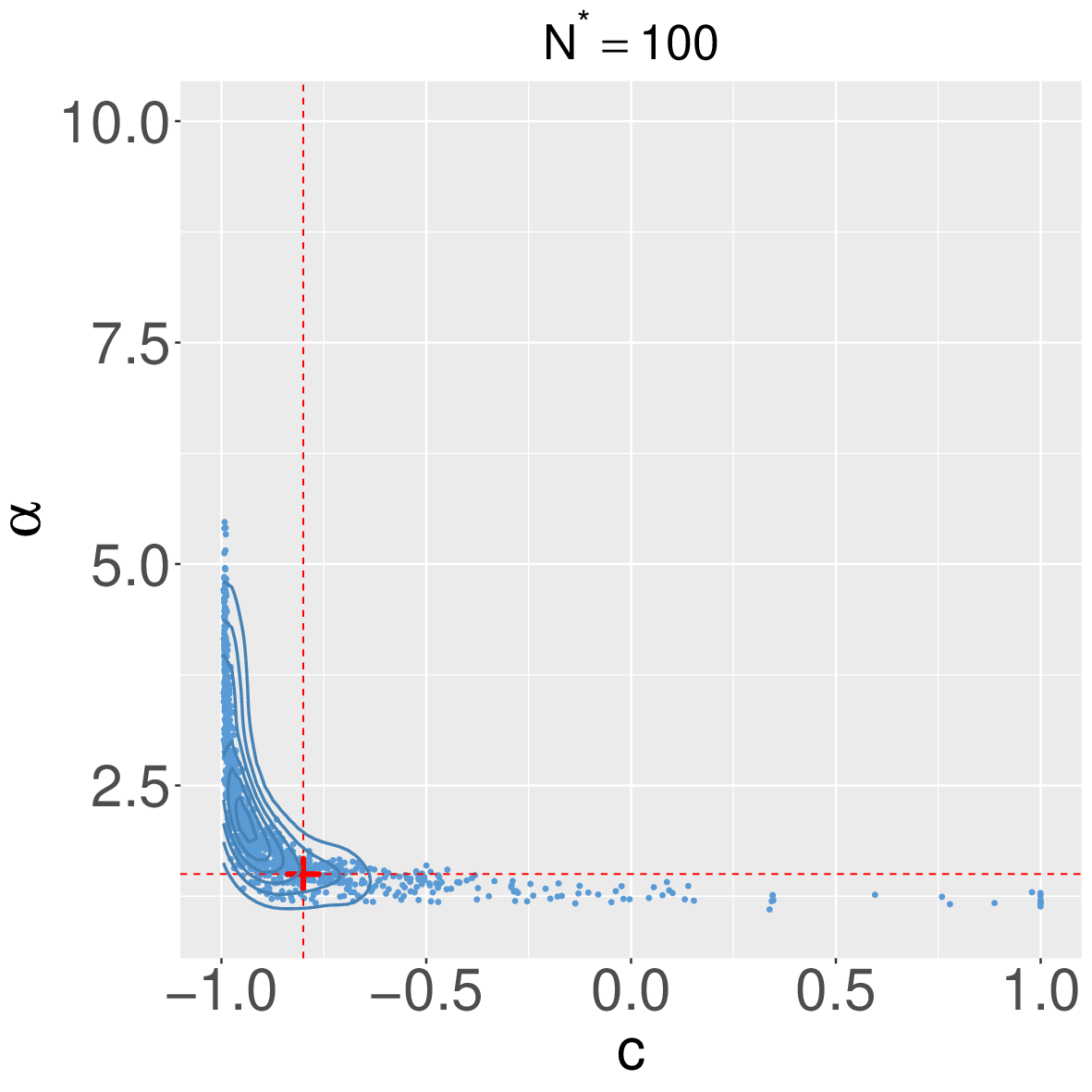}
        \caption{}
        \label{fig:sim_pairs_lag100}
    \end{subfigure}    
   % \hfill
             \caption{
    Scatter plots of the estimated parameters $(\hat{\alpha}, \hat{c})$ across 1000 Monte Carlo runs, for different numbers of lags $N^* \in \{10, 30, 40, 50, 70, 100\}$ used in the loss function. The red cross marks the true value $(\alpha_0, c_0) = (1.5, -0.8)$. With few lags, the estimates are scattered along a curve, but they concentrate around the true value once $N^*$ reaches  40.
              }
    \label{fig:simresults-contours1}
\end{figure}

%%%%%%%%%%%%%%%%%%%%%%%%%%%%%%%%%%%%%%
%%%%%%%%%%%%%%%%%%%%%%%%%%%%%%%%%%%%%%%%
%%%%%%%%%%%%%%%%%%%%%%%%%%%%%%%%%%%%%%%%%
\section{Empirical study}\label{sec:emp}
We will now illustrate our new model and estimation methodology in an empirical study. For this, we will focus on the RE Europe dataset, see \cite{JP2017}, which is a publicly available large-scale dataset, designed for research into a highly renewable European electricity system.
The dataset includes the network model of 
1,494 buses (nodes) in Europe. 
We will use a part of this dataset and focus on 24 nodes of the network in Portugal, which we depict in Figure \ref{fig:MapData}, see 
Figure \ref{fig:map} in particular, and Figure \ref{fig:normA} shows the corresponding heatmap of the column normalised adjacency matrix $\overline{A}$.
We consider an undirected network in the following. 

\begin{figure}[htbp]
    \centering
    % First row
        \begin{subfigure}[b]{0.3\textwidth}
        \centering
        \includegraphics[trim= 150 0 150 0,clip, scale=0.3]{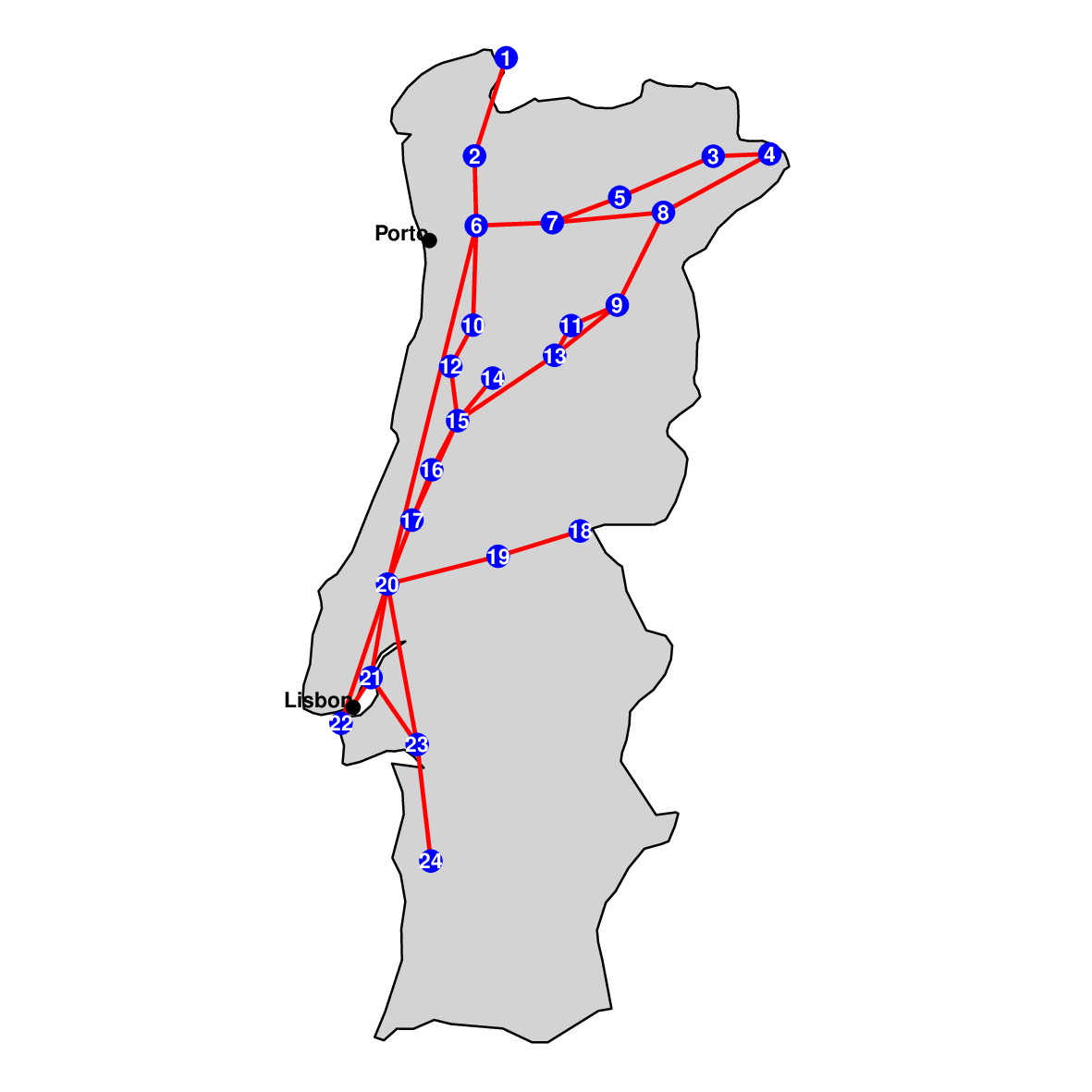}
        \caption{}
        \label{fig:map}
    \end{subfigure}    
    %\hfill
    \begin{subfigure}[b]{0.4\textwidth}
        \centering
        \includegraphics[scale=0.3]{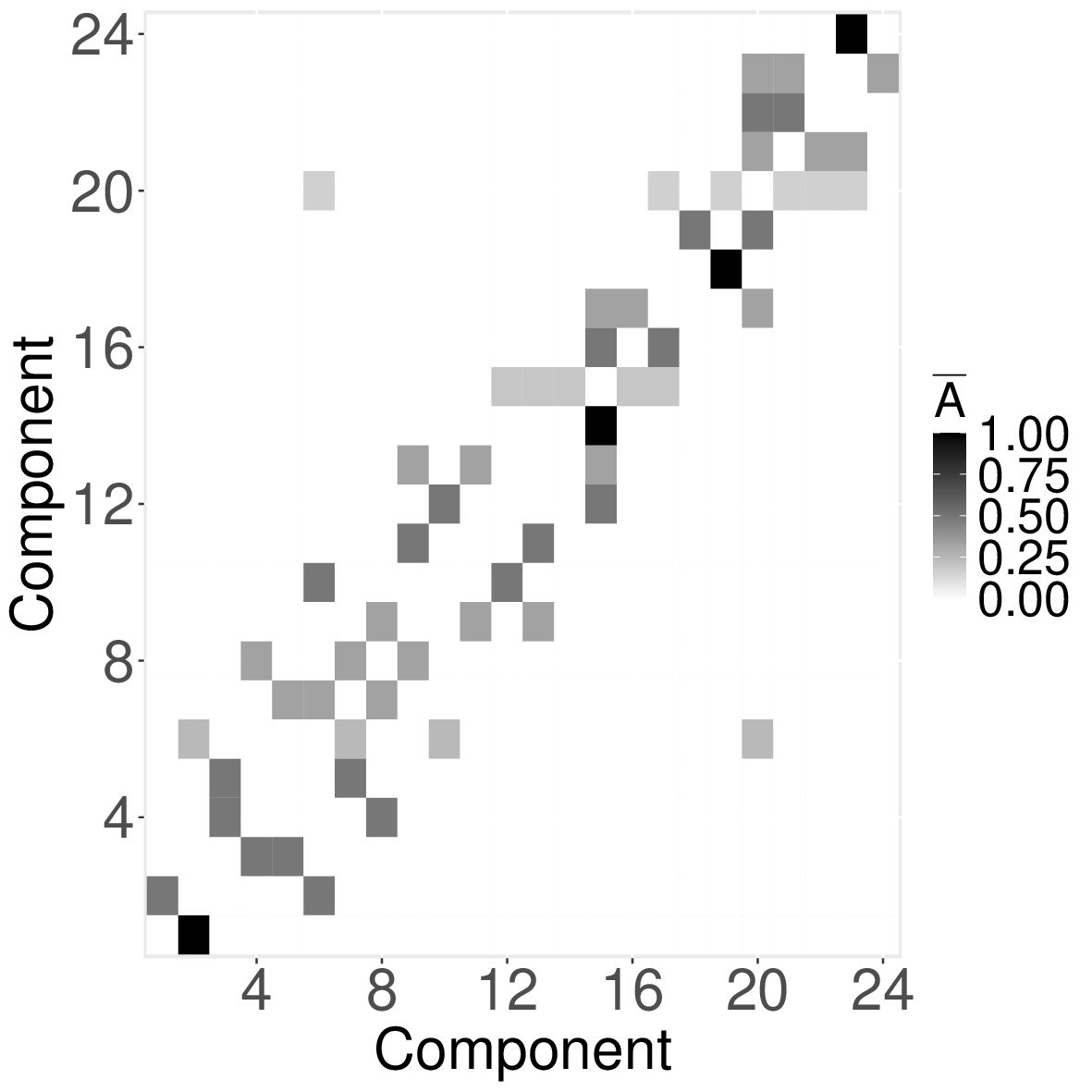}
        \caption{}
        \label{fig:normA}
    \end{subfigure}
             \caption{Graph structure: Figure \ref{fig:map} shows the 24-node network in Portugal; Figure \ref{fig:normA}  contains the heatmap of the column normalised adjacency matrix associated with the graph.}
    \label{fig:MapData}
\end{figure}
In our study, we focus on wind capacity factors (WCFs). The WCFs  were derived from the 
COSMO-REA$6$ data set of wind speeds from the period $2012-2014$ on an approximately  $7\times 7$ $\mathrm{km}^2$ spatial grid. These grid-based WCFs  are then mapped to the nodes as follows, see \cite{JP2017}: A grid cell gets mapped to a particular  node 1) if it is closest to the node compared to all other nodes or 2) the grid cell is closer to the node than any other grid cell. Using this approach, it is possible that a grid cell can be associated to multiple
nodes; in that case,  the contribution from that grid cell was split evenly among all associated nodes. 
WCFs can be considered as the proportion of the maximum power generation at a given time, as such, they take values in $[0,1]$.  
 We have hourly observations ($\Delta=1$ hour) over three years, resulting in
 $N = 26,304$ observations for each of the $d=24$ nodes. 
We depict the heatmap of the 24 hourly time series in Figure \ref{fig:heatmap-original}, the boxplots of their autocorrelation functions in Figure \ref{fig:empacfs-original} and the time series of node 22 (Lisbon) in Figure \ref{fig:Lisbondata-original}. 
 As expected, we find that the data exhibits daily and yearly seasonality and a mild trend. These have been removed using LOESS regression and the deseasonalised and detrended data are depicted in Figures 
  \ref{fig:heatmap}, \ref{fig:empacfs} and \ref{fig:Lisbondata}. The resulting time series are no longer restricted to taking values in the interval $[0,1]$.

% %%%%%%%%%%%%%%%%%%%%%%%Deseasonalised data
\begin{figure}[htbp]
    \centering
    % First row
    \begin{subfigure}[b]{0.3\textwidth}
        \centering
        \includegraphics[scale=0.19]%{heatmap-original.png}
        {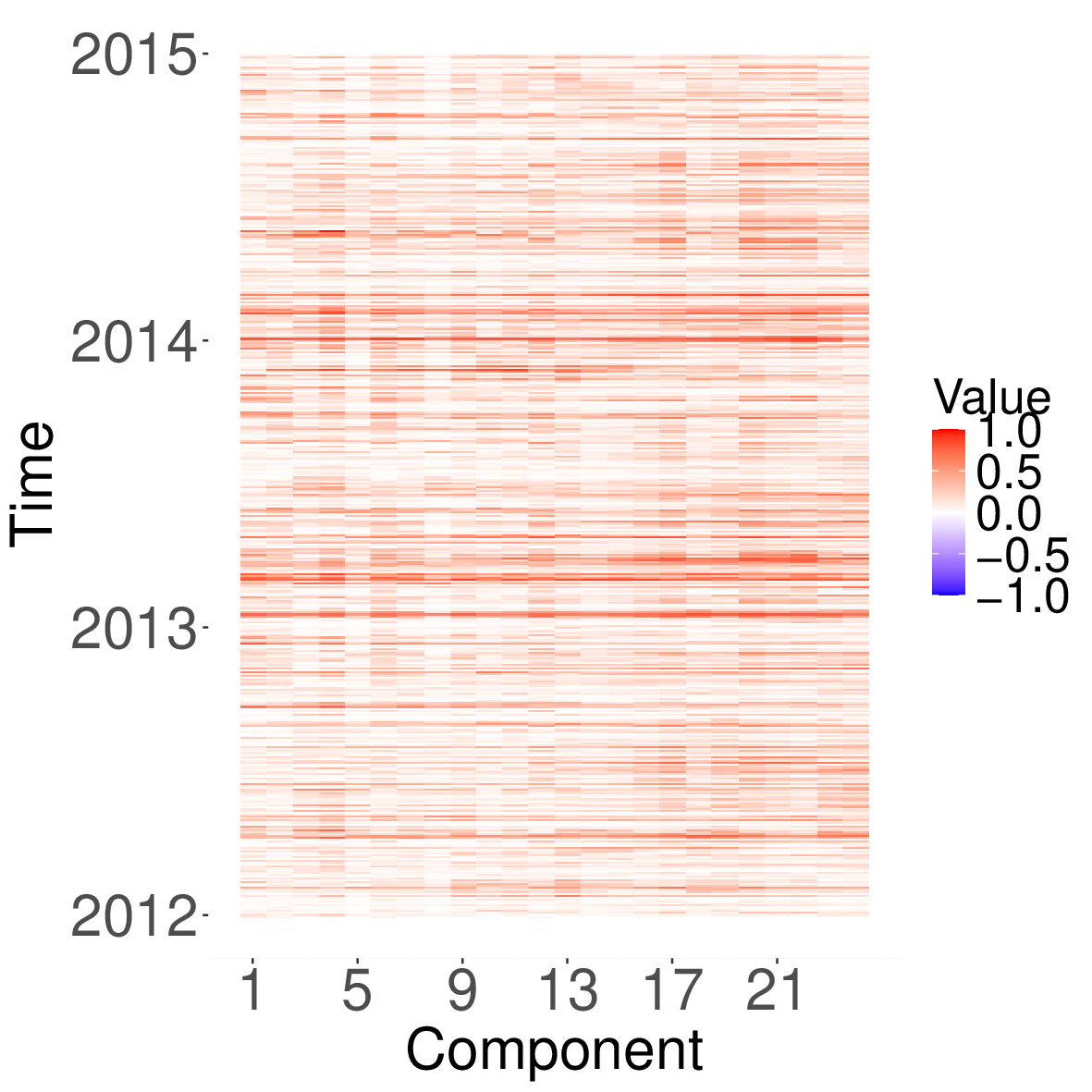}
        \caption{}
        %\caption{Heatmap of the 24 time series.}
        \label{fig:heatmap-original}
    \end{subfigure}    
    %\vspace{0.5cm} % Add vertical space between rows
    % Second row
    \begin{subfigure}[b]{0.3\textwidth}
        \centering
        \includegraphics[scale=0.2]{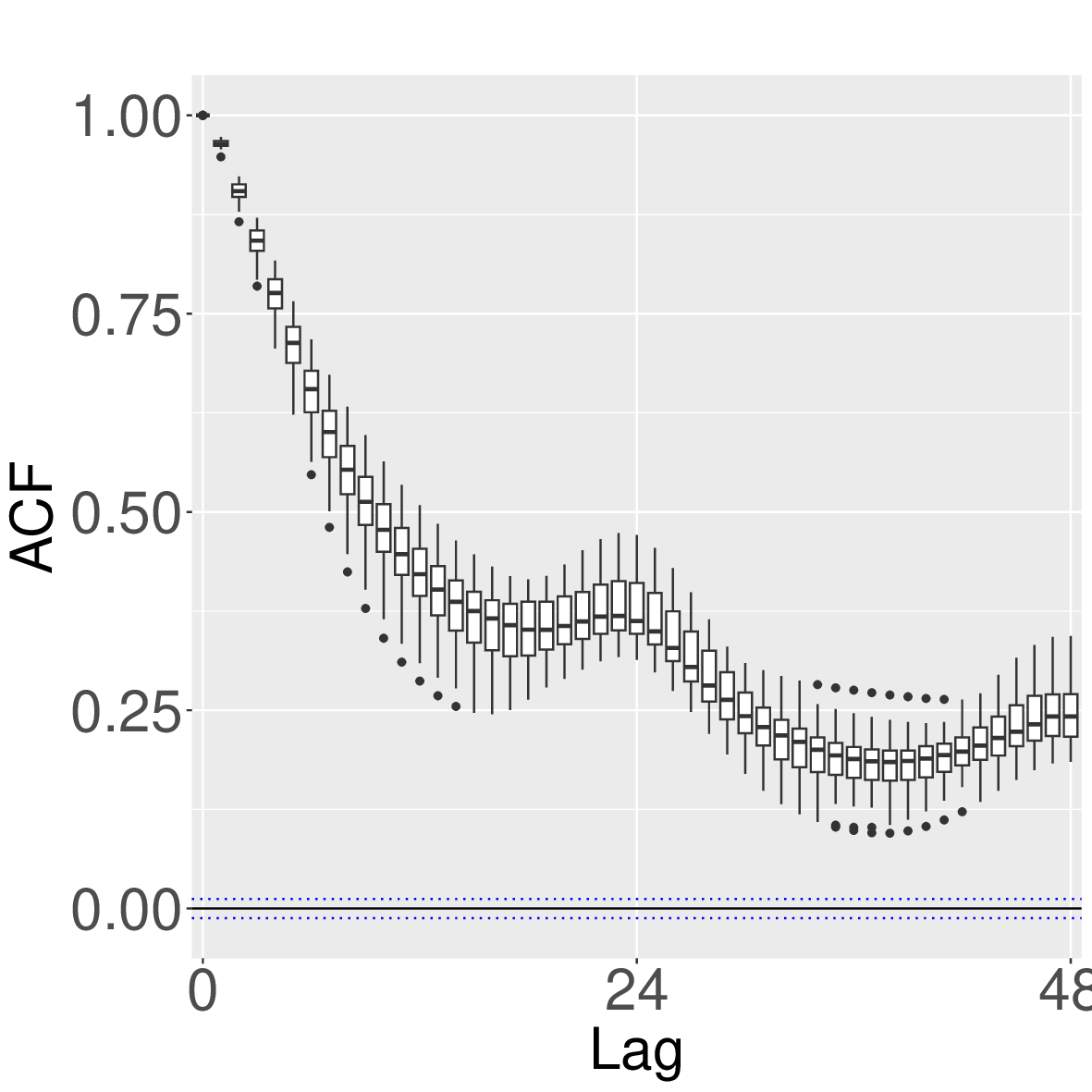}
        \caption{}
        %%\caption{Empirical ACFs.}
        \label{fig:empacfs-original}
    \end{subfigure}
       % \hfill
    \begin{subfigure}[b]{0.3\textwidth}
        \centering
        \includegraphics[scale=0.2]{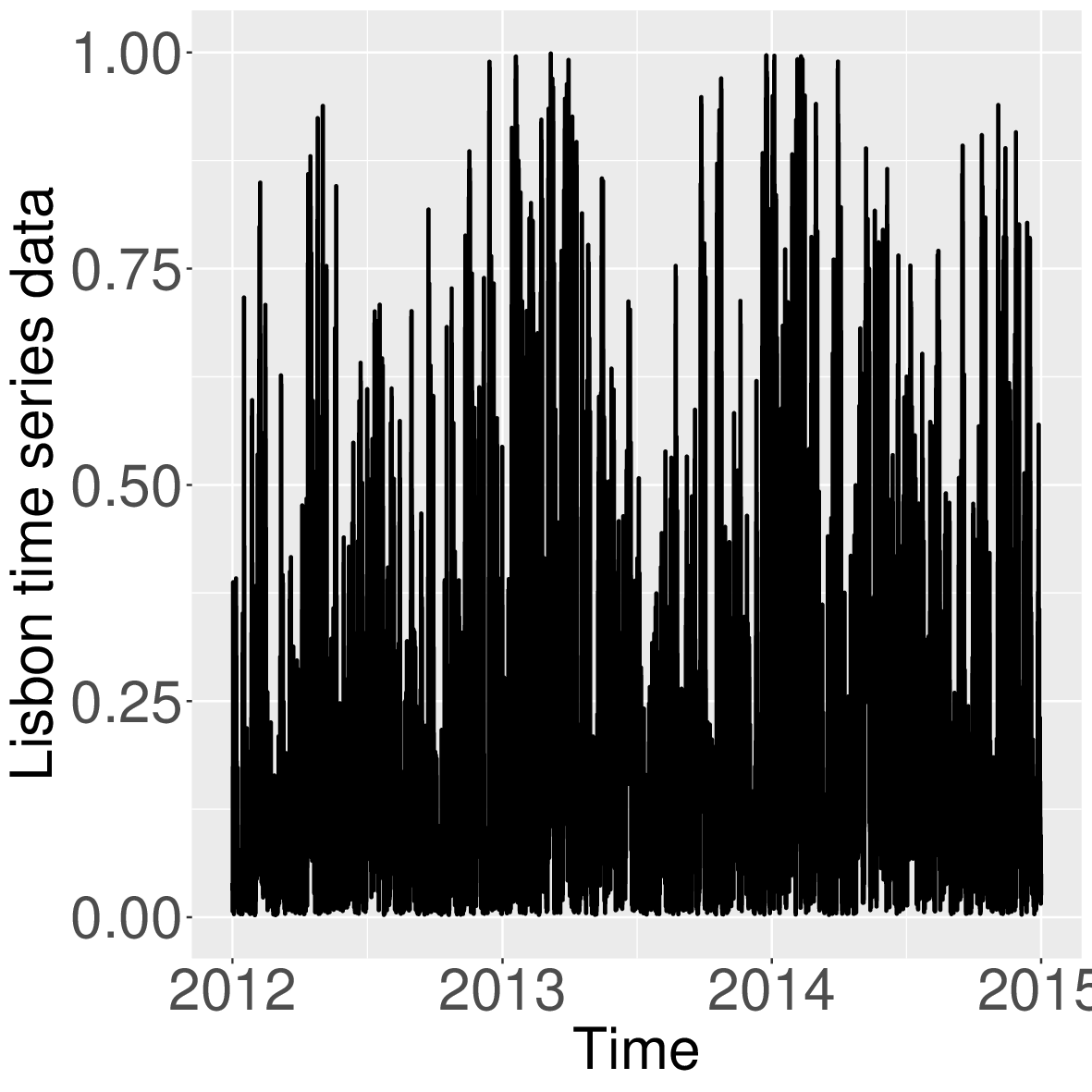}
        \caption{}
        %\caption{Hourly Lisbon data.}
        \label{fig:Lisbondata-original}
    \end{subfigure}
\\
%Second row:
        \begin{subfigure}[b]{0.3\textwidth}
        \centering
        \includegraphics[scale=0.19]%{heatmap.png}
        {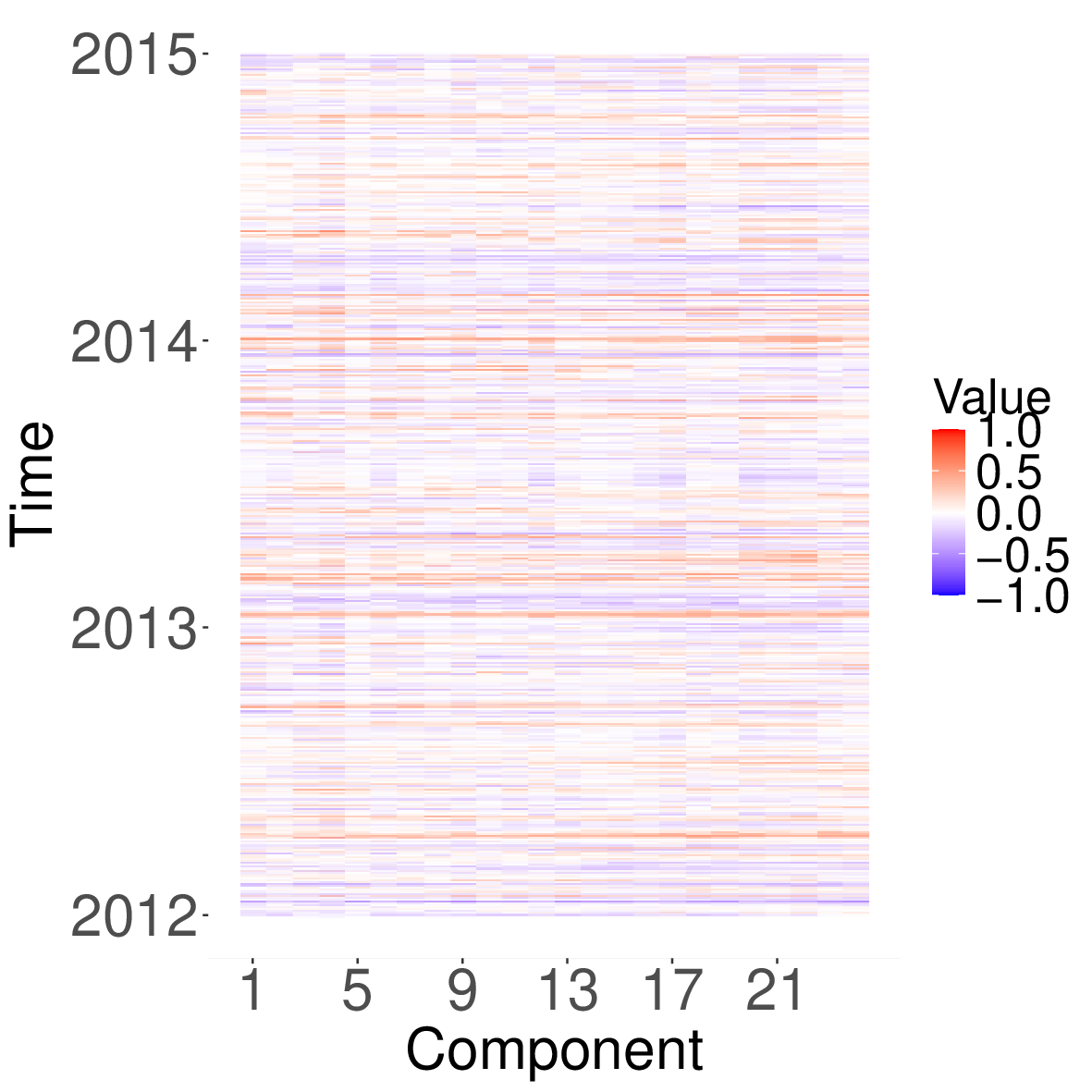}
        \caption{}
       % \caption{Heatmap of the 24 time series.}
        \label{fig:heatmap}
    \end{subfigure}    
   %\hfill
    \begin{subfigure}[b]{0.3\textwidth}
        \centering
        \includegraphics[scale=0.2]{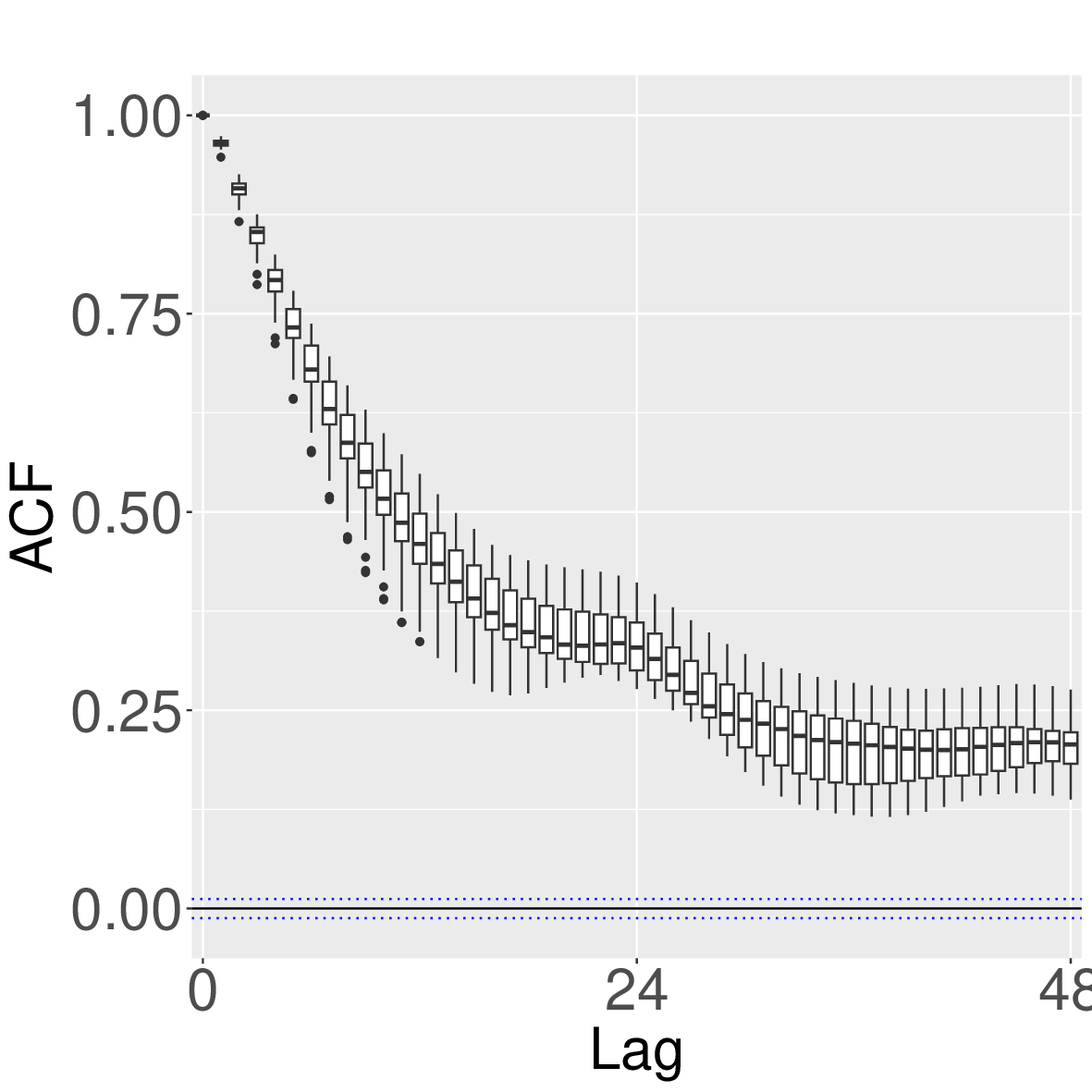}
        \caption{}
       % \caption{Empirical ACFs.}
        \label{fig:empacfs}
    \end{subfigure}
   % \hfill
    \begin{subfigure}[b]{0.3\textwidth}
        \centering
        \includegraphics[scale=0.2]{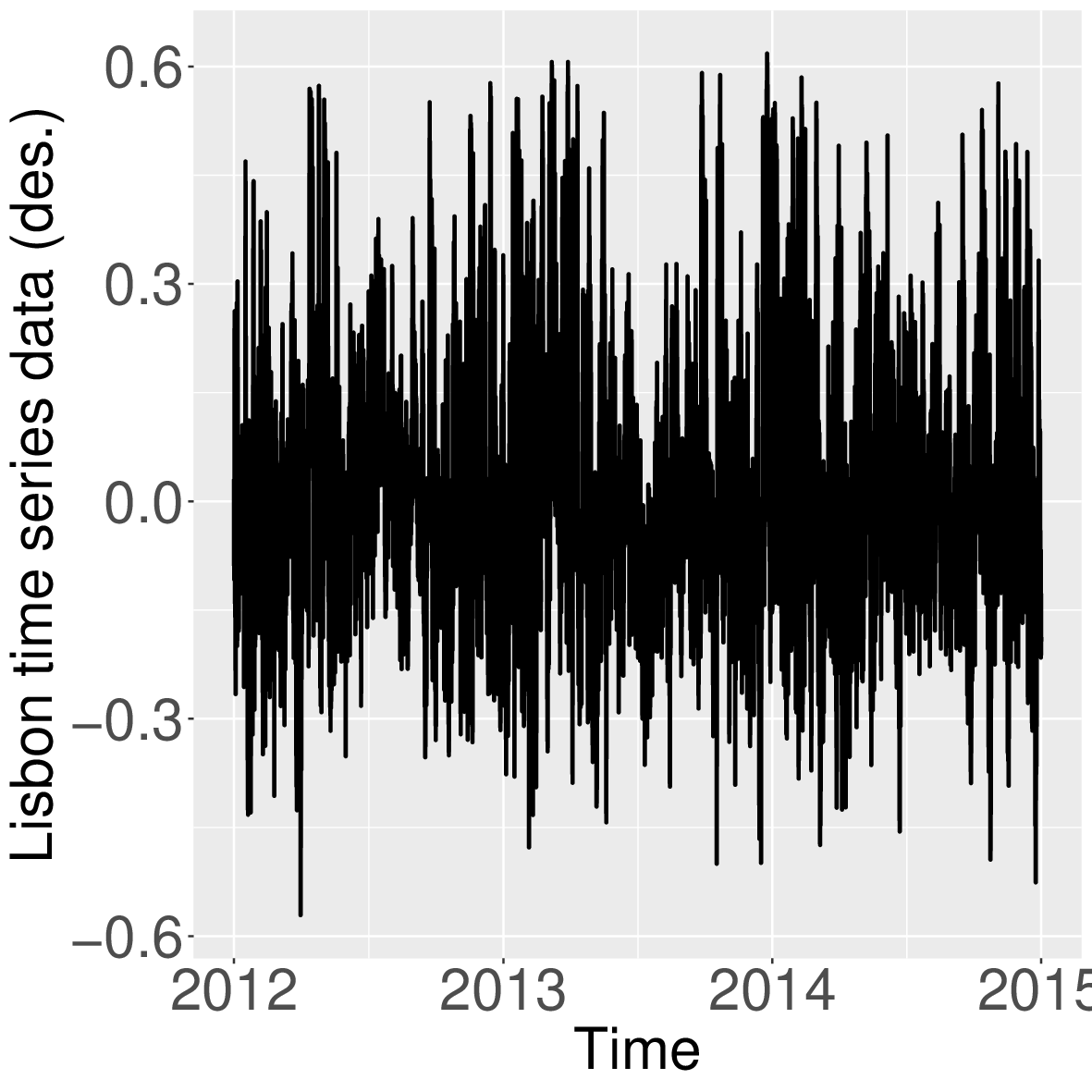}
        \caption{}
       % \caption{Hourly Lisbon data.}
        \label{fig:Lisbondata}
    \end{subfigure}
    \caption{Plots of the 24-dimensional time series of wind capacity factors in Portugal: Heatmap of the 24 hourly time series ( \ref{fig:heatmap-original}); boxplots of the ACFs (\ref{fig:empacfs-original}); time series of node 22 (Lisbon) (\ref{fig:Lisbondata-original}). The figures in the second row show the corresponding plots based on the deseasonalised and detrended data: Heatmap of the deseasonalised 24 hourly time series ( \ref{fig:heatmap}); boxplots of the ACFs of the deseasonalised time series (\ref{fig:empacfs}); deseasonalised time series of node 22 (Lisbon) (\ref{fig:Lisbondata}). 
     }
    \label{fig:TimeSeries}
\end{figure}

%%%%%%%%%%%%%%%%%
%Fitted data 40 lages 100 lags
\begin{figure}[htbp]
    \centering
    % First row
        \begin{subfigure}[b]{0.3\textwidth}
        \centering
        \includegraphics[scale=0.18]{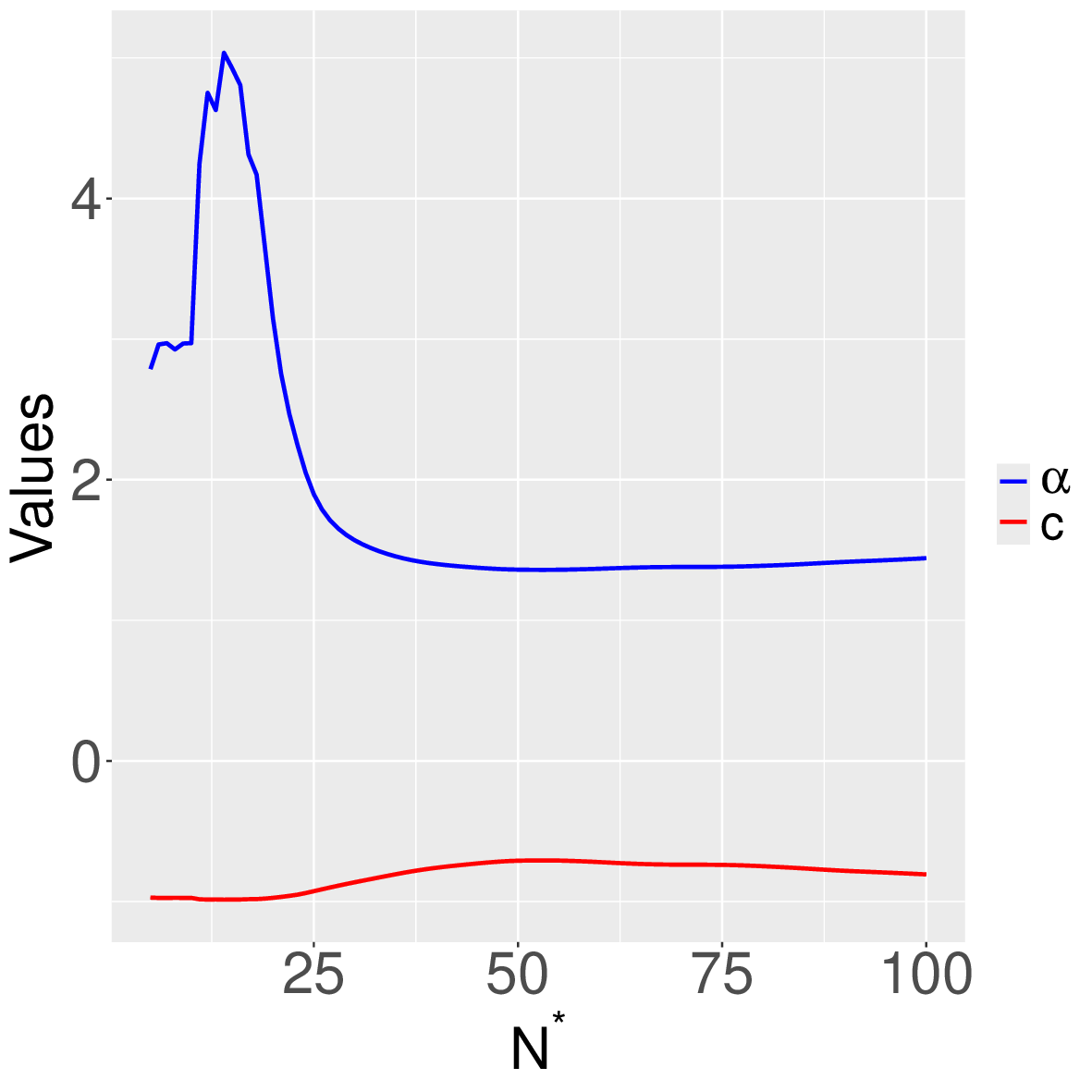}
        \caption{}
        \label{fig:alphac-vec}
    \end{subfigure}    
    %\hfill
        \begin{subfigure}[b]{0.3\textwidth}
        \centering
        \includegraphics[scale=0.18]{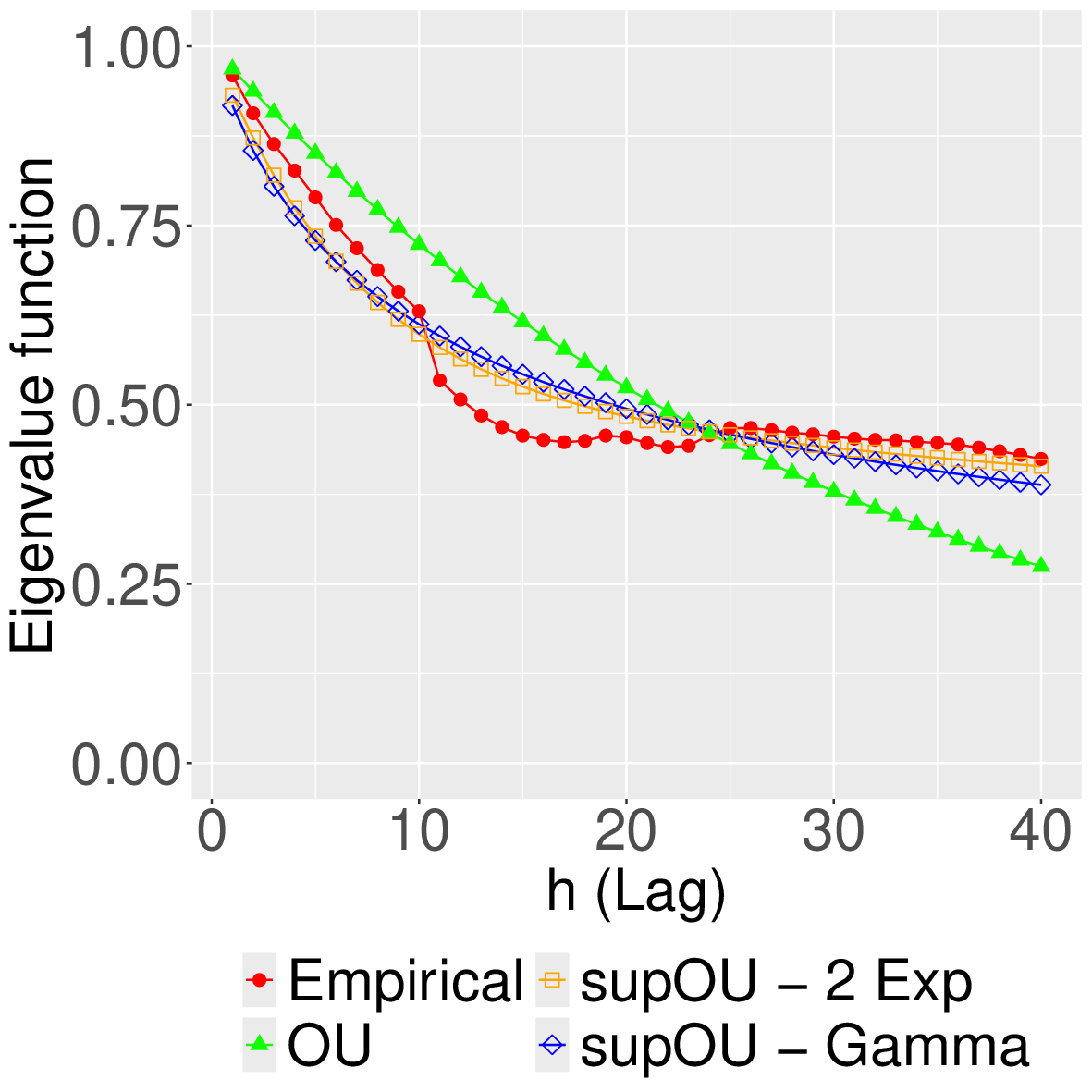}
        \caption{}
        \label{fig:Fit}
    \end{subfigure}
       % \hfill
    \begin{subfigure}[b]{0.3\textwidth}
        \centering
        \includegraphics[scale=0.18]{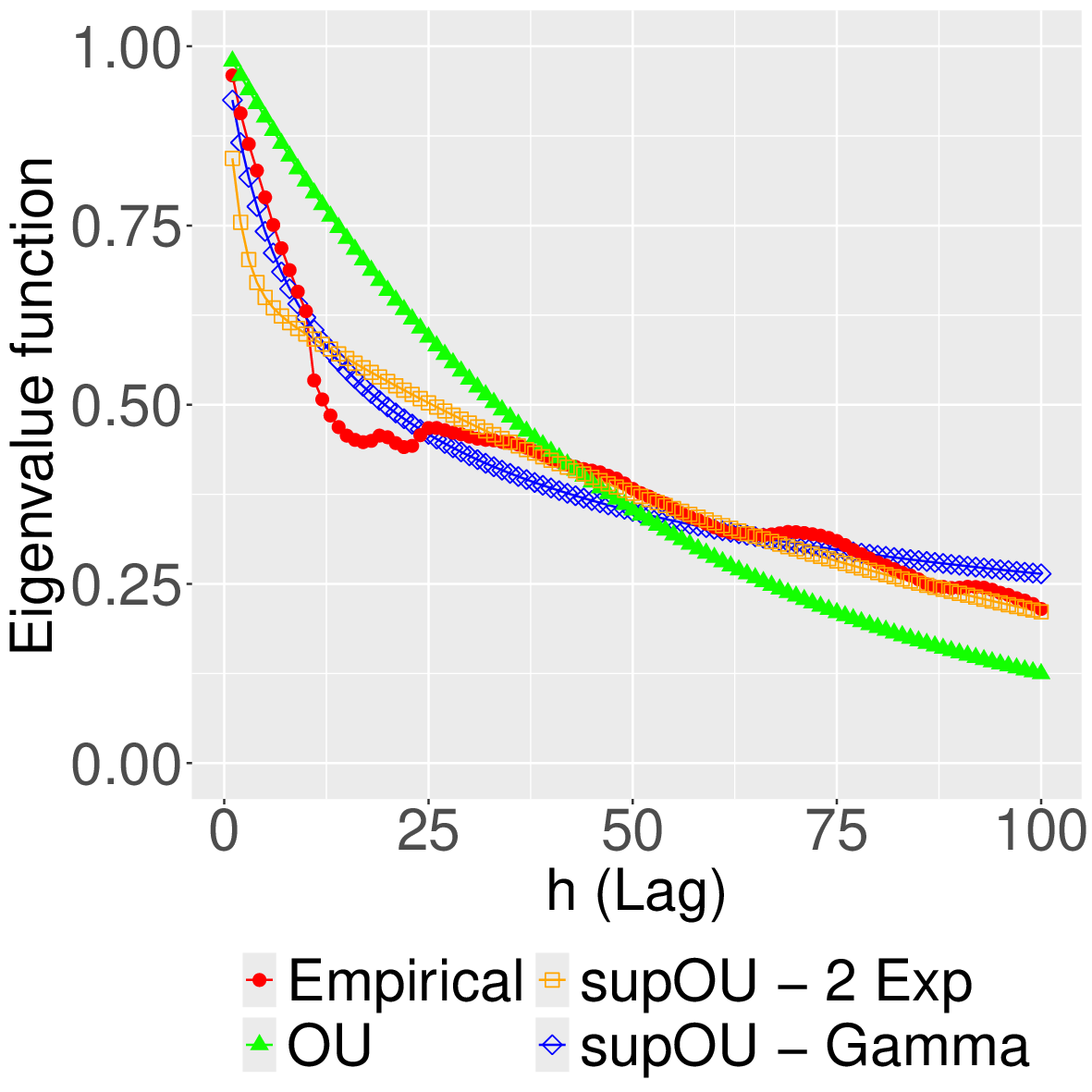}
        \caption{}
        \label{fig:Fit100}
    \end{subfigure}
    \caption{Estimation of $\alpha$ and $c$: Figure \ref{fig:alphac-vec} provides graphs of $\hat \alpha$ and $\hat c$ based on different choices of the lag $h$. Figures \ref{fig:Fit} (based on 40 lags) and \ref{fig:Fit100} (based on 100 lags) show the empirical eigenvalue function (red) together with the fitted counterparts based on a  graph OU  (green) and two graph supOU specifications, the gamma model (blue) and a sum of two weighted exponentials (orange). }
    \label{fig:Fitted}
\end{figure}
In Figure \ref{fig:Fitted}, we provide the results from estimating the parameters  $\alpha$ and $c$. First,  Figure \ref{fig:alphac-vec} compares the estimates of  $\hat \alpha$ and $\hat c$ based on different choices of $N^*$. We observe that at around $N^*=40$, the estimates stabilise at around $\hat \alpha=1.44$ and $\hat c=-0.81$.  
Using the estimate for $\alpha$ and $c$ obtained for $N^*=40$, we then show a graph of the largest real eigenvalue $\hat l(h\Delta)$ (red) and the function $\rho(h\Delta; (\alpha,c)^{\top})$ defined in  \eqref{eq:estfct-LM} (blue). For comparison, we also provide the corresponding fitted curve (green) when a graph OU process was fitted %(with resulting memory parameter 0.02) 
and when a supOU process with a sum of two weighted exponentials was used as a kernel function (orange).

We repeat the exercise when we use $N^*=100$ for the original estimation of $\alpha$ and $c$ and show the results in Figure  \ref{fig:Fit100}. 
We observe the following: The parameter estimates of $\alpha$ and $c$ are not particularly sensitive to the choice of  $N^*$ as soon as a reasonable length (here $N^*\geq 40$) is chosen. We note that the exponentially decaying function (green) coming from the fitted graph OU process decays too fast and provides an inferior model fit compared to the fitted graph supOU process (blue and orange). When estimating the Gamma-distribution-based graph supOU model, the parameter estimate indicated that we are in a long-memory regime. However, we note that also a weighted sum of exponentials seems to provide an acceptable fit, which is also superior to the fit provided by the graph OU model. 
When comparing the results in Figures \ref{fig:Fit} (based on $N^*=40$) and \ref{fig:Fit100} (based on $N^*=100$), we note that the findings are very similar, both showcasing the flexibility of the graph supOU model over the earlier considered graph OU process. We note that confidence intervals for the parameter estimates can be constructed via parametric bootstrap, as in our simulation study, see Figure \ref{fig:p_vio}, or via the derived asymptotic theory in a short-memory setting.
\begin{figure}[htbp]
    \centering
    % First row
        \begin{subfigure}[b]{0.3\textwidth}
        \centering
        \includegraphics[scale=0.2]{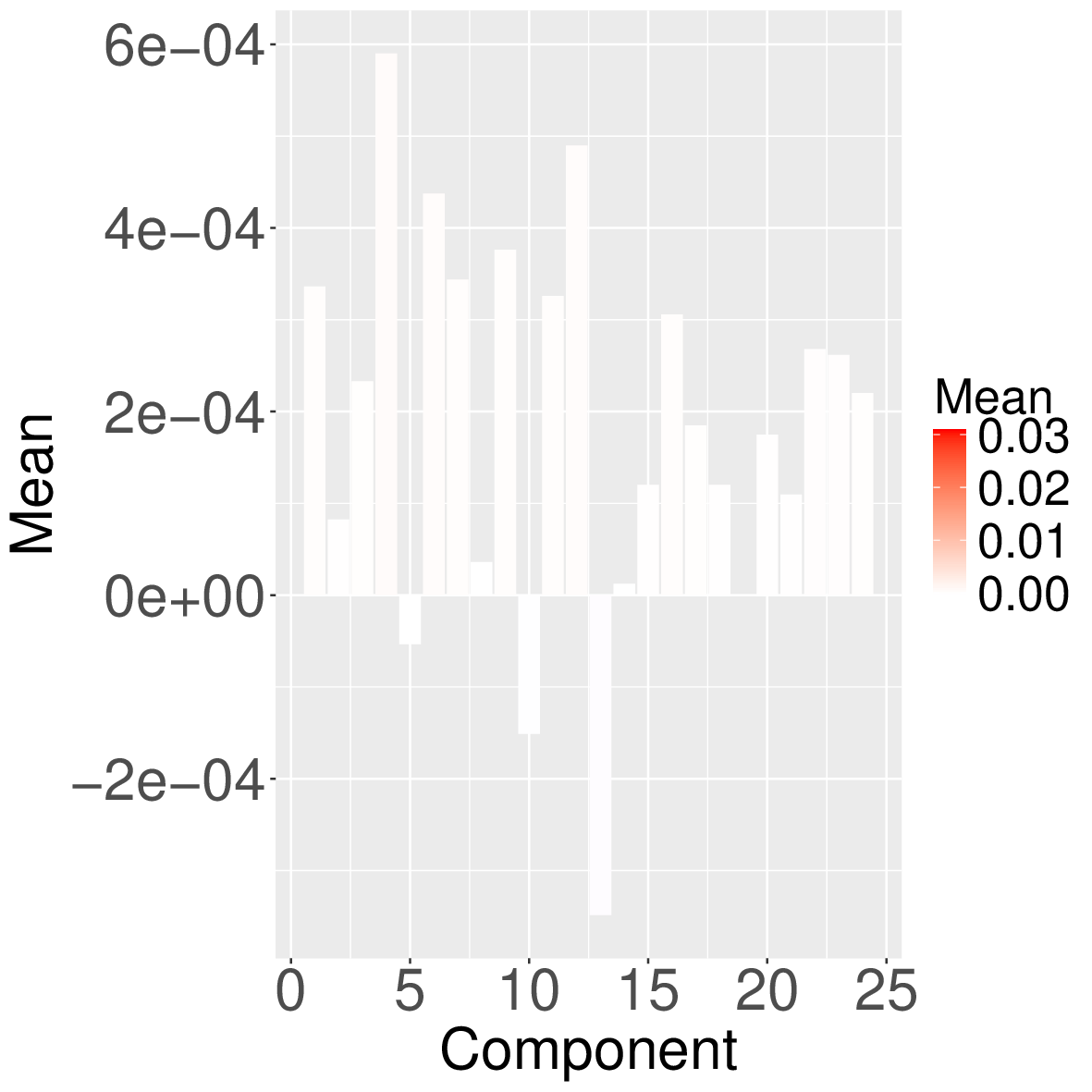}
        \caption{}
        \label{fig:estmean}
    \end{subfigure}    
    %\hfill
    \begin{subfigure}[b]{0.3\textwidth}
        \centering
        \includegraphics[scale=0.2]{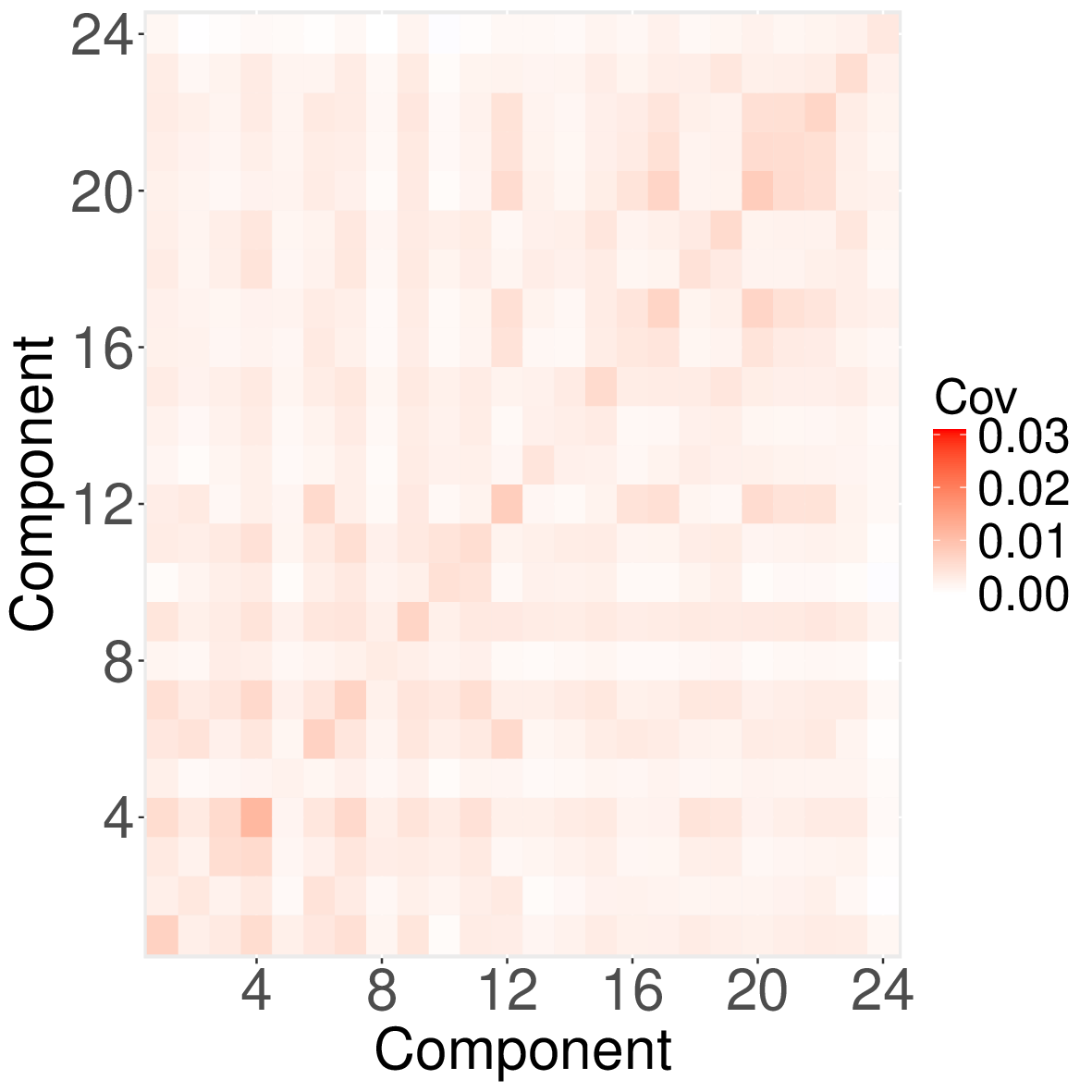}
        \caption{}
        \label{fig:estvar}
    \end{subfigure}
   % \hfill
    \begin{subfigure}[b]{0.3\textwidth}
        \centering
        \includegraphics[scale=0.2]{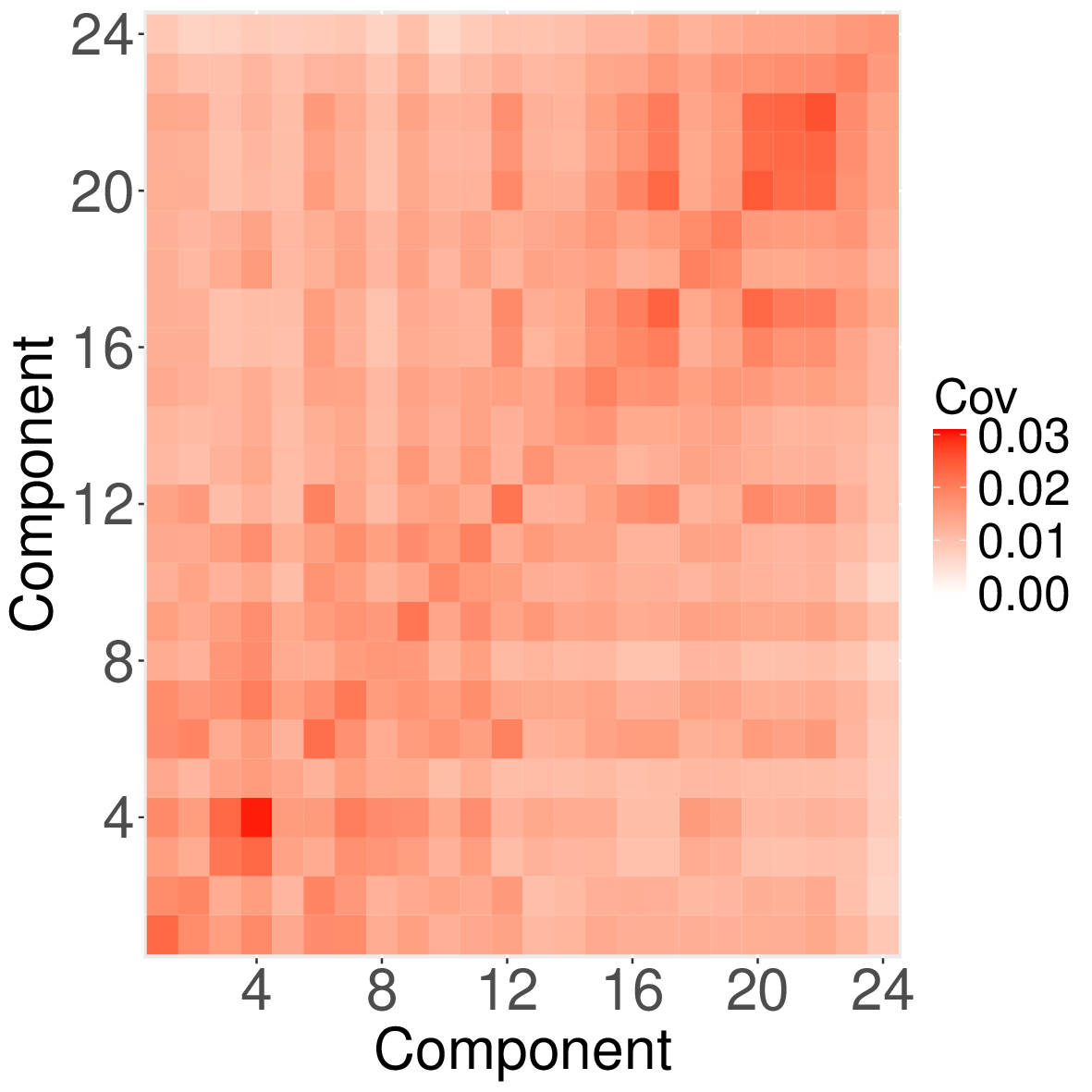}
        \caption{}
        \label{fig:empvar}
    \end{subfigure}
         \caption{Estimated mean and variance of the L\'{e}vy basis: Figure \ref{fig:estmean} Estimated mean $\widehat{\mu_L}$, Figure  \ref{fig:estvar} estimated variance $\widehat{{\boldsymbol \sigma}_L^2}$, and Figure  \ref{fig:empvar} empirical covariance $\widehat{\mathrm{Var}(X)}$. }
    \label{fig:Fitted2}
\end{figure}

Finally, we estimate the mean and variance of the driving L\'{e}vy basis. Figure \ref{fig:estmean} shows the estimated mean $\widehat{\mu_L}$, which is close to 0. Figure  \ref{fig:estvar} depicts the estimated variance matrix $\widehat{{\boldsymbol \sigma}_L^2}$ showing very low covariances throughout, and Figure  \ref{fig:empvar} provides a heatmap of the empirical covariance $\widehat{\mathrm{Var}(X)}$ for comparison.

\section{Conclusion and outlook}\label{sec:conclusion}
 This article proposed the graph supOU process as a new and flexible modelling framework for high-dimensional network stochastic processes, that can handle short- and long-memory settings and accommodate a wide range of marginal distributions. We have further developed an estimation methodology that is easy to implement and avoids high-dimensional optimisation procedures. We have proven the consistency of the methodology and have also considered the classical GMM framework and derived the corresponding asymptotic theory. 
 The article is accompanied by  the {\tt R} code that reproduces the simulation and empirical study.  

 In future work, various generalisations could be considered further. For example, we have explored the case of undirected graphs; however, the extension to the directed case is straightforward and follows the same arguments as the ones presented here. Also, similarly to the work of \cite{Bar11}, we could consider defining graph-based matrix-valued processes. Such processes would find applications in multivariate stochastic volatility models in finance and beyond.

\begin{appendix}

%%%%%%%%%%%%%%%%%%%%%%%%%%%%%%%%%%
\section{GMM framework and asymptotic theory} \label{app:GMM}
This appendix collects the full GMM framework underlying
Theorems~\ref{th:consistency-gamma} and~\ref{th:AN-main} in the main text,
including all intermediate assumptions, the weak dependence theory for
graph supOU processes, and the central limit theorem for the moment
function.
The arguments of the proofs follow the GMM framework  presented in \cite{Mat} and extend ideas by  \cite{CS} and \cite{BLSA} for mixed moving average and trawl processes, respectively.

\subsection{General GMM framework}
We present the GMM framework for graph supOU processes in full generality,
not restricted to the parametric setting of Section~\ref{sec:bridge}.
The moment function $f$ and sample counterpart $f_N$ are defined in
equations~\eqref{eq:f}--\eqref{eq:samplem} of the main text, and the
GMM estimator $\hat{\xi}_0^{N}$ in~\eqref{eq:Estimator}.

\subsection{Assumptions for consistency}
We will review the  assumptions postulated by  \cite{Mat}.
Our Assumptions \ref{A1}, \ref{A2}, and \ref{A3}, corresponds to \cite[Assumptions 1.1-1.3]{Mat}.

%\textbf{Assumption 1} 
\begin{assumption}\label{A1}
\begin{enumerate}
    \item[(i)]  $\mathrm{E}(f(X_{t:t+m},\xi))$ exists and is finite for all $\xi\in\Xi$ and all $t$. 
    \item[(ii)] Let $g_t(\xi):=\mathrm{E}(f(X_{t:t+m},\xi))$. There exists $\xi_0\in\Xi$ such that $g_t(\xi)=0$ for all $t$ if and only if $\xi=\xi_0$.
\end{enumerate}
\end{assumption}
\begin{remark}\label{rem:ident}
Assumption \ref{A1} (i) holds by the construction of $f(X_{t:t+m},\xi)$ under a suitable second-order moment assumption; the identifiability condition (ii) must be verified for each specific parametric model, and follows directly for the Gamma and exponential mixture models of
Section~\ref{sec:bridge}.
\end{remark}

The next assumption concerns the convergence of the sample moments to the population moments. Recall, that $f_{N}(\xi)$ denotes the vector of the sample moments, see \eqref{eq:samplem}; we denote by $$g_N(\xi):=\frac{1}{N-m}\sum_{t=1}^{N-m}g_t(\xi),$$ the corresponding vector of the population moments. The $i$th component of these two $q$-dimensional vectors are denoted by $f_{N}^{(i)}(\xi)$ and $g_N^{(i)}(\xi)$, respectively.

%\textbf{Assumption 2}\\
\begin{assumption}\label{A2}
Suppose that $\sup_{\xi \in \Xi}|f_{N}^{(i)}(\xi)-g_N^{(i)}(\xi)|\stackrel{\mathrm{P}}{\longrightarrow} 0$ for all $i=1, \ldots, q$, as $N\to \infty$.
\end{assumption}
Moreover, the following assumption regarding the weighting matrix is also required.

%\textbf{Assumption 3}\\
\begin{assumption}\label{A3}
There is a  sequence of deterministic, positive definite matrices denoted by $\overline{\mathbf{W}}_{N}$ such that  $\mathbf{W}_{N}-\overline{\mathbf{W}}_{N}\stackrel{\mathrm{P}}{\longrightarrow} 0$, as $N\to \infty$.
\end{assumption}
Then we have the following result.

\begin{theorem}\label{th:consistency} Suppose that the graph supOU process has finite second moments and that the Assumptions \ref{A1}, \ref{A2}, and \ref{A3} hold. Then the GMM estimator $\hat{\xi}_0^{N}$,  defined in equation \eqref{eq:Estimator}, is weakly consistent, i.e.
$$
\hat{\xi}_0^{N} \stackrel{\mathrm{P}}{\longrightarrow} \xi_0, \quad \mathrm{as}\, 
N\to \infty.$$
 \end{theorem}
 %%%%%%%%%%%%%%%%%%%
% Proof consistency
\begin{proof}[Proof of Theorem \ref{th:consistency}]
 The weak consistency of $\hat{\xi}_0^{N}$ follows directly  from  \cite[Theorem 1.1]{Mat}.
\end{proof}

\cite{Mat} points out that Assumptions \ref{A2} and \ref{A3} are rather high-level; he shows that Assumption \ref{A2} can be replaced with the sufficient, more basic, assumptions \cite[Assumptions 1.4-1.6]{Mat}, which we refer to as Assumptions \ref{A4}, \ref{A5}, and \ref{A6} below.

%\textbf{Assumption 4}\\
\begin{assumption}\label{A4}
The parameter space $\Xi$ is compact and  contains the true parameter $\xi_0$.
\end{assumption}
It is worth noting that although the parameter space may not inherently be closed and bounded, it can be made compact by imposing appropriate constraints during the optimisation procedure to determine the GMM estimator. Consequently, in practice, Assumption \ref{A4} is less restrictive than it might initially appear.

%\textbf{Assumption 5}\\
\begin{assumption}\label{A5}
$|f_{N}^{(i)}(\xi)-g_N^{(i)}(\xi)|\stackrel{\mathrm{P}}{\longrightarrow} 0$ pointwise on $\Xi$ for all $i=1, \ldots, q$, as $N\to \infty$. 
\end{assumption}
\begin{remark}\label{rem:mixing}
Since the (graph) supOU process belongs to the class of mixed moving average processes and is, hence, mixing and ergodic, it satisfies Assumption \ref{A5}.
\end{remark}

%\textbf{Assumption 6}\\
\begin{assumption}\label{A6}
Suppose that  $f_{N}^{(i)}(\xi)$ is stochastically equicontinuous and $g_N^{(i)}(\xi)$ is equicontinuous for all $i=1,\ldots, q$.
\end{assumption}

\begin{remark} In order to  verify Assumption \ref{A6}, we need to show that, for each $i=1,\ldots, q$, $f_N^{(i)}$ and $g_N^{(i)}$ satisfy a stochastic Lipschitz and Lipschitz property, respectively. Let us study the individual components of $f(X_{t:t+m},\xi)$. Let $\xi_1, \xi_2\in \Xi$ represent two parameter vectors.
For the   component $f_E$, we have
\begin{align}\begin{split}\label{eq:A6-comp1}
f_\mathrm{E}(X_{t:t+m},\xi_1)- f_\mathrm{E}(X_{t:t+m},\xi_2)
&= \text{vec}     (  X_{t\Delta}-\mu(\xi_1))- \text{vec}     (  X_{t\Delta}-\mu(\xi_2))\\
&=\text{vec}     (-\mu(\xi_1)+\mu(\xi_2))
\end{split}\end{align}
For the other components $f_i$, for $i=0,\ldots, m$, we have 
\begin{align}\begin{split}\label{eq:A6-comp2}
&f_\mathrm{i}(X_{t:t+m},\xi_1)- f_\mathrm{i}(X_{t:t+m},\xi_2)\\
&=  \text{vec}_h ( X_{t\Delta} X_{(t+i)\Delta}^{\top}-\mathbf{D}_i(\xi_1))
-
 \text{vec}_h         ( X_{t\Delta} X_{(t+i)\Delta}^{\top}-\mathbf{D}_i(\xi_2))
 \\
 &=-\text{vec}_h (\mathbf{D}_i(\xi_1)) +\text{vec}_h (\mathbf{D}_i(\xi_2)).
\end{split}\end{align}

We observe that the terms involving $X_{t:t+m}$ cancel out by the construction of the moment function.  This leaves us with requiring a Lipschitz condition on the non-random terms in each component of $f(X_{t:t+m},\xi)$. For instance,  if we consider partial derivatives with respect to the model parameters and  find that these partial derivatives are bounded, then we can deduce that the components are Lipschitz continuous and Assumption \ref{A6} is satisfied.
\end{remark}

%%%%%%%%%%%%%%%%%%%%%%%%%%%%%%%%
We recall Theorem \ref{th:consistency-gamma} from the main article:
\begin{theorem*}[Consistency, Theorem \ref{th:consistency-gamma},  restated]%\label{th:consistency-gamma}
Suppose that the graph supOU process has finite second moments, the drift is given by \eqref{eq:Qtheta2} , for $|c|<1$. Moreover, for $\alpha>1$, we have $\theta_2\sim \Gamma(\alpha,1)$ as in \eqref{eq:PD} or $\theta_2$ is chosen as in \eqref{eq:Exp}.
Suppose that Assumptions \ref{A3} and \ref{A4}  hold. 
Then the GMM estimator $\hat{\xi}_0^{N}$,  defined in equation \eqref{eq:Estimator}, is weakly consistent, i.e.
$$
\hat{\xi}_0^{N} \stackrel{\mathrm{P}}{\longrightarrow} \xi_0, \quad \mathrm{as}\, 
N\to \infty.$$
 \end{theorem*}
%Proof consistency Gamma
\begin{proof}[Proof of Theorem \ref{th:consistency-gamma}]
 We proceed as in \cite{Mat} who shows that the consistency result holds under Assumptions \ref{A1}-\ref{A3}.
 Moreover, \cite{Mat} argues that Assumption \ref{A2} can be replaced by Assumptions \ref{A4}-\ref{A6}. I.e. under Assumptions \ref{A1}, \ref{A3}, \ref{A4}-\ref{A6}, we will get the stated consistency result. Since Assumptions \ref{A3} and \ref{A4} are stated in the theorem, it remains to verify Assumptions \ref{A1}, \ref{A5} and \ref{A6}.
 We have already verified the parameter identifiability for this particular model  in Section \ref{sec:ident}. Together with our earlier discussion, see Remark \ref{rem:ident}, this ensures that Assumption \ref{A1} holds. 
 Also, we have shown that Assumption \ref{A5} holds, see Remark \ref{rem:mixing}. Hence, it remains to verify Assumption \ref{A6} in the Gamma case. 
As before, see \eqref{eq:A6-comp1} and \eqref{eq:A6-comp2}, we need to show that, for each $i=1,\ldots, q$, $f_N^{(i)}$ and $g_N^{(i)}$ satisfy a stochastic Lipschitz and Lipschitz property, respectively. Let us study the individual components of $f(X_{t:t+m},\xi)$. Let $\xi_1, \xi_2\in \Xi$ represent two parameter vectors.
For the   component $f_E$, we have
\begin{align*}
 & f_\mathrm{E}(X_{t:t+m},\xi_1)- f_\mathrm{E}(X_{t:t+m},\xi_2)
= \text{vec}     (  X_{t\Delta}-\mu(\xi_1))- \text{vec}     (  X_{t\Delta}-\mu(\xi_2))\\
&=\text{vec}     (-\mu(\xi_1)+\mu(\xi_2))\\
&=
\text{vec}\left(-\frac{1}{\alpha_1-1} \mathbf{K}(c_1)^{-1}  \mu_{L,1}
+\frac{1}{\alpha_2-1} \mathbf{K}(c_2)^{-1}  \mu_{L,2}\right)\\
&=
\text{vec}\left(-\frac{1}{\alpha_1-1} \left(-\mathbf{I}_{d\times d}-c_1\bar{\textbf{A}}^{\top}\right)^{-1}  \mu_{L,1}
+\frac{1}{\alpha_2-1} \left(-\mathbf{I}_{d\times d}-c_2\bar{\textbf{A}}^{\top}\right)^{-1}  \mu_{L,2}\right).
\end{align*}

For the other components $f_i$, for $i=0,\ldots, m$, we have 
\begin{align}\begin{split}\label{eq:A6-comp2}
&f_\mathrm{i}(X_{t:t+m},\xi_1)- f_\mathrm{i}(X_{t:t+m},\xi_2)\\
&=  \text{vec}_h ( X_{t\Delta} X_{(t+i)\Delta}^{\top}-\mathbf{D}_i(\xi_1))
-
 \text{vec}_h         ( X_{t\Delta} X_{(t+i)\Delta}^{\top}-\mathbf{D}_i(\xi_2))
 \\
 &=-\text{vec}_h (\mathbf{D}_i(\xi_1)) +\text{vec}_h (\mathbf{D}_i(\xi_2)),
\end{split}\end{align}
where, for $j=1,2$, 
\begin{align*}
\mathbf{D}_i(\xi_j)&=
-\frac{1}{\alpha_j-1}(\mathbf{I}_{d\times d}-\mathbf{K}(c_j)i\Delta)^{1-\alpha} (\mathcal{A}(\mathbf{K}(c_j))^{-1}(\boldsymbol{\sigma}^2_{L,j})
\\
& \quad + \frac{1}{\alpha_j-1} \mathbf{K}(c_j)^{-1}  \mu_{L,j} \mu_{L,j}^{\top}\left( \frac{1}{\alpha_j-1} \mathbf{K}(c_j)^{-1}  \right)^{\top}.
\end{align*}
As before, we observe that the terms involving $X_{t:t+m}$ cancel out by the construction of the moment function.  This leaves us with requiring a Lipschitz condition on the non-random terms in each component of $f(X_{t:t+m},\xi)$. For our particular model, the partial derivatives with respect to the model parameters are  bounded, which implies  the  Lipschitz continuity  and Assumption \ref{A6}.
\end{proof}

%%%%weak dependence
\subsection{Weak dependence of graph supOU processes}
As a prerequisite for proving the asymptotic normality of the GMM estimator, we  will now 
 study the weak dependence properties of graph supOU processes.
 Weak dependence in a process refers to a situation where the values of the process are not independent,  
 but exhibit diminishing correlation or association as the time lag between observations increases.
  We give the definition of $\zeta$-weakly dependent processes in the following, which were called $\theta$-weakly dependent processes  in \cite{CS}. 
  Similar to \cite[Section 3.1]{CS}, we introduce the following notation: Let $\mathcal{G}=\cup_{u\in \mathbb{N}^+}\mathcal{G}_u$, where the set $\mathcal{G}_u$ is given by the class of bounded functions from $(\mathbb{R}^d)^u$ to $\mathbb{R}$ that are Lipschitz continuous with respect to the distance function $\delta^\ast$ on $(\mathbb{R}^d)^u$ that is defined as
  $ \delta^\ast((x_1,\ldots,x_u)^{\top},(y_1,\ldots,y_u)^{\top})=\sum_{i=1}^u\|x_i-y_i\|$, 
    for $x_i, y_i \in \mathbb{R}^d$. We define
    $\mathcal{H}:=\{f\in \mathcal{G}: \|f\|_{\infty}\leq 1 \}$. Similarly, we define $\mathcal{G}^{\ast}=\cup_{u\in \mathbb{N}^+}\mathcal{G}_u^{\ast}$, where  $\mathcal{G}_u$ is the set of  bounded functions from $(\mathbb{R}^d)^u$ to $\mathbb{R}$. We set  $\mathcal{H}^{\ast}:=\{f\in \mathcal{G}^{\ast}: \|f\|_{\infty}\leq 1 \}$.

 \begin{definition}{(Definition 3.2, \cite{CS})}
A process $X=(X_t)_{t\in\mathbb{R}}$ taking values in $\mathbb{R}^d$ is called a $\zeta$-weakly dependent process if there exists a sequence $(\zeta(r)_{r\in\mathbb{R}^+})$ that converges to 0, and this sequence satisfies the following condition
\[
|\mathrm{cov}(F(X_{i_1},X_{i_2},\dots,X_{i_u}),G(X_{j_1},X_{j_2},\dots,X_{j_v}))|\leq c(v\mathrm{Lip}(G)\|F\|_\infty)\zeta(r),
\]
for all
$(u,v)\in\mathbb{N}^+\times\mathbb{N}^+,$ $r\in\mathbb{R}^+$,
$(i_1,\dots, i_u)\in\mathbb{R}^u$ and $(j_1,\dots, j_v)\in\mathbb{R}^v$, with $i_1\leq \dots \leq i_u\leq i_u+r\leq j_1\leq \dots \leq j_v$, functions $F:(\mathbb{R}^d)^u\rightarrow \mathbb{R}$ and $G:(\mathbb{R}^d)^v\rightarrow \mathbb{R}$ respectively belonging to $\mathcal{H}^{\ast}$ and $\mathcal{H}$, 
and $\mathrm{Lip}(G)=\mathrm{sup}_{x\neq y}\frac{|G(x)-G(y)|}{\|x_1-y_1\|+\|x_2-y_2\|+\dots+\|x_d-y_d\|}$, 
where $c$ is a constant independent of $r$. We call $(\zeta(r))_{r\in\mathbb{R}^+}$ the sequence of the $\zeta$-coefficients.
\end{definition}
We will now present a general representation result for the $\zeta$-coefficients for graph supOU processes and derive upper bounds which are more easily computable in practice.

\begin{theorem}\label{th:zetasupOU} Under the assumptions of Proposition \ref{prop:mom}, the $\zeta$-coefficients of a graph supOU process satisfy, for all $r\geq 0$,
 \begin{multline*}%{\label{eq:zetaG}}
     \zeta(r)=\left(\int_{\Theta}\int_{-\infty}^{-r}\mathrm{tr}(e^{-\mathbf{Q}(\theta)s}
     \Sigma_L
     e^{-\mathbf{Q}(\theta)^{\top}s})ds\pi(d\theta)\right.
     \\
     \left.+\left\|\int_{\Theta}\int_{-\infty}^{-r}e^{-\mathbf{Q}(\theta)s}\mu ds\pi(d\theta)\right\|^2
     \right)^{\frac{1}{2}},
 \end{multline*}
 and can be bounded, for all $r\geq 0$, by 
$ \zeta(r)
     \leq \widehat{\zeta}(r)$,
   where
   \begin{multline*}  \widehat{\zeta}(r)=\left(C\int_{\Theta}\|\boldsymbol{\sigma}^2_{L}\|\frac{\kappa(\mathbf{Q}(\theta))^2}{2\rho(\mathbf{Q}(\theta))} e^{-2\rho(\mathbf{Q}(\theta))r}\pi(d\theta)\right.\\
   \left.+\|\mu_{L}\| \int_{\Theta}\frac{\kappa(\mathbf{Q}(\theta))}{ \rho(\mathbf{Q}(\theta))}e^{-\rho(\mathbf{Q}(\theta))r}\pi(d\theta) \right)^{\frac{1}{2}},
    \end{multline*}
    for a constant $C>0$.

    If, in addition, $\mathbf{Q}=\mathbf{Q}(\theta)$ is diagonalisable and given by 
\eqref{eq:Qtheta2},  i.e.~$\mathbf{Q}=\mathbf{Q}(\theta_2)
    =\theta_2\mathbf{K}(c)$, then 
\begin{multline*}
\widehat{\zeta}(r)=\left(C\|\boldsymbol{\sigma}^2_{L}\|\frac{\kappa(\mathbf{K}(c))^2}{2\rho(\mathbf{K}(c))} 
\int_{0}^{\infty}\theta_2e^{-2\theta_2\rho(\mathbf{K}(c))r}\pi(d\theta_2)\right.
\\
+\left.\|\mu_L\| \frac{\kappa(\mathbf{K}(c))}{ \rho(\mathbf{K}(c))}\int_{0}^{\infty}e^{-\theta_2\rho(\mathbf{K}(c))r}\pi(d\theta_2)\right)^{\frac{1}{2}}.
\end{multline*}
In the case, when $\pi(\theta_2)$ has $\Gamma(\alpha, 1)$ law, see \eqref{eq:PD}, we  have
$\widehat{\zeta}(r)= 
    O(r^{-\alpha/2})$.

\end{theorem}

%Proof weak dependence
\begin{proof}[Proof of Theorem \ref{th:zetasupOU}]
From \cite[Corollary 3.4]{CS}, we can deduce that, for all $r\geq 0$,
\begin{align*}
\zeta(r)&=\left(\int_{M_d^-}\int_{-\infty}^{-r}\mathrm{tr}(e^{-\mathbf{Q}(\theta)s}
     \Sigma_L
     e^{-\mathbf{Q}(\theta)^{\top}s})ds\pi(d\mathbf{Q})\right.
     \\
     &\quad \left.+\left\|\int_{M_d^-}\int_{-\infty}^{-r}e^{-\mathbf{Q}(\theta)s}\mu ds\pi(d\mathbf{Q}(\theta))\right\|^2
     \right)^{\frac{1}{2}}\\
     &=
     \left(\int_{\Theta}\int_{-\infty}^{-r}\mathrm{tr}(e^{-\mathbf{Q}(\theta)s}
     \Sigma_L
     e^{-\mathbf{Q}(\theta)^{\top}s})ds\pi(d\theta)\right.\\
     &\left.+\left\|\int_{\Theta}\int_{-\infty}^{-r}e^{-\mathbf{Q}(\theta)s}\mu ds\pi(d\theta)\right\|^2
     \right)^{\frac{1}{2}}.
     \end{align*}
     For the upper bounds, we use the triangle inequality, 
the equivalence of norms, the submultiplicativity property of the Frobenius norm $||\cdot||_F$ and the properties of the graph supOU process 
to deduce that, for a constant $C>0$, using the notation $\mathbf{Q}=\mathbf{Q}(\theta)$,  
\begin{equation*} 
  \begin{split}
&\zeta^2(r)=\int_{M_d^-}\int_{-\infty}^{-r}\mathrm{tr}(e^{-\mathbf{Q}s}\boldsymbol{\sigma}^2_{L} e^{-\mathbf{Q}^\top s})ds\pi(d\mathbf{Q})+\|\int_{M_d^-}\int_{-\infty}^{-r}e^{-\mathbf{Q}s}\mu_{L} ds\pi(d\mathbf{Q})\|^2
      \\
     &\leq\int_{M_d^-}\int_{-\infty}^{-r}|\mathrm{tr}(e^{-\mathbf{Q}s}\boldsymbol{\sigma}^2_{L} e^{-\mathbf{Q}^\top s})|ds\pi(d\mathbf{Q})+\|\int_{M_d^-}\int_{-\infty}^{-r}e^{-\mathbf{Q}s}\mu_{L} ds\pi(d\mathbf{Q})\|^2
     \\
     &\leq C \int_{M_d^-}\int_{-\infty}^{-r}
     %\sqrt{gd}
     \|e^{-\mathbf{Q}s}\| \|\boldsymbol{\sigma}^2_{L}\| \|e^{-\mathbf{Q}^\top s}\|ds\pi(d\mathbf{Q})+\|\int_{M_d^-}\int_{-\infty}^{-r}e^{-\mathbf{Q}s}\mu_{L} ds\pi(d\mathbf{Q})\|\\
      &
    \leq  C\int_{M_d^-}\int_r^{\infty}
\|e^{-\mathbf{Q}s}\| \|\boldsymbol{\sigma}^2_{L}\| \|e^{-\mathbf{Q}^\top s}\|
ds\pi(d\mathbf{Q})+\|\int_{M_d^-}\int_r^{\infty}e^{\mathbf{Q}s}\mu_{L} ds\pi(d\mathbf{Q})\|
    \\
    &\leq C \int_{M_d^-}\int_{r}^{\infty}\|\boldsymbol{\sigma}^2_{L}\|\kappa(\mathbf{Q})^2 e^{-2\rho(\mathbf{Q})s}ds\pi(d\mathbf{Q})\\
    &\quad +\|\mu_{L}\| \int_{M_d^-}\int_{r}^{\infty}\kappa(\mathbf{Q}) e^{-\rho(\mathbf{Q})s}ds\pi(d\mathbf{Q})\\
     &\leq C\int_{M_d^-}\|\boldsymbol{\sigma}^2_{L}\|\frac{\kappa(\mathbf{Q})^2}{2\rho(\mathbf{Q})} e^{-2\rho(\mathbf{Q})r}\pi(d\mathbf{Q})+\|\mu_{L}\| \int_{M_d^-}\frac{\kappa(\mathbf{Q})}{ \rho(\mathbf{Q})}e^{-\rho(\mathbf{Q})r}\pi(d\mathbf{Q}).
  \end{split}       
     \end{equation*}
     Hence, $\zeta(r)\leq \widehat{\zeta}(r)$, for 
     \begin{align*}
     \widehat{\zeta}(r)=     \left(C\int_{\Theta}\|\boldsymbol{\sigma}^2_{L}\|\frac{\kappa(\mathbf{Q}(\theta))^2}{2\rho(\mathbf{Q}(\theta))} e^{-2\rho(\mathbf{Q}(\theta))r}\pi(d\theta)\right.\\
     \quad \left.+\|\mu_{L}\| \int_{\Theta}\frac{\kappa(\mathbf{Q}(\theta))}{ \rho(\mathbf{Q}(\theta))}e^{-\rho(\mathbf{Q}(\theta))r}\pi(d\theta) \right)^{\frac{1}{2}}.
     \end{align*}
     From  \cite[Remark 3.2]{Bar11} we can deduce that, for any $\mathbf{Q}(\theta)\in M_d^-$, we can find constants $\kappa\geq 1$ and $\rho \in (0, -\max(\mathrm{Re}(\sigma(\mathbf{Q}(\theta))]$ such that $\|e^{\mathbf{Q}(\theta)s}\|\leq \kappa e^{-\rho s}$ for all $s>0$. 

If $\mathbf{Q}(\theta)$ is diagonalisable, we can choose the function $\kappa(\mathbf{Q}(\theta))=\|U\|\|U^{-1}\|$ where $U \in \mathrm{GL}_d(\mathbb{C})$ is such that $U\mathbf{Q}(\theta)U^{-1}$ is diagonal and \\ $\rho(\mathbf{Q}(\theta))=-\mathrm{max}(\mathrm{Re}(\sigma(\mathbf{Q}(\theta)))$.

In case of \eqref{eq:Qtheta2}, where the drift matrix is linear in $\theta_2$, i.e.~$\mathbf{Q}(\theta_2)
    =\theta_2\mathbf{K}(c)$, for $\mathbf{K}(c)=\left(-\mathbf{I}_{d\times d}-c\bar{\textbf{A}}^{\top}\right)$, we get $\kappa(\mathbf{Q})=\theta_2\kappa(\mathbf{K}(c))$ and $\rho(\mathbf{Q})=\theta_2\rho(\mathbf{K}(c))$. 
Then, we get that
\begin{align*}
I_1&:=\int_{M_d^-}\|\boldsymbol{\sigma}^2_{L}\|\frac{\kappa(\mathbf{Q})^2}{2\rho(\mathbf{Q})} e^{-2\rho(\mathbf{Q})r}\pi(d\mathbf{Q})
\\
&=\|\boldsymbol{\sigma}^2_{L}\|\frac{\kappa(\mathbf{K}(c))^2}{2\rho(\mathbf{K}(c))} 
\int_{0}^{\infty}\theta_2e^{-2\theta_2\rho(\mathbf{K}(c))r}\pi(d\theta_2),\\
I_2&:=
\|\mu_{L}\| \int_{M_d^-}\frac{\kappa(\mathbf{Q})}{ \rho(\mathbf{Q})}e^{-\rho(\mathbf{Q})r}\pi(d\mathbf{Q})\\
&
=\|\mu_{L}\| \frac{\kappa(\mathbf{K}(c))}{ \rho(\mathbf{K}(c))}\int_{0}^{\infty}e^{-\theta_2\rho(\mathbf{K}(c))r}\pi(d\theta_2).
\end{align*}
In the case, when $\pi(\theta_2)$ has gamma law as in \eqref{eq:PD}, we get that
\begin{align*}
\int_{0}^{\infty}\theta_2e^{-2\theta_2\rho(\mathbf{K}(c))r}\pi(d\theta_2)
&=\int_{0}^{\infty}\theta_2e^{-2\theta_2\rho(\mathbf{K}(c))r}
\frac{1}{\Gamma(\alpha)}\theta_2^{\alpha-1}e^{- \theta_2}d\theta_2\\
&=
\frac{1}{\Gamma(\alpha)}
\int_{0}^{\infty}\theta_2^{\alpha}e^{-(2\rho(\mathbf{K}(c))r +1) \theta_2}d\theta_2\\
&
=\alpha (2\rho(\mathbf{K}(c))r +1)^{-(\alpha+1)}
=O(r^{-(\alpha+1)}),\\
\int_{0}^{\infty}e^{-2\theta_2\rho(\mathbf{K}(c))r}\pi(d\theta_2)
&=
\frac{1}{\Gamma(\alpha)}\int_{0}^{\infty}\theta_2^{\alpha-1}e^{-(2\rho(\mathbf{K}(c))r+1) \theta_2}d\theta_2\\
&
=(2\rho(\mathbf{K}(c))r +1)^{-\alpha}
=O(r^{-\alpha}).
\end{align*}
Hence, 
$\widehat{\zeta}(r)=
    O(r^{-\alpha/2})$.

\end{proof}
%%%%%

\subsection{Assumptions for asymptotic normality}
We want to prove the asymptotic normality of the estimator \eqref{eq:Estimator}. For this, we need additional assumptions. Our Assumptions \ref{A7}, \ref{A8}, and \ref{A9} correspond to the assumptions \cite[Assumptions 1.7-1.9]{Mat}.

%\textbf{Assumption 7}\\
\begin{assumption}\label{A7}
The moment function $f(X_{t:t+m},\xi)$ is continuously differentiable in  $\xi \in \Xi$. 
\end{assumption}
Under Assumption \ref{A7}, $f_N(\xi)$ is continuously differentiable and we write 
$$
{\bf F}_N(\xi):=\frac{\partial f_N(\xi)}{\partial \xi^{\top}}=\frac{1}{N-m}\sum_{t=1}^{N-m}\frac{\partial f_N(X_{t:t+m}, \xi)}{\partial \xi^{\top}},
$$
for its first derivative, which is a $q\times \mathrm{dim}(\xi)$-dimensional matrix, where  $\mathrm{dim}(\xi)$ denotes the dimension of the vector $\xi$.

%\textbf{Assumption 8}\\
\begin{assumption}\label{A8}
For any sequence $\xi^{\ast}_N$ satisfying $\xi^{\ast}_N\stackrel{\mathrm{P}}{\longrightarrow} \xi_0$, as $N\to \infty$, we have 
${\bf F}_N(\xi^{\ast}_N)-\overline{{\bf F}}_N\stackrel{\mathrm{P}}{\longrightarrow}0$,
as $N\to \infty$, where $\overline{{\bf F}}_N$ denotes a sequence of $q\times \mathrm{dim}(\xi)$-dimensional matrices that do not depend on the parameter $\xi$. 
\end{assumption}

%\textbf{Assumption 9}\\
\begin{assumption}\label{A9}
We have a  central limit theorem of the form 
\begin{align*}
\overline{\mathbf{V}}_N^{-1/2}\sqrt{N}f_N(\xi_0)  \stackrel{\mathrm{d}}{\longrightarrow} \mathrm{N}(0, \mathbf{I}_{q\times q}),
\end{align*}
where $\overline{\mathbf{V}}_N=N \mathrm{var}(f_N(\xi_0))$ denotes a sequence of $q\times q$-dimensional, deterministic, positive definite matrices. 
\end{assumption}
In order to show that 
Assumption \ref{A9} is verified, we present  Theorem \ref{th:weakdepempi}, where  we prove the central limit theorem for the moment function $f(X_{t:t+m},\xi)$,  extending  \cite[Theorem 6.1]{CS} from univariate supOU processes to graph (multivariate) supOU processes.
\begin{theorem}\label{th:weakdepempi}
Suppose that the graph supOU process satisfies the conditions of 
Theorem \ref{th:supOU}; moreover, 
assume that $\int_{\|x\|>1}\|x\|^{4+\delta}\nu(dx)<\infty$ for some $\delta>0$.
Suppose that the upper bound, see Theorem \ref{th:zetasupOU},  of the $\zeta$-weakly dependent coefficient of graph supOU process is given by $\widehat{\zeta}(r)=O(r^{-\epsilon})$, for $\epsilon >(1+\frac{1}{\delta})\frac{3+\delta}{2+\delta}$.  
Then $(f(X_{t:t+m},\xi_0))_{t\in \mathbb{Z}}$, defined in \eqref{eq:f}, is a $\zeta$-weakly dependent process; the matrix 
\begin{align}
\label{eq:Fsigma}
\mathbf{F}_{\Sigma}=\sum_{l\in\mathbb{Z}}\mathrm{cov}(f(X_{0:m},\xi_0),f(X_{l:l+m},\xi_0)),
\end{align}
is finite and positive semidefinite, and, as $N\rightarrow \infty$,
\begin{equation*}%{\label{eq:CLTmf}}
    \sqrt{N}f_{N}(\xi_0)\xrightarrow{d} N(0,\mathbf{F}_{\Sigma}).
\end{equation*}
\end{theorem}

%%%%%%%%%%%%%%%%%%%%%%%%%
%Proof of CLT for moment function
\begin{proof}[Proof of Theorem \ref{th:weakdepempi}]
 Since a graph supOU process is a special case of a multivariate mixed moving average process and given $\int_{\|x\|>1}\|x\|^{4+\delta}\nu(dx)<\infty$, % and \eqref{eq:expinteg},
 we deduce from  \cite[Proposition 2.1]{CS}, that the $4+\delta$ moments of the graph supOU process exist.

Define a function $H:\mathbb{R}^{d(m+1)}\rightarrow \mathbb{R}^{q}$ such that
\begin{align*}
H(X_{t:t+m})=f(X_{t:t+m},\xi_0)+\left(\begin{matrix}
      \text{vec}     (  \mu(\xi_0))\\
      \text{vec}_h      (\mathbf{D}_0(\xi_0))\\
    \text{vec}_h         ( \mathbf{D}_1(\xi_0))\\
             \vdots\\
    \text{vec}_h          (\mathbf{D}_m(\xi_0))\\
         \end{matrix}\right)=\left(\begin{matrix}
      \text{vec}     (  X_{t\Delta})\\
      \text{vec}_h      ( X_{t\Delta} X_{t\Delta}^{\top})\\
    \text{vec}_h         ( X_{t} X_{t+1}^{\top})\\
             \vdots\\
    \text{vec}_h          (X_{t} X_{t+m}^{\top})\\
         \end{matrix}\right).
         \end{align*}
This function $H$ satisfies the conditions of \cite[Proposition 3.4]{CS}, %\ref{Prop3.4}, 
for $p=4+\delta, c=1, a=2$. Hence, $H(X_{t:t+m})$ is a $\zeta$-weakly dependent process with coefficient %$\zeta_F(r)=
$\mathcal{C}(\mathcal{D}\zeta(r-m\Delta))^{\frac{2+\delta}{3+\delta}}$ for constants $\mathcal{C},\mathcal{D}>0$ (independent of $r$).
This implies that $f(X_{t:t+m},\xi_0)$ is a $\zeta$-weakly dependent process with  coefficient
$\mathcal{C}(\mathcal{D}\zeta(r-m\Delta))^{\frac{2+\delta}{3+\delta}}$
and zero mean. 
    We would like to apply  \cite[Theorem 1]{DR}.
That theorem requires a moment condition which is satisfied as soon as $\widehat{\zeta}(r)=O(r^{-\epsilon})$, for $\epsilon\left(\frac{2+\delta}{3+\delta}\right) >1+1/\delta$, see \cite{CSS2022}.
  Applying  \cite[Theorem 1]{DR} together with the
    Cramér Wold device  concludes the proof. 
    \end{proof}

\subsection{General asymptotic normality result}
 For the  asymptotic normality result for the general graph supOU process, we need two additional assumptions:
\begin{assumption}\label{as:A_weightmatrix}
The weight matrix $\mathbf{W}_N$ converges in probability to a positive definite matrix $\mathbf{W}$ of constants.
\end{assumption}

\begin{assumption}\label{as:asymcov}
The matrix $\mathbf{F}_{\Sigma}$ defined in \eqref{eq:Fsigma}
is positive definite.
\end{assumption}

\begin{theorem}\label{th:AN}
Suppose that the graph supOU process satisfies the conditions of 
Theorem \ref{th:supOU}; moreover, 
assume that $\int_{\|x\|>1}\|x\|^{4+\delta}\nu(dx)<\infty$ for some $\delta>0$.
Suppose that the upper bound, see Theorem \ref{th:zetasupOU},  of the $\zeta$-weakly dependent coefficient of graph supOU process is given by $\widehat{\zeta}(r)=O(r^{-\epsilon})$, for $\epsilon >(1+\frac{1}{\delta})\frac{3+\delta}{2+\delta}$. 
Suppose that Assumptions \ref{A1}, \ref{A2}, \ref{A3}, \ref{as:A_weightmatrix}, and \ref{as:asymcov} hold.
%Moreover, Assumptions 6-9 hold. 
Then as $N\rightarrow \infty$,
\begin{equation*}%{\label{CLTe}}
    \sqrt{N}(\hat{\xi}_0^{N}-\xi_0)\xrightarrow{d} N(0,\mathbf{M}\mathbf{F}_{\Sigma}\mathbf{M}^{\top}),
\end{equation*}    
where 
\begin{align*}
\mathbf{F}_{\Sigma}&=\sum_{l\in\mathbb{Z}}\mathrm{cov}(f(X_{0:m},\xi_0),f(X_{l:l+m},\xi_0)),\\
\overline{{\bf F}}&=\frac{\partial f(X_{t:t+m},\xi)}{\partial \xi^{\top}}
|_{\xi=\xi_0}, %G_0
\\
\mathbf{M}&=(\overline{{\bf F}}^{\top}\mathbf{W}\overline{{\bf F}})^{-1}\overline{{\bf F}}^{\top}\mathbf{W}.
\end{align*}
\end{theorem}
 %%%%%
    %Proof of Asymptotic normality
    \begin{proof}[Proof of Theorem \ref{th:AN}]
    The asymptotic normality results for the GMM estimators, see \cite[Theorem 1.2]{Mat}, was proven under the Assumptions
 \ref{A1}-\ref{A3} and \ref{A7}-\ref{A9}. In our case, we note that 
 Assumption \ref{A7} holds by the  construction of the function $f$.
 Moreover, we  verify that  Assumption \ref{A8} is satisfied for our graph supOU process: By the construction of the moment function, the partial derivative matrix $\frac{\partial f(X_{t:t+m},\xi)}{\partial \xi^T}$ is independent of $X_{t:t+m}$.
Hence, we have
${\bf F}_N(\xi)=\mathrm{E}(\frac{\partial f(X_{t:t+m},\xi)}{\partial \xi^T})=\frac{\partial f(X_{t:t+m},\xi)}{\partial \xi^{\top}}$ and can set  \begin{align*}%\label{eq:Fbar}
\overline{{\bf F}}:=\overline{{\bf F}}_N:=\frac{\partial f(X_{t:t+m},\xi)}{\partial \xi^{\top}}
|_{\xi=\xi_0}.
\end{align*}
An application of  the continuous mapping theorem leads the property stated in Assumption \ref{A8}.
Theorem \ref{th:weakdepempi} implies that Assumption \ref{A9} holds.  Together with Assumptions \ref{as:A_weightmatrix} and \ref{as:asymcov}, the asymptotic normality stated in the theorem follows. 
\end{proof}

The proof of the asymptotic normality in the special case of a weighted-exponential or Gamma model follows directly:
\begin{proof}[Proof of Theorem \ref{th:AN-main}]
In the case of a graph supOU process with Gamma law for $\theta_2$ (as in Theorem \ref{th:consistency-gamma} and the discussion in the corresponding proof), 
Assumptions \ref{A1} and \ref{A2} are satisfied and we only need to require the Assumptions 
% \ref{A3} and %\ref{A4} in Theorem \ref{th:AN}.
 \ref{A3}, \ref{as:A_weightmatrix}, and \ref{as:asymcov}.
 \end{proof}

\end{appendix}

\subsection*{Acknowledgment} SM acknowledges funding by the   
Schrödinger-Roth Scholarship awarded by the Faculty of Natural Sciences, Imperial College
London. 
AEDV's work has been supported through the 
EPSRC NeST Programme grant EP/X002195/1. 
%\bibliography{References}
%\bibliographystyle{agsm}

\end{document}